\documentclass[11pt,a4paper]{article}

\usepackage[utf8]{inputenc}
\usepackage[T1]{fontenc}
\usepackage[left=2.5cm,right=2.5cm,top=3cm,bottom=3cm]{geometry}
\usepackage[toc,page]{appendix}
\usepackage{amsmath,amsfonts,amssymb,amsthm,graphicx,dsfont,color,stmaryrd,subcaption}
\usepackage{enumitem}
\usepackage{hyperref}
\usepackage{float}
\usepackage[dvipsnames]{xcolor}
\usepackage{pifont}
\usepackage[normalem]{ulem}
\usepackage{wrapfig}
\usepackage{mathabx}
\usepackage{array, booktabs}
\usepackage{diagbox}

\usepackage[english]{babel}

\DeclareMathOperator{\e}{e}

\DeclareMathOperator{\Proj}{\mathrm{Proj}}

\DeclareMathOperator{\A}{\mathcal{A}}
\DeclareMathOperator{\F}{\mathcal{F}}
\DeclareMathOperator{\N}{\mathcal{N}}

\DeclareMathOperator{\R}{\mathds{R}}

\DeclareMathOperator{\V}{\mathds{V}\mathrm{ar}}
\DeclareMathOperator{\E}{\mathds{E}}
\DeclareMathOperator{\Cov}{\mathds{C}\mathrm{ov}}
\DeclareMathOperator{\Corr}{\mathds{C}\mathrm{orr}}
\DeclareMathOperator{\Prob}{\mathds{P}}
\DeclareMathOperator{\Q}{\mathds{Q}}
\DeclareMathOperator{\Law}{\mathcal{L}}

\DeclareMathOperator{\1}{\mathds{1}}

\DeclareMathOperator\supp{supp}

\DeclareMathOperator{\Distortion}{\mathcal{Q}_{2,N}}

\newcommand\indep{\protect\mathpalette{\protect\indeT}{\perp}}
\def\indeT#1#2{\mathrel{\rlap{$#1#2$}\mkern2mu{#1#2}}}

\newcommand{\vertiii}[1]{{\left\vert\kern-0.25ex\left\vert\kern-0.25ex\left\vert #1
\right\vert\kern-0.25ex\right\vert\kern-0.25ex\right\vert}}

\theoremstyle{plain}
\newtheorem{theorem}{Theorem}[section]
\newtheorem{lemme}[theorem]{Lemma}
\newtheorem{proposition}[theorem]{Proposition}

\theoremstyle{definition}
\newtheorem{definition}[theorem]{Definition}

\newtheorem{remark}[theorem]{Remark}
\newtheorem{remarks}[theorem]{Remarks}

\numberwithin{equation}{section}

\renewenvironment{abstract}
 {\small
    \begin{center}
        \bfseries \abstractname\vspace{-.5em}\vspace{0pt}
    \end{center}
    \list{}{%
    \setlength{\leftmargin}{20mm}
    \setlength{\rightmargin}{\leftmargin}%
 }%
 \item\relax}
 {\endlist}

\providecommand{\keywords}[1]
{
  {
  \small
  \textbf{\textit{Keywords---}} #1
  }
}

\newcommand*\samethanks[1][\value{footnote}]{\footnotemark[#1]}

\begin{document}

\graphicspath{{graphics/}}

\title{\textbf{Quantization-based Bermudan option pricing in the $FX$ world}}
\author{
	\sc Jean-Michel Fayolle
	\thanks{The Independent Calculation Agent, ICA, 112 Avenue Kleber, 75116 Paris, France.}
	\and
	\sc Vincent Lemaire
	\thanks{Sorbonne Université, Laboratoire de Probabilités, Statistique et Modélisation, LPSM, Campus Pierre et Marie Curie, case 158, 4 place Jussieu, F-75252 Paris Cedex 5, France.}
	\and
	\sc Thibaut Montes
	\samethanks[1]
	\samethanks[2]
	\and
	\sc Gilles Pagès
	\samethanks[2]
}
\maketitle


\begin{abstract}
	This paper proposes two numerical solution based on Product Optimal Quantization for the pricing of Foreign Exchange (FX) linked long term Bermudan options e.g. Bermudan Power Reverse Dual Currency options, where we take into account stochastic domestic and foreign interest rates on top of stochastic FX rate, hence we consider a 3-factor model. For these two numerical methods, we give an estimation of the $L^2$-error induced by such approximations and we illustrate them with market-based examples that highlight the speed of such methods.
\end{abstract}

\keywords{Foreign Exchange rates; Bermudan Options; Numerical method; Power Reverse Dual Currency; Product Optimal Quantization.}
~ \\



\section*{Introduction}

Persistent low levels of interest rates in Japan in the latter decades of the 20th century were one of the core sources that led to the creation of structured financial products responding to the need of investors for coupons higher than the low yen-based ones. This started with relatively simple dual currency notes in the 80s where coupons were linked to foreign (i.e. non yen-based) currencies enabling payments of coupons significantly higher. As time (and issuers’ competition) went by, such structured notes were iteratively “enhanced” to reverse dual currency, power reverse dual currency (PRDC), cancellable power reverse dual currency etc., each version adding further features such as limits, early repayment options, etc. Finally, in the early 2000s, the denomination xPRD took root to describe those structured notes typically long-dated (over 30y initial term) and based on multiple currencies (see \cite{wystup2017fx}). The total notional invested in such notes is likely to be in the hundreds of billions of USD.
The valuation of such investments obviously requires the modeling of the main components driving the key risks, namely the interest rates of each pair of currencies involved as well as the corresponding exchange rates. In its simplest and most popular version, that means 3 sources of risk: domestic and foreign rates and the exchange rate. The 3-factor model discussed herein is an answer to that problem.

Gradually, as the note's features became more and more complex, further refinements to the modeling were needed, for instance requiring the inclusion of the volatility smile, the dependence of implied volatilities on both the expiry and the strike\footnote{In the case of the FX, the implied volatility is expressed in function of the delta.} of the option, prevalent in the $FX$ options market. Such more complete modeling should ideally consist in successive refinements of the initial modeling enabling consistency across the various flavors of xPRDs at stake.

The model discussed herein was one of the answers popular amongst practitioners for multiple reasons: it was accounting for the main risks -- interest rates in the currencies involved and exchange rates -- in a relatively simple manner and the numerical implementations proposed at that time were based on simple extensions of well-known single dimensional techniques such as 3 dimensional trinomial trees, PDE based method (see \cite{piterbarg2005multi}) or on Monte Carlo simulations.

Despite the qualities of these methods, the calculation time could be rather slow (around 20 minutes with a trinomial tree for one price), especially when factoring in the cost for hedging (that is, measuring the sensitivities to all the input parameters) and even more post 2008, where the computation of risk measures and their sensitivities to market values became a central challenge for the financial markets participants. Indeed, even though these products were issued towards the end of the 20th century, they are still present in the banks's books and need to be considered when evaluating counterparty risk computations such as Credit Valuation Adjustment (CVA), Debt Valuation Adjustment (DVA), Funding Valuation Adjustment (FVA), Capital Valuation Adjustment (KVA), ..., in short xVA's (see \cite{brigo2013counterparty,crepey2014counterparty,gregory2015xva} for more details on the subject). Hence, a fast and accurate numerical method is important for being able to produce the correct values in a timely manner. The present paper aims at providing an elegant and efficient answer to that problem of numerical efficiency based on Optimal Quantization. Our novel method allows us reach a computation time of 1 or 2 seconds at the expense of a systematic error that we quantify in Section \ref{3F:section:bermudan_evaluation_quantization}.

\medskip

Let $P(t, T)$ be the value at time $t$ of one unit of the currency delivered (that is, paid) at time $T$, also known as a zero coupon price or discount factor. A few iterations were needed by researchers and practitioners before the seminal family of Heath-Jarrow-Morton models came about. The general Heath-Jarrow-Morton (HJM) family of yield curve models can be expressed as follows -- although originally expressed by its authors in terms of rates dynamics, the two are equivalent, see \cite{heath1992bond} -- in a $n$-factor setting, we have for the curve $P(t, T)$ that
\begin{equation}
	\frac{dP(t, T)}{P(t, T)} = r_t dt + \sum_i \sigma_i \big( t, T, P(t, T) \big)dW^i_t
\end{equation}
where $r_t$ is the instantaneous rate at time $t$ (therefore a random variable), $W^i, \, i = 1, \cdots, n$ are $n$ correlated Brownian motions and $\sigma_i \big( t, T, P(t, T) \big)$ are volatility functions in the most general settings (with the obvious constraint that $\sigma_i \big( T, T, P(T, T) \big) = 0$). Indeed, the general HJM framework allows for the volatility functions $\sigma_i \big( t, T, P(t, T) \big)$ to also depend on the yield curve’s (random) levels up to $t$ -- actually through forward rates -- and therefore be random too. However, it has been demonstrated in \cite{el1992arbitrage} that, to keep a tractable version (i.e. a finite number of state variables), the volatility functions must be of a specific form, namely, of the mean-reverting type (where the mean reversion can also depend on time). We use this way of expressing the model as a mean to recall that such model is essentially the usual and well-known Black Scholes model applied to all and any zero-coupon prices, with various enhancements regarding number of factors and volatility functions, to keep the calculations tractable. For further details and theory, one can refer to some of the following articles \cite{elkaroui1996noteonbehavior,el1992arbitrage,heath1992bond,black1973pricing}.
Of course, such a framework can be applied to any yield curve. In its simplest form (i.e. flat volatility and one-factor), we have under the risk-neutral measure
\begin{equation}\label{3F:GYC}
	\frac{dP(t,T)}{P(t,T)} = r_t dt + \sigma (T-t) dW_t
\end{equation}
where $W$ is a standard Brownian motion under the risk-neutral probability. In that case, $\sigma$ is the flat volatility, which means the volatility of (zero-coupon) interest rates. That is often referred to as a Hull-White model without mean reversion (see \cite{hull1993one}) or a continuous-time version of the Ho-Lee model. In the rest of the paper, we work with the model presented in \eqref{3F:GYC} for the diffusion of the zero coupon although the extension to non-flat volatilities is easily feasible.

About the Foreign Exchange ($FX$) rate, we denote by $S_t$ the value at time $t>0$ of one unit of foreign currency in the domestic one. The diffusion is that of a standard Black-Scholes model with the following equation
\begin{equation}
	\frac{ dS_t }{ S_t } = ( r^d_t - r^f_t ) dt + \sigma_S dW^S_t
\end{equation}
where $r^d_t$ is the instantaneous rate of the domestic currency at time $t$, $r^f_t$ is the instantaneous rate of the foreign currency at time $t$, $\sigma_S$ is the volatility of the $FX$ rate and $W^S$ is a standard Brownian motion under the risk-neutral probability.

\medskip

Let us briefly recall the principle of the adopted numerical method, Optimal quantization. Optimal Quantization is a numerical method whose aim is to approximate optimally, for a given norm, a continuous random signal by a discrete one with a given cardinality at most $N$. \cite{sheppard1897calculation} was the first to work on it for the uniform distribution on unit hypercubes. Since then, it has been extended to more general distributions with applications to Signal transmission in the 50's at the Bell Laboratory (see \cite{gersho1982special}). Formally, let $Z$ be an $\R^d$-valued random vector with distribution $\Prob_{_{Z}}$ defined on a probability space $( \Omega, \A, \Prob )$ such that $Z \in L^2(\Prob)$. We search for $\Gamma_N$, a finite subset of $\R^d$ defined by $\Gamma_N := \{ z_1^N, \dots, z_N^N \} \subset \R^d$, solution to the following problem
\begin{equation*}
	\min_{\Gamma_N \subset \R^d, \vert \Gamma_N \vert \leq N } \Vert Z - \widehat Z^N \Vert_{_2}
\end{equation*}
where $\widehat Z^N$ denotes the nearest neighbour projection of $Z$ onto $ \Gamma_N $. This problem can be extended to the $L^p$-optimal quantization by replacing the $L^2$-norm by the $L^p$-norm but this not in the scope of this paper. In our case, we mostly consider quadratic one-dimensional optimal quantization, i.e $d=1$ and $p=2$. The existence of an optimal quantizer at level $N$ goes back to \cite{cuesta1997existence} (see also \cite{pages1998space,graf2000foundations} for further developments). In the one-dimensional case, if the distribution of $Z$ is absolutely continuous with a \textit{log-concave} density, then there exists a unique optimal quantizer at level $N$, see \cite{kieffer1983uniqueness}. We scale to the higher dimension using Optimal Product Quantization which deals with multi-dimensional quantizers built by considering the cartesian product of one-dimensional optimal quantizers.

Considering again $Z = (Z^{\ell})_{\ell=1:d}$, a $\R^d$-valued random vector. First, we look separately at each component $Z^{\ell}$ independently by building a one-dimensional optimal quantization $\widehat Z^{\ell}$ of size $N^{\ell}$, with quantizer $ \Gamma_{\ell}^{N_{\ell}} = \big\{ z_{i_{\ell}}^{\ell}, i_{\ell} \in \{ 1, \cdots, N_{\ell} \} \big\} $ and then, by applying the cartesian product between the one-dimensional optimal quantizers, we build the product quantizer $\Gamma^{N}= \prod_{\ell=1}^d \Gamma_{\ell}^{N_{\ell}}$ with cardinality $N = N^{1} \times \cdots \times N^{d}$ by
\begin{equation}
	\Gamma^{N} = \big\{ ( z_{i_1}^{1}, \cdots, z_{i_{\ell}}^{\ell}, \cdots, z_{i_d}^{d} ),\quad i_{\ell} \in \{ 1, \cdots, N_{\ell} \},\quad \ell \in \{ 1, \cdots, d \} \big\}.
\end{equation}

Then, in the 90s, \cite{pages1998space} developed quantization-based cubature formulas for numerical integration purposes and expectation approximations. Indeed, let $f$ be a continuous function $f:\R^d \longrightarrow \R$ such that $f(Z) \in L^1 ( \Prob )$, we can define the following quantization-based cubature formula using the discrete property of the quantizer $\widehat Z^N$
\begin{equation*}
	\E \big[ f(\widehat Z^N) \big] = \sum_{i=1}^{N} p_i f(z_i^N)
\end{equation*}
where $p_i = \Prob ( \widehat Z^N = z_i^N)$. Then, one could want to approximate $\E \big[ f(Z) \big]$ by $\E \big[ f(\widehat Z^N) \big]$ when the first expression cannot be computed easily. For example, this case is exactly the problem one encounters when trying to price European options. We know the rate of convergence of the weak error induced by this cubature formula, i.e $\exists \alpha \in (0,2]$, depending on the regularity of $f$ such that
\begin{equation}
	\lim_{N \rightarrow + \infty} N^{\alpha} \big\vert \E \big[ f(Z) \big] - \E \big[ f ( \widehat Z^N ) \big] \big\vert \leq C_{f, X} < + \infty.
\end{equation}
For more results on the rate of convergence, the value of $\alpha$, we refer to \cite{pages2018numerical} for a survey in $\R^d$ and to \cite{lemaire2019weakerror} for recent improved results in the one-dimensional case.

Later on, in a series of papers, among them \cite{bally2003quantization,printems2005quantization} extended this method to the computation of conditional expectations allowing to deal with nonlinear problems in finance and, more precisely, to the pricing and hedging of American/Bermudan options, which is the part we are interested in. These problems are of the form
\begin{equation*}
	\sup_{\tau} \E \big[ \e^{- \int_{0}^{\tau} r_s^d ds } \psi_{\tau} ( S_{\tau} ) \big]
\end{equation*}
where $\big( \e^{- \int_{0}^{t_k} r_s^d ds } \psi_{t_k} ( S_{t_k} ) \big)_{k=0, \dots, n}$ is the obstacle function and $\tau : \Omega \rightarrow \{ t_0, t_1, \dots, t_n \}$ is a stopping time for the filtration $(\F_{t_k})_{k \geq 0}$ where $\F_{t} = \sigma \big( S_s, P^d(s,T), P^f(s,T), s \leq t \big)$ is the natural filtration to consider because the foreign exchange rate and the zero-coupon curves are observables in the market.

\bigskip

In this paper, we will present two numerical solutions, motivated by the works described above, to the problem of the evaluation of Bermudan option on Foreign Exchange rate with stochastic interest rates.
The paper is organised as follows. First, in Section \ref{3F:section:diffusion_models}, we introduce the diffusion models for the zero coupon curves and the foreign exchange rate we work with. In Section \ref{3F:section:bermudan_options}, we describe in details the financial product we want to evaluate: Bermudan option on foreign exchange rate. In this Section, we express the \textit{Backward Dynamic Programming Principle} and study the regularity of the obstacle process and the value function. Then, in Section \ref{3F:section:bermudan_evaluation_quantization}, we propose two numerical solutions for pricing the financial product defined above based on Product Quantization and we study the $L^2$-error induced by these numerical approximations. In Section \ref{3F:section:numerics}, several examples are presented in order to compare the two methods presented in Section \ref{3F:section:bermudan_evaluation_quantization}. First, we begin with plain European option, this test is carried out in order to benchmark the methods because a closed-form formula is known for the price of a European Call/Put in the 3-factor model. Then, we compare the two methods in the case of a Bermudan option with several exercise dates. Finally, in Appendix \ref{3F:section:proof_foreignBM}, we make some change of numéraire and in Appendix \ref{3F:section:closed_form_EU_price}, we give the closed-form formula for the price of an European Call, in the 3-factor model, used in Section \ref{3F:section:numerics} as a benchmark.

\section{Diffusion Models} \label{3F:section:diffusion_models}

\paragraph{Interest Rate Model.}

We shall denote by $P(t, T)$ the value at time $t$ of one unit of the currency delivered (that is, paid) at time $T$, also known as a zero coupon price or discount factor. When $t$ is today, this function can usually be derived from the market price of standard products, such as bonds and interest rate swaps in the market, along with an interpolation scheme (for the dates different than the maturities of the market rates used). In a simple single-curve framework, the derivation of the initial curve, that is, the zero coupons $P(0, T)$ for $T > 0$ is rather simple, through relatively standard methods of curve stripping. In more enhanced frameworks accounting for multiple yield curves such as having different for curves for discounting and forward rates, those methods are somewhat more demanding but still relatively straightforward. We focus herein on the simple single-curve framework.

In our case we are working with financial products on Foreign Exchange ($FX$) rates between the domestic and the foreign currency, hence we will be working with zero coupons in the domestic currency denoted by $P^d(t, T)$ and zero coupons in the foreign currency denoted by $P^f(t, T)$. The diffusion of the domestic zero-coupon curve under the domestic risk-neutral probability $\Prob$ is given by
\begin{equation*}
	\frac{ d P^d(t,T) }{ P^d(t,T) } = r_t^d dt + \sigma_d (T-t) d W_t^d
\end{equation*}
where $W^d$ is a $\Prob$-Brownian Motion, $r_t^d$ is the domestic instantaneous rate at time $t$ and $\sigma_d$ is the volatility for the domestic zero coupon curve. For the foreign zero-coupon curve, the diffusion is given, under the foreign neutral probability $\widetilde \Prob$, by
\begin{equation*}
	\frac{ d P^f(t,T) }{ P^f(t,T) } = r_t^f dt + \sigma_f (T-t) d \widetilde{W}_t^f
\end{equation*}
where $\widetilde{W}^f$ is a $\widetilde \Prob$-Brownian Motion, $r_t^f$ is the foreign instantaneous rate at time $t$ and $\sigma_f$ is the volatility for the foreign zero coupon curve. The two probabilities $\widetilde \Prob$ and $\Prob$ are supposed to be equivalent, i.e $\widetilde \Prob \sim \Prob$ and it exists $\rho_{df}$ defined as limit of the quadratic variation $\langle W^d, \widetilde{W}^f \rangle_t = \rho_{df} t$.

\begin{remarks}
	Such a framework to model random yield curves has been quite popular with practitioners due to its elegance, simplicity and intuitive understanding of rates dynamics through time yet providing a comprehensive and consistent modelling of an entire yield curve through time. Indeed, it is mathematically and numerically easily tractable. It carries no path dependency and allows the handling of multiple curves for a given currency as well as multiple currencies -- and their exchange rates -- as well as equities (when one wishes to account for random interest rates). It allows negative rates and can be refined by adding factors (Brownian motions).

	However, it cannot easily cope with smile or non-normally distributed shocks or with internal curve ”oddities” or specifics such as different volatilities for different swap tenors within the same curve dynamics. Nonetheless, our aim being to propose a model and a numerical method which make possible to produce risk computations (such as xVA's) in an efficient way, these properties are of little importance. That said, when it comes to deal with accounting for random rates in long-dated derivatives valuations, its benefits far outweigh its limitations and its use for such applications is popular, see \cite{nunes2014pricingswaption} for the pricing of swaptions, \cite{piterbarg2005multi} for PRDCs...
\end{remarks}

\paragraph{Foreign Exchange Model.}

The diffusion of the foreign exchange ($FX$) rate defined under the domestic risk-neutral probability is
\begin{equation*}
	\frac{ dS_t }{ S_t } = ( r^d_t - r^f_t ) dt + \sigma_S dW^S_t
\end{equation*}
with $W_t^S$ a $\Prob$-Brownian Motion under the domestic risk-neutral probability such that their exist $\rho_{Sd}$ and $\rho_{Sf}$ defined as limit of the quadratic variation $\langle W^S, W^d \rangle_t = \rho_{Sd} t$ and $\langle W^S, \widetilde{W}^f \rangle_t = \rho_{Sf} t$, respectively.

Finally, the processes, expressed in the domestic risk-neutral probability $\Prob$, are
\begin{equation} \label{3F:3FModelFX} \left\{
	\begin{aligned}
		\frac{ d P^d(t,T) }{P^d(t,T) } & = r_t^d dt + \sigma_d (T-t) dW_t^d                                                 \\
		\frac{ d S_t }{S_t}            & = ( r^d_t - r^f_t ) dt + \sigma_S dW^S_t                                           \\
		\frac{ d P^f(t,T) }{P^f(t,T) } & = \big( r_t^f - \rho_{Sf} \sigma_S \sigma_f (T-t) \big) dt + \sigma_f (T-t) dW_t^f
	\end{aligned} \right.
\end{equation}
where $W^f$, defined by $d W^f_s = d \widetilde W^f_s + \rho_{Sf} \sigma_S ds$, is a $\Prob$-Brownian motion, as shown in Appendix \ref{3F:section:proof_foreignBM}.
Using Itô's formula, we can explicitly express the processes
\begin{equation*} \left\{
	\begin{aligned}
		P^d(t,T) & = P^d(0,T) \exp \Bigg( \int_0^t \bigg(r_s^d - \frac{\sigma_d^2 (T-s)^2}{2} \bigg) ds + \sigma_d \int_0^t (T-s) dW_s^d \Bigg)                                     \\
		S_t      & = S_{0} \exp \Bigg( \int_0^t \bigg( r^d_s - r^f_s - \frac{\sigma^2_S}{2} \bigg) ds + \sigma_S W_t^S \Bigg)                                                       \\
		P^f(t,T) & = P^f(0,T) \exp \Bigg( \int_0^t \bigg(r_s^f - \rho_{Sf} \sigma_S \sigma_f (T-s) - \frac{\sigma_f^2 (T-s)^2}{2} \bigg) ds + \sigma_f \int_0^t (T-s) dW_s^f \Bigg)
	\end{aligned} \right. .
\end{equation*}
From these equations, we deduce $\exp \bigg(- \int_0^t r_s^d ds \bigg)$ and $\exp \bigg( - \int_0^t r_s^f ds \bigg)$, by taking $T=t$ and using that $P^d(t,t) = P^f(t,t) = 1$, it follows that
\begin{equation*}\left\{
	\begin{aligned}
		\exp \bigg( - \int_0^t r_s^d ds \bigg) & = \varphi_d(t) \exp \bigg( \sigma_d \int_0^t (t-s) dW_s^d \bigg)  \\
		\exp \bigg( - \int_0^t r_s^f ds \bigg) & = \varphi_f(t) \exp \bigg( \sigma_f \int_0^t (t-s) dW_s^f \bigg),
	\end{aligned}\right.
\end{equation*}
where
\begin{equation}
	\varphi_d(t) = P^d(0,t) \exp \bigg( - \sigma_d^2 \int_0^t \frac{(t-s)^2}{2} ds \bigg)
\end{equation}
and
\begin{equation}
	\varphi_f(t) = P^f(0,t) \exp \Bigg( - \int_0^t \bigg(\rho_{Sf} \sigma_S \sigma_f (t-s) + \frac{\sigma_f^2 (t-s)^2}{2} \bigg) ds \Bigg).
\end{equation}
These expressions for the domestic and the foreign discount factors will be useful in the following sections of the paper.

\section{Bermudan options} \label{3F:section:bermudan_options}

\subsection{Product Description}

Let $(\Omega, \A, \Prob)$ our domestic risk neutral probability space. We want to evaluate the price of a Bermudan option on the $FX$ rate $S_t$ defined by
\begin{equation*}
	S_t = \frac{1}{\exp \bigg( - \int_0^t r_s^d ds \bigg) } S_{0} \varphi_f(t) \exp \bigg( - \frac{\sigma^2_S}{2} t + \sigma_S W^S_t + \sigma_f \int_0^t (t-s) dW_s^f \bigg)
\end{equation*}
with
\begin{equation*}
	\exp \bigg( - \int_0^t r_s^d ds \bigg) = \varphi_d(t) \exp \bigg( \sigma_d \int_0^t (t-s) dW_s^d \bigg)
\end{equation*}
where the owner of the financial product can exercise its option at predetermined dates $t_0, t_1, \cdots, t_n$ with payoff $\psi_{t_k}$ at date $t_k$, where $t_0 = 0$.

At a given time $t$, the observables in the market are the foreign exchange rate $S_t$ and the zero-coupon curves $\big( P^d(t,T) \big)_{T \geq t}$ and $\big( P^f(t,T) \big)_{T \geq t}$, hence the natural filtration to consider is
\begin{equation}
	\F_{t} = \sigma \big( S_s, P^d(s,T), P^f(s,T), s \leq t \big)  = \sigma \big( W^S_s, W^d_s, W^f_s, s \leq t \big).
\end{equation}

Let $\tau : \Omega \rightarrow \{ t_0, t_1, \dots, t_n \}$ a stopping time for the filtration $(\F_{t_k})_{k \geq 0}$ and $\mathcal{T}$ the set of all stopping times for the filtration $(\F_{t_k})_{k \geq 0}$. In this paper, we consider problems where the horizon is finite then we define $\mathcal{T}_k^n$, the set of all stopping times taking finite values
\begin{equation}
	\mathcal{T}_k^n = \big\{ \tau \in \mathcal{T}, \Prob ( t_k \leq \tau \leq t_n ) = 1 \big\}.
\end{equation}

Hence, the price at time $t_k$ of the Bermudan option is given by
\begin{equation*}
	V_k = \sup_{\tau \in \mathcal{T}_k^n } \E \big[ \e^{- \int_{0}^{\tau} r_s^d ds } \psi_{\tau} ( S_{\tau} ) \mid \F_{t_k} \big]
\end{equation*}
and $V_k$ is called the \textit{Snell envelope} of the obstacle process $\big( \e^{- \int_{0}^{t_k} r_s^d ds } \psi_{t_k} ( S_{t_k} ) \big)_{k=0:n}$ such that
\begin{equation}
	\E \big[ \psi_{t_k} ( S_{t_k} )^2 \big] < + \infty, \quad \forall k = 0, \dots, n.
\end{equation}

\begin{remark}
	The financial products we consider in the applications are PRDC. Their payoffs (see Figure \ref{3F:fig:prdc}) have the following expression
	\begin{equation}\label{3F:payoffPRDC}
		\psi_{t_k} ( x ) = \min \Bigg( \max \bigg( \dfrac{C_f(t_k)}{S_{0}} x - C_d(t_k), \textrm{Floor}(t_k) \bigg) , \textrm{Cap}(t_k) \Bigg)
	\end{equation}
	where $\textrm{Floor}(t_k)$ and $\textrm{Cap}(t_k)$ are the floor and cap values chosen at the creation of the product, as well as $C_f(t_k)$ and $C_d(t_k)$ that are the coupons value we wish to compare to the foreign and the domestic currency, respectively.
	\begin{figure}[H]
		\centering
		\includegraphics[width=0.47\textwidth]{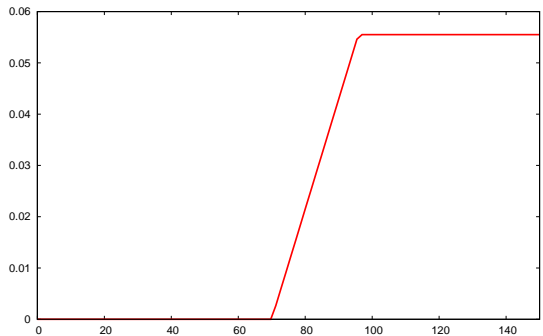}
		\caption[Example of a PRDC payoff]{\textit{Example of a PRDC payoff $\psi_{t_k} ( S_{t_k} ) =\min \Big(  \Big(0.189 \frac{S_{t_k}}{88.17} -0.15 \Big)_+,0.0555\Big)$ at time $t_k$.}}
		\label{3F:fig:prdc}
	\end{figure}
	The interesting feature of such functions is that their (right) derivative have a compact support.
\end{remark}

\subsection{Backward Dynamic Programming Principle}

$V_k$ can also be defined recursively by
\begin{equation}\label{3F:BDPP}
	\left\{
	\begin{aligned}
		 & V_n = \e^{- \int_{0}^{t_n} r_s^d ds } \psi_n( S_{t_n} ),                                                                                  \\
		 & V_k = \max \Big( \e^{- \int_{0}^{t_k} r_s^d ds } \psi_k( S_{t_k} ), \E [ V_{k+1} \mid \F_{t_k} ] \Big) \mathrm{,\qquad} 0 \leq k \leq n-1
	\end{aligned} \right.
\end{equation}
and this representation is called the \textit{Backward Dynamic Programming Principle} (BDPP).

First, noticing that the obstacle process $\e^{- \int_{0}^{t} r_s^d ds } \psi_{t} ( S_{t} )$ can be rewritten as a function $h_t$ of two processes $X_t$ and $Y_t$ such that
\begin{equation*}
	h_t (X_t,Y_t) = \e^{- \int_{0}^{t} r_s^d ds } \psi_{t} ( S_{t} )
\end{equation*}
where $h$ is given by
\begin{equation}\label{3F:payoff_function_of_XY}
	h_t (x,y) = \varphi_d(t) \e^{-y} \psi_t \bigg( S_{0} \frac{\varphi_f(t)}{\varphi_d(t)} \e^{- \sigma^2_S t / 2 + x + y } \bigg)
\end{equation}
and $(X,Y)$ is defined by
\begin{equation}
	(X_t,Y_t) = \bigg( \sigma_S W^S_t + \sigma_f \int_0^t (t-s) dW_s^f, - \sigma_d \int_0^t (t-s) dW_s^d \bigg).
\end{equation}

Now, in order to alleviate notations, we denote by $X_k = X_{t_k}$, $W_k^f = W_{t_k}^f$, $Y_k = Y_{t_k}$, $W_k^d = W_{t_k}^d$, $W_k^S = W_{t_k}^S$ and $h_k = h_{t_k}$.

Using this new form, the Snell envelope becomes
\begin{equation*}
	V_k = \sup_{\tau \in \mathcal{T}_k^n} \E \big[ h_{\tau} (X_{\tau}, Y_{\tau}) \mid \F_{t_k} \big]
\end{equation*}
and the \textit{Backward Dynamic Programming Principle} \eqref{3F:BDPP} rewrites
\begin{equation}\label{3F:BDPP_XY}
	\left\{
	\begin{aligned}
		 & V_n = h_n (X_n, Y_n), \\
		 & V_k = \max \Big( h_k (X_k, Y_k), \E \big[ V_{k+1} \mid \F_{t_k} ] \Big) \mathrm{,\qquad} 0 \leq k \leq n-1.
	\end{aligned} \right.
\end{equation}

Second, in order to solve the problem theoretically by dynamic programming it is required to associate a $\F_t$-Markov process to this problem and in our case, the simplest of them (i.e. of minimal dimension) is $(X_t, W^f_t, Y_t, W_t^d)$ which is $\F_t$-adapted and a Markov process because
\begin{equation*}
	\left\{
	\begin{aligned}
		X_{k+1}   & = X_k + \sigma_f \delta W_k^f + \sigma_S \int_{t_k}^{t_{k+1}} dW^S_s + \sigma_f \int_{t_k}^{t_{k+1}} (t_{k+1}-s) dW_s^f \\
		W_{k+1}^f & = W_k^f + \int_{t_k}^{t_{k+1}} dW_s^f                                                                                   \\
		Y_{k+1}   & = Y_k - \sigma_d \delta W_k^d - \sigma_d \int_{t_k}^{t_{k+1}} (t_{k+1}-s) dW_s^d                                        \\
		W_{k+1}^d & = W_k^d + \int_{t_k}^{t_{k+1}} dW_s^d                                                                                   \\
	\end{aligned}
	\right.
\end{equation*}
where $\delta = \frac{T}{n}$ and can be written as
\begin{equation}\label{3F:eq:markov_property_tuple4}
	\left\{
	\begin{aligned}
		X_{k+1}   & = X_k + \sigma_f \delta W_k^f + G^1_{k+1} \\
		W_{k+1}^f & = W_k^f + G^2_{k+1}                       \\
		Y_{k+1}   & = Y_k - \sigma_d \delta W_k^d + G^3_{k+1} \\
		W_{k+1}^d & = W_k^d + G^4_{k+1},                      \\
	\end{aligned}
	\right.
\end{equation}
where the increments are normally distributed
\begin{equation}
	\begin{pmatrix}
		G^1_{k+1} \\ G^2_{k+1} \\ G^3_{k+1} \\ G^4_{k+1}
	\end{pmatrix}
	\sim
	\N \Big( \mu_{k+1}, \Sigma_{k+1} \Big)
\end{equation}
with
\begin{equation}
	\mu_{k+1} =
	\begin{pmatrix}
		0 \\ 0 \\ 0 \\ 0
	\end{pmatrix}
	\qquad \textrm{and} \qquad
	\Sigma_{k+1} = \bigg( \Cov \big( G_{k+1}^{i}, G_{k+1}^{j} \big) \bigg)_{i,j = 1:4}.
\end{equation}
One notices that $\big( (G^1_{k}, G^2_{k}, G^3_{k}, G^4_{k}) \big)_{k=1 \dots, n}$ are i.i.d. Based on Equation \eqref{3F:eq:markov_property_tuple4}, we deduce the Markov process transition of $(X_k, W_k^f, Y_k, W_k^d)$, for any integrable function $f:\R^4 \rightarrow \R$, given by
\begin{equation}\label{3F:transition_markov_process}
	P f(x,u,y,v) = \E \big[ f ( x + \sigma_f \delta u + G^1_{k+1}, u + G^2_{k+1}, y - \sigma_d \delta v + G^3_{k+1}, v + G^4_{k+1}) \big].
\end{equation}

\begin{remark}
	Using the Markov process $(X, W^f, Y, W^d)$ newly defined, we rewrite the filtration $\F_t$ as
	\begin{equation}
		\F_{t} = \sigma \big( W^S_s, W^d_s, W^f_s, s \leq t \big)
		= \sigma \big( X_s, W^f_s, Y_s, W^d_s, s \leq t \big).
	\end{equation}
\end{remark}
Then, using the new expression for the filtration and the Markov property of $(X_k, W_k^f, Y_k, W_k^d)$, the BDPP \eqref{3F:BDPP_XY} reads as follows,
\begin{equation}\label{3F:BDPP_markovian}
	\left\{
	\begin{aligned}
		 & V_n = h_n (X_n, Y_n),                                                                                                           \\
		 & V_k = \max \Big( h_k (X_k, Y_k), \E \big[ V_{k+1} \mid (X_k, W_k^f, Y_k, W_k^d) \big] \Big) \mathrm{,\qquad} 0 \leq k \leq n-1.
	\end{aligned} \right.
\end{equation}

Moreover, by backward induction we get $V_k = v_k (X_k, W_k^f, Y_k, W_k^d)$ where
\begin{equation}\label{3F:BDPP_markovian_smallv}
	\left\{
	\begin{aligned}
		 & v_n (X_n, W_n^f, Y_n, W_n^d) = h_n (X_n, Y_n),                                                                                         \\
		 & v_k (X_k, W_k^f, Y_k, W_k^d) = \max \Big( h_k (X_k, Y_k), P v_{k+1} (X_k, W_k^f, Y_k, W_k^d) \Big) \mathrm{,\qquad} 0 \leq k \leq n-1.
	\end{aligned} \right.
\end{equation}

\paragraph{Payoff regularity.} First, we look at the regularity of the payoff. The next proposition will then allow us to study the regularity of the value function through the propagation of the local Lipschitz property by the transition of the Markov process.
\begin{proposition} \label{3F:regularity_prdc}
	If $\psi_{t_k}$ is are Lipschitz continuous with Lipschitz coefficient $[\psi_{t_k}]_{_{Lip}}$ with compactly supported (right) derivative (such as the payoff defined in \eqref{3F:payoffPRDC}) then $h_k( x, y )$ given by \eqref{3F:payoff_function_of_XY} is locally Lipschitz continuous, for every $x, x', y, y' \in \R$
	\begin{equation}
		\begin{aligned}
			\vert h_k(x,y) - h_k(x',y') \vert
			 & \leq \e^{ \vert y \vert \vee \vert y' \vert } \big( [\widebar \psi_k]_{_{Lip}} \vert x - x' \vert + ( \varphi_d(t_k) \Vert \psi_{t_k} \Vert_{_{\infty}} + [\widebar \psi_k]_{_{Lip}} ) \vert y - y' \vert \big)
		\end{aligned}
	\end{equation}
	with $[\widebar \psi_k]_{_{Lip}} = [\psi_{t_k}]_{_{Lip}} S_{0} \varphi_f(t_k) \e^{ - \sigma^2_S t_{k} / 2 } \Vert \psi_{t_k}' \Vert_{_{\infty}} \e^c$ with $\psi_{t_k}'$ the right derivative of $\psi_{t_k}$.
\end{proposition}

\begin{proof}
	Let $g_k$ be defined by
	\begin{equation}
		g_k (x,y) = \psi_{t_k} \bigg( S_{0} \frac{\varphi_f(t_k)}{\varphi_d(t_k)} \e^{- \sigma^2_S t_k / 2 + x + y } \bigg).
	\end{equation}
	As $\psi'_{t_k}$ has a compact support, then it exists $c \in \R$ such that
	\begin{equation}
		\vert (\psi_{t_k} (\e^x))' \vert = \vert \e^x \psi_{t_k}' (\e^x) \vert \leq \Vert \psi_{t_k}' \Vert_{_{\infty}} \sup_{x \in \supp \psi_{t_k}' }\e^x \leq \Vert \psi_{t_k}' \Vert_{_{\infty}} \e^c.
	\end{equation}
	Hence
	\begin{equation}
		\begin{aligned}
			\vert g_k(x,y) - g_k(x',y') \vert
			 & \leq \frac{[\widebar \psi_k]_{_{Lip}}}{\varphi_d(t_k)} \big( \vert x - x' \vert + \vert y - y' \vert \big)
		\end{aligned}
	\end{equation}
	with $[\widebar \psi_k]_{_{Lip}} = [\psi_{t_k}]_{_{Lip}} S_{0} \varphi_f(t_k) \e^{ - \sigma^2_S t_{k} / 2 } \Vert \psi_{t_k}' \Vert_{_{\infty}} \e^c $. Then for every $ x, x', y, y' \in \R $, we have
	\begin{equation}
		\begin{aligned}
			\vert h_k(x,y) - h_k(x',y') \vert
			 & = \varphi_d(t_k) \big\vert \e^{-y} g_k(x,y) - \e^{-y'} g_k(x',y') \big\vert                                                                                                                                      \\
			 & \leq \varphi_d(t_k) \Big( \big\vert \e^{-y} g_k(x,y) - \e^{-y'} g_k(x,y) \big\vert + \big\vert \e^{-u'} g_k(x,y) - \e^{-y'} g_k(x',y') \big\vert \Big)                                                           \\
			 & \leq \varphi_d(t_k) \Big( \big\vert \e^{-y} - \e^{-y'} \big\vert \cdot \Vert \psi_{t_k} \Vert_{_{\infty}} + \e^{-y'} \big\vert g_k(x,y) - g_k(x',y') \big\vert \Big)                                             \\
			 & \leq \e^{ \vert y \vert \vee \vert y' \vert } \big( [\widebar \psi_k]_{_{Lip}} \vert x - x' \vert + ( \varphi_d(t_k) \Vert \psi_{t_k} \Vert_{_{\infty}} + [\widebar \psi_k]_{_{Lip}} ) \vert y - y' \vert \big).
		\end{aligned}
	\end{equation}
\end{proof}

The next Lemma shows that the transition of the Markov process propagates the local Lipschitz continuity of a function $f$. This result will be helpful to estimate the error induced by the numerical approximation \eqref{3F:BDPP_markovian_smallv}.

\begin{lemme}\label{3F:KindOfLip_Markov}
	Let $P f(x,u,y,v) = \E \big[ f ( x + \sigma_f \delta u + G^1, u + G^2, y - \sigma_d \delta v + G^3, v + G^4) \big]$ be a Markov kernel.
	If the function $f$ satisfies the following local Lipschitz property,
	\begin{equation}
		\begin{aligned}
			\vert f(x,u,y,v) - f(x',u',y',v') \vert
			 & \leq \big(A \vert x - x' \vert + B \vert u - u' \vert + C\vert y - y' \vert + D \vert v - v' \vert \big)          \\
			 & \qquad \qquad \qquad \qquad \times \e^{ \vert y \vert \vee \vert y' \vert + b \vert v \vert \vee \vert v' \vert } \\
		\end{aligned}
	\end{equation}
	then
	\begin{equation}
		\begin{aligned}
			\vert Pf(x,u,y,v) - Pf(x',u',y',v') \vert
			 & \leq \big(\widetilde A \vert x - x' \vert + \widetilde B \vert u - u' \vert + \widetilde C \vert y - y' \vert + \widetilde D \vert v - v' \vert \big) \\
			 & \qquad \qquad \qquad \qquad \times \e^{ \vert y \vert \vee \vert y' \vert + \widetilde b \vert v \vert \vee \vert v' \vert }.
		\end{aligned}
	\end{equation}
\end{lemme}

\begin{proof}
	It follows from Jensen's inequality and our assumption on $f$
	\begin{equation}
		\begin{aligned}
			\vert Pf(x,u,y,v) & - Pf(x',u',y',v') \vert                                                                                                                                                                       \\
			                  & \leq \E \Big[ \big\vert f ( x + \sigma_f \delta u + G^1, u + G^2, y - \sigma_d \delta v + G^3, v + G^4)                                                                                       \\
			                  & \qquad \qquad \qquad - f ( x' + \sigma_f \delta u' + G^1, u' + G^2, y' - \sigma_d \delta v' + G^3, v' + G^4) \big\vert \Big]                                                                  \\
			                  & \leq \big( A \vert x - x' \vert + (B +  A \sigma_f \delta) \vert u - u' \vert + C \vert y - y' \vert + (D + C \sigma_d \delta) \vert v - v' \vert \big)                                       \\
			                  & \qquad \qquad \qquad \qquad \times \e^{ \vert y \vert \vee \vert y' \vert + (b + \sigma_d \delta) \vert v \vert \vee \vert v' \vert } \E \big[ \e^{\vert G^3 \vert + b \vert G^4 \vert} \big] \\
			                  & \leq \big(\widetilde A \vert x - x' \vert + \widetilde B \vert u - u' \vert + \widetilde C \vert y - y' \vert + \widetilde D \vert v - v' \vert \big)                                         \\
			                  & \qquad \qquad \qquad \qquad \times \e^{ \vert y \vert \vee \vert y' \vert + \widetilde b \vert v \vert \vee \vert v' \vert }                                                                  \\
		\end{aligned}
	\end{equation}
	where
	\begin{equation}\label{3F:coeffsLemme_1_prdcpayoff}
		\widetilde A = A \E [ \kappa ], \qquad \widetilde B = ( B+A\sigma_f \delta ) \E [ \kappa ]
	\end{equation}
	and
	\begin{equation}\label{3F:coeffsLemme_2_prdcpayoff}
		\widetilde C = C \E [ \kappa ], \qquad \widetilde D = ( D + C \sigma_d \delta ) \E [ \kappa ], \qquad \widetilde b = b + \sigma_d \delta
	\end{equation}
	with $\kappa = \exp(\vert G^3 \vert + b \vert G^4 \vert ) $ and $\E [ \kappa] < +\infty$.
\end{proof}

\paragraph{Value function regularity.}

If the functions $(\psi_{t_k})_{k=0:n}$ are defined as in Equation \eqref{3F:payoffPRDC} then $v_n(x,u,y,v)$ preserves a local Lipschitz property. Hence, for every $x, x', u, u', y, y', v, v' \in \R$,
\begin{equation}
	\begin{aligned}
		\vert v_n(x,u,y,v) - v_n(x',u',y',v') \vert
		 & \leq \big( A_n \vert x - x' \vert + B_n \vert u - u' \vert + C_n \vert y - y' \vert + D_n \vert v - v' \vert \big)  \\
		 & \qquad \qquad \qquad \qquad \times \e^{ \vert y \vert \vee \vert y' \vert + b_n \vert v \vert \vee \vert v' \vert } \\
	\end{aligned}
\end{equation}
where
\begin{equation}
	\begin{aligned}
		A_n=[\widebar \psi_{n}]_{_{Lip}}, \qquad B_n=0, \qquad C_n = \varphi_d(t_n) \Vert \psi_n \Vert_{_{\infty}} + [\widebar \psi_{n}]_{_{Lip}}, \qquad D_n=0, \qquad b_n = 0 \\
	\end{aligned}
\end{equation}
with $[\widebar \psi_{n}]_{_{Lip}} = [\psi_{t_n}]_{_{Lip}} S_{0} \varphi_f(t_n) \exp ( - \sigma^2_S t_{n} / 2 ) \Vert \psi_{t_n}' \Vert_{_{\infty}} \e^c$. Using now Lemma \ref{3F:KindOfLip_Markov} recursively and the elementary inequality $\max(a, b+c) \leq \max (a,b) + c$ (as $x \mapsto \max(a,x)$ is $1$-Lipschitz), we have
\begin{equation} \label{3F:kindoflip_valuefunction_prdc}
	\begin{aligned}
		\vert v_k(x,u,y,v) & - v_k(x',u',y',v') \vert                                                                                                                                                                                                            \\
		                   & \leq \max \big( \vert h_k( x,y ) - h_k( x', y' ) \vert, \vert P v_{k+1}(x,u,y,v) - P v_{k+1}(x',u',y',v') \vert \big)                                                                                                               \\
		                   & \leq \max \bigg( \e^{ \vert y \vert \vee \vert y' \vert } \big( [\widebar \psi_k]_{_{Lip}} \vert x - x' \vert + \big( \varphi_d(t_k) \Vert \psi_{t_k} \Vert_{_{\infty}} + [\widebar \psi_k]_{_{Lip}} \big) \vert y - y' \vert \big) \\
		                   & \qquad \qquad, \big(\widetilde A_k \vert x - x' \vert + \widetilde B_k \vert u - u' \vert + \widetilde C_k \vert y - y' \vert + \widetilde D_k \vert v - v' \vert \big)                                                             \\
		                   & \qquad \qquad \qquad \qquad \times \e^{ \vert y \vert \vee \vert y' \vert + \widetilde b_k \vert v \vert \vee \vert v' \vert } \bigg)                                                                                               \\
		                   & \leq \big( A_k \vert x - x' \vert + B_k \vert u - u' \vert + C_k \vert y - y' \vert + D_k \vert v - v' \vert \big)                                                                                                                  \\
		                   & \qquad \qquad \qquad \qquad \times \e^{ \vert y \vert \vee \vert y' \vert + b_k \vert v \vert \vee \vert v' \vert }                                                                                                                 \\
	\end{aligned}
\end{equation}
where
\begin{equation}
	A_k = [\widebar \psi_k]_{_{Lip}} \vee \big( A_{k+1} \E [ \kappa_{k+1} ] \big), \qquad B_k = ( B_{k+1}+A_{k+1}\sigma_f \delta ) \E [ \kappa_{k+1} ], \qquad b_k = b_{k+1} + \sigma_d \delta
\end{equation}
and
\begin{equation}
	C_k = \big( \varphi_d(t_k) \Vert \psi_{t_k} \Vert_{_{\infty}} + [\widebar \psi_k]_{_{Lip}} \big) \vee \big( C_{k+1} \E [ \kappa_{k+1} ] \big), \qquad D_k = (D_{k+1} + C_{k+1} \sigma_d \delta) \E [ \kappa_{k+1} ]
\end{equation}
with $\kappa_{k+1} = \exp(\vert G^3_{k+1} \vert + b_{k+1} \vert G^4_{k+1} \vert ) $. Or equivalently
\begin{equation}\label{3F:usefullCoeffs1_prdc}
	A_k = \max_{l \geq k} \bigg( [\widebar \psi_{l}]_{_{Lip}} \prod_{j=k+1}^{l} \E [ \kappa_{j} ] \bigg), \qquad B_k = \sigma_f \frac{T}{n} \sum_{l=k+1}^{n} \max_{ l \leq i \leq n} \bigg( [\widebar \psi_{i}]_{_{Lip}} \prod_{j=k+1}^{i} \E [ \kappa_{j} ] \bigg)
\end{equation}
and
\begin{equation}\label{3F:usefullCoeffs2_prdc}
	\begin{aligned}
		C_k & = \max_{l \geq k} \bigg( \big(\varphi_d(t_l) \Vert \psi_l \Vert_{_{\infty}} + [\widebar \psi_{l}]_{_{Lip}} \big) \prod_{j=k+1}^{l} \E [ \kappa_{j} ] \bigg) ,                                              \\
		D_k & = \sigma_d \frac{T}{n} \sum_{l=k+1}^{n} \max_{ l \leq i \leq n} \bigg( \big( \varphi_d(t_i) \Vert \psi_i \Vert_{_{\infty}} + [\widebar \psi_{i}]_{_{Lip}} \big) \prod_{j=k+1}^{i} \E [ \kappa_{j} ] \bigg)
	\end{aligned}
\end{equation}
with
\begin{equation}
	b_k = \sigma_d T \Big( 1-\frac{k-1}{n} \Big).
\end{equation}

\section{Bermudan pricing using Optimal Quantization} \label{3F:section:bermudan_evaluation_quantization}

In this section, we propose two numerical solutions based on Product Optimal Quantization for the pricing of Bermudan options on the $FX$ rate $S_t$. First, we remind briefly what is an optimal quantizer and what we mean by a product quantization tree. Second, we present a first numerical solution, based on quantization of the Markovian tuple $(X,W^f,Y,W^d)$, to solve the numerical problem \eqref{3F:BDPP_markovian} and detail the $L^2$-error induced by this approximation. However, remember that we are looking for a method that makes possible to compute xVA's risk measures in a reasonable time but this solution can be too time consuming in practice due to the dimensionality of the quantized processes. That is why we present a second numerical solution which reduces the dimensionality of the problem by considering an approximate problem, based on quantization of the non-Markovian couple $(X,Y)$, introducing a systematic error induced by the non-markovianity and we study the $L^2$-error produced by this approximation.

\subsection{About Optimal Quantization} \label{3F:subsection:optimal_quantiz}

\paragraph{Theoretical background (the one-dimensional case).}

The aim of Optimal Quantization is to determine $\Gamma_N$, a set with cardinality at most $N$, which minimises the quantization error among all such sets $\Gamma$. We place ourselves in the one-dimensional case. Let $Z$ be an $\R$-valued random variable with distribution $\Prob_{_{Z}}$ defined on a probability space $( \Omega, \A, \Prob )$ such that $Z \in L^2_{\R}$.

\begin{definition}
	Let $\Gamma_N = \{ z_1, \dots, z_N \} \subset \R$ be a subset of size $N$, called $N$-quantizer. A Borel partition $\big( C_i (\Gamma_N) \big)_{i \in \llbracket 1, N \rrbracket}$ of $\R$ is a Voronoï partition of $\R$ induced by the $N$-quantizer $\Gamma_N$ if, for every $i = \{ 1, \cdots, N \}$,
	\begin{equation*}
		C_i \big( \Gamma_N \big) \subset \Big\{ \xi \in \R, \vert \xi - z_i \vert \leq \min_{j \neq i }\vert \xi - z_j \vert \Big\}.
	\end{equation*}
	The Borel sets $C_i (\Gamma_N)$ are called Voronoï cells of the partition induced by $\Gamma_N$.
\end{definition}

One can always consider that the quantizers are ordered: $z_1 < z_2 < \cdots < z_{N-1} < z_{N} $ and in that case the Voronoï cells are given by
\begin{equation*}
	C_{k} \big( \Gamma_N \big) = \big( z_{k - 1/2}, z_{k + 1/2} \big], \qquad k \in \llbracket 1, N-1 \rrbracket, \qquad C_{N} \big( \Gamma_N \big) = \big( z_{N - 1/2}, z_{N + 1/2} \big)
\end{equation*}
where $\forall k \in \{ 2, \cdots, N \}, \quad z_{k-1/2} = \frac{z_{k-1} + z_k}{2}$ and $z_{1/2} = \inf \big( \supp (\Prob_{_{Z}}) \big)$ and $z_{N+1/2} = \sup \big( \supp (\Prob_{_{Z}}) \big)$.

\begin{definition}
	Let $\Gamma_N = \{ z_1, \dots, z_N \}$ be an $N$-quantizer. The nearest neighbour projection $\Proj_{\Gamma_N} : \R \rightarrow \{ z_1, \dots, z_N \} $ induced by a Voronoï partition $\big( C_i (\Gamma_N) \big)_{i \in \{ 1, \cdots, N \} }$ is defined by
	\begin{equation*}
		\forall \xi \in \R, \qquad \Proj_{\Gamma_N} (\xi) = \sum_{i = 1}^N z_i \1_{\xi \in C_i (\Gamma_N) }.
	\end{equation*}
	Hence, we can define the quantization of $Z$ as the nearest neighbour projection of $Z$ onto $\Gamma_N$ by composing $\Proj_{\Gamma_N}$ and $X$
	\begin{equation*}
		\widehat{Z}^{\Gamma_N} = \Proj_{\Gamma_N} (Z) = \sum_{i = 1}^N z_i \1_{Z \in C_i (\Gamma_N) }.
	\end{equation*}
\end{definition}

In order to alleviate notations, we write $\widehat{Z}^N$ from now on in place of $\widehat{Z}^{\Gamma_N}$.

Now that we have defined the quantization of $Z$, we explain where does the term "optimal" comes from in the term optimal quantization. First, we define the quadratic distortion function.

\begin{definition}
	The $L^2$-mean quantization error induced by the quantizer $\widehat Z^N$ is defined as
	\begin{equation}\label{3F:L2quantizerror}
		\Vert Z - \widehat{Z}^N \Vert_{_2} = \bigg( \E \Big[ \min_{i \in \{ 1, \cdots, N \} } \vert Z - z_i \vert^2 \Big] \bigg)^{1/2} = \bigg( \int_{\R} \min_{i \in \{ 1, \cdots, N \} } \vert \xi - z_i \vert^2 \Prob_{_{Z}} ( d \xi ) \bigg)^{1/2}.
	\end{equation}

	It is convenient to define the quadratic distortion function at level $N$ as the squared mean quadratic quantization error on $(\R)^N$:
	\begin{equation*}
		\Distortion : z = \big(z_1, \dots, z_N \big) \longmapsto \E \Big[ \min_{i \in \{ 1, \cdots, N \} } \vert Z - z_i \vert^2 \Big] = \Vert Z - \widehat{Z}^N \Vert_{_2}^2.
	\end{equation*}
\end{definition}

\begin{remark}
	All these definitions can be extended to the $L^p$ case. For example the $L^p$-mean quantization error induced by a quantizer of size $N$ is
	\begin{equation}\label{3F:Lpquantizerror}
		\Vert Z - \widehat{Z}^N \Vert_{_p} = \bigg( \E \Big[ \min_{i \in \{ 1, \cdots, N \} } \vert Z - z_i \vert^p \Big] \bigg)^{1/p} = \bigg( \int_{\R} \min_{i \in \{ 1, \cdots, N \} } \vert Z - z_i \vert^p \Prob_{_{Z}} ( d \xi ) \bigg)^{1/p}.
	\end{equation}
\end{remark}

The existence of a $N$-tuple $z^{(N)} = ( z_1, \dots, z_N )$ minimizing the quadratic distortion function $\Distortion$ at level $N$ has been shown and its associated quantizer $\Gamma_N = \{ z_i, i \in \{ 1, \cdots, N \} \}$ is called an optimal quadratic $N$-quantizer, see e.g. \cite{pages2018numerical} for further details and references. We now turn to the asymptotic behaviour in $N$ of the quadratic mean quantization error. The next Theorem, known as Zador's Theorem, provides the sharp rate of convergence of the $L^p$-mean quantization error.

\begin{theorem}{(Zador's Theorem)}\label{3F:zador} Let $p \in (0, + \infty)$.
	\begin{enumerate}[label=(\alph*)]
		\item {\sc Sharp rate}. Let $Z \in L^{p+ \delta}_{\R}(\Prob)$ for some $\delta > 0$. Let $\Prob_{_{Z}} (d \xi) = \varphi(\xi) \cdot \lambda ( d \xi ) + \nu ( d \xi ) $, where $\nu ~ \bot ~ \lambda$ i.e. denotes the singular part of $\Prob_{_{Z}}$ with respect to the Lebesgue measure $\lambda$ on $\R$. Then,
		      \begin{equation}
			      \lim_{N \rightarrow + \infty} N \min_{\Gamma_N \subset \R, \vert \Gamma_N \vert \leq N } \Vert Z - \widehat{Z}^N \Vert_{_p} = \frac{1}{2^p (p+1)} \bigg[ \int_{\R} \varphi^{\frac{1}{1+p}} d \lambda \bigg]^{1+\frac{1}{p} }.
		      \end{equation}

		\item {\sc Non asymptotic upper-bound}. Let $\delta > 0$. There exists a real constant $C_{1,p,\delta} \in (0, +\infty )$ such that, for every $\R$-valued random variable $Z$,
		      \begin{equation}
			      \forall N \geq 1, \qquad \min_{\Gamma_N \subset \R, \vert \Gamma_N \vert \leq N } \Vert Z - \widehat{Z}^N \Vert_{_p} \leq C_{1,p,\delta} \sigma_{\delta+p} (Z) N^{- 1}
		      \end{equation}
		      where, for $r \in (0, + \infty),\sigma_r(Z) = \min_{a \in \R} \Vert Z - a \Vert_{_r} < + \infty$.
	\end{enumerate}
\end{theorem}

The next result answers to the following question: what can be said about the convergence rate of $\E \big[ \vert Z - \widehat{Z}^N \vert^{2 + \beta} \big]$, knowing that $\widehat Z^N$ is a quadratic optimal quantization?

This problem is known as the distortion mismatch problem and has been first addressed by \cite{graf2008distortion} and the results have been extended in Theorem 4.3 of \cite{pages2018improved}.

\begin{theorem}\label{3F:LrLsdistortionmismatch}[$L^r$-$L^s$-distortion mismatch]
	Let $Z:(\Omega, \A, \Prob ) \rightarrow \R$ be a random variable and let $r \in (0, + \infty)$. Assume that the distribution $\Prob_{_{Z}}$ of $Z$ has a non-zero absolutely continuous component with density $\varphi$. Let $(\Gamma_N)_{N \geq 1}$ be a sequence of $L^r$-optimal grids.
	Let $s \in (r, r+1)$. If
	\begin{equation}
		Z \in L^{\frac{s}{1+r-s}+\delta}(\Omega, \A, \Prob)
	\end{equation}
	for some $\delta>0$, then
	\begin{equation}
		\limsup_N N \Vert Z - \widehat{Z}^N \Vert_{_{s}} < + \infty.
	\end{equation}
\end{theorem}

\paragraph{Product Quantization.} Now, let $Z = (Z^{\ell})_{\ell=1:d}$ be an $\R^d$-valued random vector with distribution $\Prob_{_{Z}}$ defined on a probability space $( \Omega, \A, \Prob )$. There are two approaches if one wishes to scale to higher dimensions. Either one applies the above framework directly to the random vector $Z$ and build an optimal quantizer of $Z$, or one may consider separately each component $Z^{\ell}$ independently, build a one-dimensional optimal quantization $\widehat Z^{\ell}$, of size $N^{\ell}$, with quantizer $ \Gamma_{\ell}^{N^{\ell}} = \big\{ z_{i_{\ell}}^{\ell}, i_{\ell} \in \{ 1, \cdots, N^{\ell} \} \big\} $ and then build the product quantizer $\Gamma^{N}= \prod_{\ell=1}^d \Gamma_l^{N^{\ell}}$ of size $N = N^{1} \times \cdots \times N^{d}$ defined by
\begin{equation}
	\Gamma^{N} = \big\{ ( z_{i_1}^{1}, \cdots, z_{i_{\ell}}^{\ell}, \cdots, z_{i_d}^{d} ),\quad i_{\ell} \in \{ 1, \cdots, N_{\ell} \},\quad \ell \in \{ 1, \cdots, d \} \big\}.
\end{equation}

In our case we chose the second approach. Indeed, it is much more flexible when dealing with normal distribution, like in our case. We do not need to solve the $d$-dimensional minimization problem at each time step. We only need to load precomputed optimal quantizer of standard normal distribution $\N(0,1)$ and then take advantage of the stability of optimal quantization by rescaling in one dimension in the sense that if $\Gamma^N = \{ z_i, 1 \leq i \leq N \}$ is optimal at level $N$ for $\N(0,1)$ then $\mu+\sigma \Gamma^N$ (with obvious notations) is optimal for $\N(\mu, \sigma^2)$.

Even though it exists fast methods for building optimal quantizers in two-dimension based on deterministic methods like in the one-dimensional case, when dealing with optimal quantization of bivariate Gaussian vector, we may face numerical instability when the covariance matrix is ill-conditioned: so is the case if the variance of one coordinate is relatively high compared to the second one (which is our case in this paper). This a major drawback as we are looking for a fast numerical solution able to produce prices in a few seconds and this is possible when using product optimal quantization.

\paragraph{Quantization Tree.} Now, in place of considering a random variable $Z$, let $(Z_t)_{t\in[0,T]}$ be a stochastic process following a Stochastic Differential Equation (SDE)
\begin{equation}
	Z_t = Z_0 + \int_0^t b_s (Z_s) ds + \int_0^t \sigma(s,Z_s) dW_s
\end{equation}
with $Z_0 = z_0 \in \R^d$, $W$ a standard Brownian motion living on a probability space $(\Omega, \A, \Prob)$ and $b$ and $\sigma$ satisfy the standard assumptions in order to ensure the existence of a strong solution of the SDE.

What we call Quantization Tree is defined, for chosen time steps $t_k = Tk/n, k=0, \cdots, n$, by quantizers $\widehat Z_k$ of $Z_k$ (Product Quantizers in our case) at dates $t_k$ and the transition probabilities between date $t_k$ and date $t_{k+1}$. Although $(\widehat Z_k)_k$ is no longer a Markov process we will consider the transition probabilities $\pi_{ij}^k = \Law (\widehat Z_{k+1} \mid \widehat Z_k)$. We can apply this methodology because, with the model we consider, we know all the marginal laws of our processes at each date of interest.

In the next subsection, we present the approach based on the quantization tree previously defined that allows us to approximate the price of Bermudan options where the risk factors are driven by the 3-factor model \eqref{3F:3FModelFX}.

\subsection{Quantization tree approximation: Markov case} \label{3F:subsection:markov}

Our first idea in order to discretize \eqref{3F:BDPP_markovian} is to replace the processes by a product quantizer composed with optimal quadratic quantizers. Indeed, at each time $t_k$, we know the law of the processes $X_k$, $W_k^f$, $Y_k$ and $W_k^d$. Then we "force" in some sense the (lost) Markov property by introducing the \textit{Quantized Backward Dynamic Programming Principle} (QBDPP) defined by
\begin{equation}\label{3F:BDPP_quantif_markovian}
	\left\{
	\begin{aligned}
		 & \widehat V_n = h_n (\widehat X_n, \widehat Y_n),                                                                                                                                                        \\
		 & \widehat V_k = \max \Big( h_k (\widehat X_k, \widehat Y_k), \E \big[ \widehat V_{k+1} \mid (\widehat X_k, \widehat W_k^f, \widehat Y_k, \widehat W_k^d) \big] \Big) \mathrm{,\qquad} 0 \leq k \leq n-1,
	\end{aligned} \right.
\end{equation}
where for every $k = 0, \dots, n$, $\widehat X_k$, $\widehat W_k^f$, $\widehat Y_k$ and $\widehat W_k^d$ are quadratic optimal quantizers of $X_k$, $W_k^f$, $Y_k$ and $W_k^d$ of size $N_k^X$, $N_k^{W^f}$, $N_k^Y$ and $N_k^{W^d}$ respectively and we denote $N_k = N_k^X \times N_k^{W^f} \times N_k^Y \times N_k^{W^d}$ the size of the grid of the product quantizer.

We are interested by the error induced by the numerical algorithm defined in \eqref{3F:BDPP_quantif_markovian} and more precisely its $L^2$-error, with in mind that we "lost" the Markov property in the quantization procedure. Moreover, this can be circumvented as shown below.
\begin{theorem} \label{3F:thm:markov}
	Let the Markov transition $P f(x,u,y,v)$ defined in \eqref{3F:transition_markov_process} be locally Lipschitz in the sense of Lemma \ref{3F:KindOfLip_Markov}. Assume that all the payoff functions $(\psi_{t_k})_{k=0:n}$ are Lipschitz continuous with compactly supported (right) derivative. Then the $L^2$-error induced by the quantization approximation $(\widehat X_k, \widehat W_k^f, \widehat Y_k, \widehat W_k^d)$ is upper-bounded by
	\begin{equation}
		\begin{aligned}
			\big\Vert V_k - \widehat V_k \big\Vert_{_2} & \leq \bigg( \sum_{l=k}^{n} C_{X_l} \big\Vert X_l - \widehat X_l \big\Vert_{_{2p}}^2 + C_{Y_l} \big\Vert Y_l - \widehat Y_l \big\Vert_{_{2p}}^2 + C_{W_l^d} \big\Vert W_l^d - \widehat W_l^d \big\Vert_{_{2p}}^2 + C_{W_l^f} \big\Vert W_l^f - \widehat W_l^f \big\Vert_{_{2p}}^2 \bigg)^{1/2},
		\end{aligned}
	\end{equation}
	where $1<p<3/2$ and $q \geq 1$ such that $\frac{1}{p} + \frac{1}{q} = 1$ and
	\begin{equation} \label{3F:cst_thm_markov_1}
		\begin{aligned}
			C_{X_l} & = [\widebar \psi_l]_{_{Lip}}^2 \big\Vert \e^{ \vert Y_l \vert \vee \vert \widehat Y_l \vert } \big\Vert_{_{2q}}^2 + \widetilde A_l^2 K_l^2, \quad                                                                 & C_{W_l^d} = \widetilde B_l^2 K_l^2, \\
			C_{Y_l} & = \big( \varphi_d(t_l) \Vert \psi_{t_l} \Vert_{_{\infty}} + [\widebar \psi_l]_{_{Lip}} \big)^2 \big\Vert \e^{ \vert Y_l \vert \vee \vert \widehat Y_l \vert } \big\Vert_{_{2q}}^2 + \widetilde C_l^2 K_l^2, \quad & C_{W_l^f} = \widetilde D_l^2 K_l^2  \\
		\end{aligned}
	\end{equation}
	with
	\begin{equation} \label{3F:cst_thm_markov_2}
		K_l = \big\Vert \e^{ \vert Y_l \vert \vee \vert \widehat Y_l \vert + \widetilde b_l \vert W_l^d \vert \vee \vert \widehat W_l^d \vert} \big\Vert_{_{2q}}.
	\end{equation}
	As a consequence if $\widebar N = \min N_k$, we have
	\begin{equation}
		\lim_{\widebar N \rightarrow + \infty} \big\Vert V_k - \widehat V_k \big\Vert_{_2}^2 = 0.
	\end{equation}
\end{theorem}

\begin{remark}
	From the definition of the processes $X_k$, $W_k^f$, $Y_k$ and $W_k^d$, all are Gaussian random variables hence all the $L^{2q}$-norms in Equations \eqref{3F:cst_thm_markov_1} and \eqref{3F:cst_thm_markov_2} are finite. Indeed, let $Z \sim \N (0, \sigma_{_Z})$ a Gaussian random variable with variance $\sigma_{_Z}^2$ and $\widehat Z$ an optimal quantizer of $Z$ with cardinality $N$ then $\forall \lambda \in \R_+$
	\begin{equation}
		\begin{aligned}
			\big\Vert \e^{\lambda \vert Z \vert \vee \vert \widehat Z \vert } \big\Vert_{_{2q}}
			= \bigg( \E \big[ \e^{ 2 q \lambda \vert Z \vert \vee \vert \widehat Z \vert } \big] \bigg)^{\frac{1}{2q}}
			\leq \bigg( 2 \E \big[ \e^{ 2 q \lambda \vert Z \vert } \big] \bigg)^{\frac{1}{2q}}
			\leq 2^{\frac{1}{2q}} \e^{ q^2 \lambda^2 \sigma_{_Z}^2 }.
		\end{aligned}
	\end{equation}
\end{remark}

\begin{proof}
	The error between the Snell envelope and its approximation is given by
	\begin{equation}
		\begin{aligned}
			\vert V_k - \widehat V_k \vert & \leq \max \Big( \big\vert h_k (X_k, Y_k) - h_k (\widehat X_k, \widehat Y_k) \big\vert,                                                                                                                                \\
			                               & \qquad \qquad \qquad \qquad \qquad \big\vert \E \big[ V_{k+1} \mid (X_k, W_k^f, Y_k, W_k^d) \big] - \E \big[ \widehat V_{k+1} \mid (\widehat X_k, \widehat W_k^f, \widehat Y_k, \widehat W_k^d) \big] \big\vert \Big)
		\end{aligned}
	\end{equation}
	thus, using the local Lipschitz property of $h_k$ established in Proposition \ref{3F:regularity_prdc} and Hölder's inequality with $p, q \geq 1$ such that $\frac{1}{p} + \frac{1}{q} = 1$, the $L^2$-error is upper-bounded by
	\begin{equation}
		\begin{aligned}
			\big\Vert V_k - \widehat V_k \big\Vert_{_2}^2
			 & \leq \big\Vert h_k (X_k, Y_k) - h_k (\widehat X_k, \widehat Y_k) \big\Vert_{_2}^2                                                                                                                                                                                                                                                 \\
			 & \qquad \qquad \qquad + \big\Vert \E \big[ V_{k+1} \mid (X_k, W_k^f, Y_k, W_k^d) \big] - \E \big[ \widehat V_{k+1} \mid (\widehat X_k, \widehat W_k^f, \widehat Y_k, \widehat W_k^d) \big] \big\Vert_{_2}^2.                                                                                                                       \\
			 & \leq \big\Vert \e^{ \vert Y_k \vert \vee \vert \widehat Y_k \vert } \big\Vert_{_{2q}}^2 \Big( \big( \varphi_d(t_k) \Vert \psi_{t_k} \Vert_{_{\infty}} + [\widebar \psi_k]_{_{Lip}} \big)^2 \big\Vert Y_k - \widehat Y_k \big\Vert_{_{2p}}^2 + [\widebar \psi_k]_{_{Lip}}^2 \big\Vert X_k - \widehat X_k \big\Vert_{_{2p}}^2 \Big) \\
			 & \qquad \qquad + \big\Vert \E \big[ V_{k+1} \mid (X_k, W_k^f, Y_k, W_k^d) \big] - \E \big[ \widehat V_{k+1} \mid (\widehat X_k, \widehat W_k^f, \widehat Y_k, \widehat W_k^d) \big] \big\Vert_{_2}^2.
		\end{aligned}
	\end{equation}
	Looking at the last term, we have
	\begin{equation}
		\begin{aligned}
			\E \big[ V_{k+1} \mid (X_k, W_k^f, & Y_k, W_k^d) \big] - \E \big[ \widehat V_{k+1} \mid (\widehat X_k, \widehat W_k^f, \widehat Y_k, \widehat W_k^d) \big]                                                                           \\
			=                                  & \E \big[ V_{k+1} \mid (X_k, W_k^f, Y_k, W_k^d) \big] - \E \big[ V_{k+1} \mid (\widehat X_k, \widehat W_k^f, \widehat Y_k, \widehat W_k^d) \big]                                                 \\
			+                                  & \E \big[ V_{k+1} \mid (\widehat X_k, \widehat W_k^f, \widehat Y_k, \widehat W_k^d) \big] - \E \big[ \widehat V_{k+1} \mid (\widehat X_k, \widehat W_k^f, \widehat Y_k, \widehat W_k^d) \big] ).
		\end{aligned}
	\end{equation}
	Now, we inspect the $L^2$-error of each term on the right-hand side of the equality.
	\begin{itemize}[wide=0pt]
		\item For the first term, notice that
		      \begin{equation}
			      \E \big[ V_{k+1} \mid (X_k, W_k^f, Y_k, W_k^d) \big] = Pv_{k+1} (X_k, W_k^f, Y_k, W_k^d)
		      \end{equation}
		      and
		      \begin{equation}
			      \E \big[ V_{k+1} \mid (\widehat X_k, \widehat W_k^f, \widehat Y_k, \widehat W_k^d) \big] = Pv_{k+1} (\widehat X_k, \widehat W_k^f, \widehat Y_k, \widehat W_k^d)
		      \end{equation}
		      then, we directly apply Lemma \ref{3F:KindOfLip_Markov} on the function $v_{k+1}$ and obtain
		      \begin{equation}
			      \begin{aligned}
				       & \vert Pv_{k+1} (X_k, W_k^f, Y_k, W_k^d) - Pv_{k+1} (\widehat X_k, \widehat W_k^f, \widehat Y_k, \widehat W_k^d) \vert                                                                                                                                                                                                                         \\
				       & \leq \Big( \widetilde A_k \vert X_k - \widehat X_k \vert + \widetilde B_k \vert W_k^f - \widehat W_k^f \vert + \widetilde C_k \vert Y_k - \widehat Y_k \vert + \widetilde D_k \vert W_k^d - \widehat W_k^d \vert \Big) \e^{ \vert Y_k \vert \vee \vert \widehat Y_k \vert + \widetilde b_k \vert W_k^d \vert \vee \vert \widehat W_k^d \vert}
			      \end{aligned}
		      \end{equation}
		      with $\widetilde A_k$, $\widetilde B_k$, $\widetilde C_k$, $\widetilde D_k$ and $\widetilde b_k$ defined by \eqref{3F:coeffsLemme_1_prdcpayoff} and \eqref{3F:coeffsLemme_2_prdcpayoff}. Hence, using Hölder's inequality with $p,q \geq 1$ such that $\frac{1}{p} + \frac{1}{q} = 1$,
		      \begin{equation}
			      \begin{aligned}
				       & \big\Vert \E \big[ V_{k+1} \mid (X_k, W_k^f, Y_k, W_k^d) \big] - \E \big[ V_{k+1} \mid (\widehat X_k, \widehat W_k^f, \widehat Y_k, \widehat W_k^d) \big] \big\Vert_{_2}^2                                                                                                                            \\
				       & \leq \Big(\widetilde A_k^2 \big\Vert X_k - \widehat X_k \big\Vert_{_{2p}}^2 + \widetilde B_k^2 \big\Vert W_k^f - \widehat W_k^f \big\Vert_{_{2p}}^2 + \widetilde C_k^2 \big\Vert Y_k - \widehat Y_k \big\Vert_{_{2p}}^2 + \widetilde D_k^2 \big\Vert W_k^d - \widehat W_k^d \big\Vert_{_{2p}}^2 \Big) \\
				       & \qquad \qquad \times \big\Vert \e^{ \vert Y_k \vert \vee \vert \widehat Y_k \vert + \widetilde b_k \vert W_k^d \vert \vee \vert \widehat W_k^d \vert} \big\Vert_{_{2q}}^2.
			      \end{aligned}
		      \end{equation}
		\item The last one is useful for the induction, indeed
		      \begin{equation}
			      \begin{aligned}
				      \big\Vert \E \big[ V_{k+1} \mid (\widehat X_k, \widehat W_k^f, \widehat Y_k, \widehat W_k^d) \big] - \E \big[ \widehat V_{k+1} \mid (\widehat X_k, \widehat W_k^f, \widehat Y_k, \widehat W_k^d) \big] \big\Vert_{_2}^2
				       & \leq \big\Vert V_{k+1} - \widehat V_{k+1} \big\Vert_{_2}^2.
			      \end{aligned}
		      \end{equation}
	\end{itemize}

	Finally, using the $L^r$-$L^s$ mismatch theorem for the quadratic optimal quantizers $\widehat X_k$ and $\widehat Y_k$, if $1<p<3/2$, then
	\begin{equation}
		\begin{aligned}
			 & \limsup_{N^X_k} N^X_k \Vert X_k - \widehat X_k \Vert_{_{2p}} < + \infty, \qquad \qquad \limsup_{N^Y_k} N^Y_k \Vert Y_k - \widehat Y_k \Vert_{_{2p}} < + \infty,                                \\
			 & \limsup_{N^{W^f}_k} N^{W^f}_k \Vert W^f_k - \widehat W^f_k \Vert_{_{2p}} < + \infty \textrm{ ~~~~and~~~~ } \limsup_{N^{W^d}_k} N^{W^d}_k \Vert W^d_k - \widehat W^d_k \Vert_{_{2p}} < + \infty
		\end{aligned}
	\end{equation}
	this yields
	\begin{equation}
		\begin{aligned}
			 & \big\Vert V_k - \widehat V_k \big\Vert_{_2}^2                                                                                                                                                                                                                                  \\
			 & \leq \big\Vert X_k - \widehat X_k \big\Vert_{_{2p}}^2 \Big( [\widebar \psi_k]_{_{Lip}}^2 \big\Vert \e^{ \vert Y_k \vert \vee \vert \widehat Y_k \vert } \big\Vert_{_{2q}}^2 + \widetilde A_k^2 K_k^2 \Big)                                                                     \\
			 & \qquad + \big\Vert Y_k - \widehat Y_k \big\Vert_{_{2p}}^2 \Big( \big( \varphi_d(t_k) \Vert \psi_{t_k} \Vert_{_{\infty}} + [\widebar \psi_k]_{_{Lip}} \big)^2 \big\Vert \e^{ \vert Y_k \vert \vee \vert \widehat Y_k \vert } \big\Vert_{_{2q}}^2 + \widetilde C_k^2 K_k^2 \Big) \\
			 & \qquad \qquad + \widetilde B_k^2 K_k^2 \big\Vert W_k^f - \widehat W_k^f \big\Vert_{_{2p}}^2 + \widetilde D_k^2 K_k^2 \big\Vert W_k^d - \widehat W_k^d \big\Vert_{_{2p}}^2 + \big\Vert V_{k+1} - \widehat V_{k+1} \big\Vert_{_2}^2                                              \\
			 & \leq \sum_{l=k}^{n} C_{X_l} \big\Vert X_l - \widehat X_l \big\Vert_{_{2p}}^2 + C_{Y_l} \big\Vert Y_l - \widehat Y_l \big\Vert_{_{2p}}^2 + C_{W_l^d} \big\Vert W_l^d - \widehat W_l^d \big\Vert_{_{2p}}^2 + C_{W_l^f} \big\Vert W_l^f - \widehat W_l^f \big\Vert_{_{2p}}^2      \\
			 & \quad \xrightarrow{\widebar N \rightarrow +\infty } 0
		\end{aligned}
	\end{equation}
	where $K_{k} = \big\Vert \e^{ \vert Y_k \vert \vee \vert \widehat Y_k \vert + \widetilde b_k \vert W_k^d \vert \vee \vert \widehat W_k^d \vert} \big\Vert_{_{2q}}$ and $\forall k= 1, \dots, n, \quad C_{X_k}, C_{Y_k}, C_{W_k^d}, C_{W_k^f} < + \infty$.
\end{proof}

\begin{remark}
	The same result can be obtained if we relax the assumption on the payoff $\psi_k$. If we only assume the payoff Lipschitz continuous, we have the same limit with the same rate of convergence, however the constants $C_{X_l}, C_{Y_l}, C_{W_l^d}, C_{W_l^f}$ are not the same.
\end{remark}

To conclude this section, although considering product optimal quantizer in four dimensions for $(X_k, W_k^f, Y_k, W_k^d)$ seems to be natural, the computational cost associated to the resulting QBDPP is too high, of order $O (n \times (\max N_k)^2 )$. Moreover the computation of the transition probabilities needed for the evaluation of the terms $\E \big[ \widehat V_{k+1} \mid (\widehat X_k, \widehat W_k^f, \widehat Y_k, \widehat W_k^d) \big]$ are challenging. These transition probabilities cannot be computed using deterministic numerical integration methods and we have to use Monte Carlo estimators. Even though it is feasible, it is a drawback for the method since it increases drastically the computation time for calibrating the quantization tree. In the next section we provide a solution to these problems which consists in reducing the dimension of the problem at the price of adding a systematic error, which turns out to be quite small in practice.

\subsection{Quantization tree approximation: Non Markov case} \label{3F:subsection:non_markov}

In this part, we want to reduce the dimension of the problem in order to scale down the numerical complexity of the pricer. For that we discard the processes $W^d$ and $W^f$ in the tree and only keep $X$ and $Y$. Doing so, we loose the Markovian property of our original model but we drastically reduce the numerical complexity of the problem. Thence, \eqref{3F:BDPP_markovian} is approximated by
\begin{equation}\label{3F:BDPP_quantif_nonmarkovian}
	\left\{
	\begin{aligned}
		 & \widehat V_n = h_n (\widehat X_n, \widehat Y_n),                                                                                                                       \\
		 & \widehat V_k = \max \Big( h_k (\widehat X_k, \widehat Y_k), \E \big[ \widehat V_{k+1} \mid (\widehat X_k, \widehat Y_k) \big] \Big) \mathrm{,\qquad} 0 \leq k \leq n-1
	\end{aligned} \right.
\end{equation}
where for every $k = 0, \dots, n$, $\widehat X_k$ and $\widehat Y_k$ are quadratic optimal quantizers of $X_k$ and $Y_k$ of size $N_k^X$ and $N_k^Y$, respectively and we denote $N_k = N_k^X \times N_k^Y$ the size of the grid of the product quantizer.

\begin{theorem} \label{3F:thm:non_markov}
	Let the Markov transition $P f(x,u,y,v)$ be defined by \eqref{3F:transition_markov_process} be locally Lipschitz in the sense of Lemma \ref{3F:KindOfLip_Markov}. Assume that all the payoff functions $(\psi_{t_k})_{k=0:n}$ are Lipschitz continuous with compactly supported (right) derivative. Then the $L^2$-error, induced by the quantization approximation $(\widehat X_k, \widehat Y_k)$ is upper-bounded by
	\begin{equation}
		\begin{aligned}
			\big\Vert V_k - \widehat V_k \big\Vert_{_2}
			 & \leq \bigg( \sum_{l=k}^{n-1} C_{W_{l+1}^f} \big\Vert W_{l+1}^f - \E [ W_{l+1}^f \mid (X_l, Y_l) ] \big\Vert_{_{2p}}^2 + C_{W_{l+1}^d} \big\Vert W_{l+1}^d - \E [ W_{l+1}^d \mid (X_l, Y_l) ] \big\Vert_{_{2p}}^2 \\
			 & \qquad \qquad \qquad + C_{X_{l}} \big\Vert X_{l} - \widehat X_{l} \big\Vert_{_{2p}}^2 + C_{Y_{l}} \big\Vert Y_{l} - \widehat Y_{l} \big\Vert_{_{2p}}^2 \bigg)^{1/2}                                              \\
		\end{aligned}
	\end{equation}
	where $1<p<3/2$ and $q \geq 1$ such that $\frac{1}{p} + \frac{1}{q} = 1$, moreover
	\begin{equation}
		\begin{aligned}
			C_{X_l} & = [\widebar \psi_l]_{_{Lip}}^2 \big\Vert \e^{ \vert Y_l \vert \vee \vert \widehat Y_l \vert } \big\Vert_{_{2q}}^2 + \widebar A_l^2 \big\Vert \e^{\widebar b_l \vert Y_l \vert \vee \vert \widehat Y_l \vert} \big\Vert_{_{2q}}^2, \quad
			        & C_{W_{l+1}^f} = B_{l+1}^2 \big\Vert \widetilde \kappa_{k+1} \big\Vert_{_{2q}}^2,                                                                                                                                                                                                                        \\
			C_{Y_l} & = \big( \varphi_d(t_l) \Vert \psi_{t_l} \Vert_{_{\infty}} + [\widebar \psi_l]_{_{Lip}} \big)^2 \big\Vert \e^{ \vert Y_l \vert \vee \vert \widehat Y_l \vert } \big\Vert_{_{2q}}^2 + \widebar C_l^2 \big\Vert \e^{\widebar b_l \vert Y_l \vert \vee \vert \widehat Y_l \vert} \big\Vert_{_{2q}}^2, \quad
			        & C_{W_{l+1}^d} = D_{l+1}^2 \big\Vert \widetilde \kappa_{k+1} \big\Vert_{_{2q}}^2.                                                                                                                                                                                                                        \\
		\end{aligned}
	\end{equation}
	Taking the limit in $\widebar N = \min N_k$, the size of the quadratic optimal quantizers, we have
	\begin{equation}\label{3F:methodo_error_approx}
		\lim_{\widebar N \rightarrow + \infty} \big\Vert V_k - \widehat V_k \big\Vert_{_2}^2 = \sum_{l=k}^{n-1} C_{W_{l+1}^f} \big\Vert W_{l+1}^f - \E [ W_{l+1}^f \mid (X_l, Y_l) ] \big\Vert_{_{2p}}^2 + C_{W_{l+1}^d} \big\Vert W_{l+1}^d - \E [ W_{l+1}^d \mid (X_l, Y_l) ] \big\Vert_{_{2p}}^2.
	\end{equation}
\end{theorem}

\begin{proof}
	We apply the same methodology as in the proof for the Markov case. The error between the Snell envelope and its approximation is given by
	\begin{equation}
		\vert V_k - \widehat V_k \vert \leq \max \Big( \big\vert h_k (X_k, Y_k) - h_k (\widehat X_k, \widehat Y_k) \big\vert, \big\vert \E \big[ V_{k+1} \mid (X_k, W_k^f, Y_k, W_k^d) \big] - \E \big[ \widehat V_{k+1} \mid (\widehat X_k, \widehat Y_k) \big] \big\vert \Big)
	\end{equation}
	thus, using Proposition \ref{3F:regularity_prdc} and Hölder's inequality with $p,q\geq 1$ such that $\frac{1}{p} + \frac{1}{q} = 1$, the $L^2$-error is given by
	\begin{equation} \label{3F:error_markov}
		\begin{aligned}
			\big\Vert V_k - \widehat V_k \big\Vert_{_2}^2
			 & \leq \big\Vert h_k (X_k, Y_k) - h_k (\widehat X_k, \widehat Y_k) \big\Vert_{_2}^2                                                                                                                                                                                                                                                 \\
			 & \qquad \qquad \qquad + \big\Vert \E \big[ V_{k+1} \mid (X_k, W_k^f, Y_k, W_k^d) \big] - \E \big[ \widehat V_{k+1} \mid (\widehat X_k, \widehat Y_k) \big] \big\Vert_{_2}^2                                                                                                                                                        \\
			 & \leq \Big( [\widebar \psi_k]_{_{Lip}}^2 \big\Vert X_k - \widehat X_k \big\Vert_{_{2p}}^2 + \big( \varphi_d(t_k) \Vert \psi_{t_k} \Vert_{_{\infty}} + [\widebar \psi_k]_{_{Lip}} \big)^2 \big\Vert Y_k - \widehat Y_k \big\Vert_{_{2p}}^2 \Big) \big\Vert \e^{ \vert Y_k \vert \vee \vert \widehat Y_k \vert } \big\Vert_{_{2q}}^2 \\
			 & \qquad \qquad + \big\Vert \E \big[ V_{k+1} \mid (X_k, W_k^f, Y_k, W_k^d) \big] - \E \big[ \widehat V_{k+1} \mid (\widehat X_k, \widehat Y_k) \big] \big\Vert_{_2}^2.
		\end{aligned}
	\end{equation}
	The last term in Equation \eqref{3F:error_markov} can be decomposed as follows
	\begin{equation}
		\begin{aligned}
			\E \big[ V_{k+1} \mid (X_k, W_k^f, Y_k, W_k^d) \big] & - \E \big[ \widehat V_{k+1} \mid (\widehat X_k, \widehat Y_k) \big]                                                           \\
			=                                                    & \E \big[ V_{k+1} \mid (X_k, W_k^f, Y_k, W_k^d) \big] - \E \big[ V_{k+1} \mid (X_k,Y_k) \big]                                  \\
			+                                                    & \E \big[ V_{k+1} \mid (X_k,Y_k) \big] - \E \big[ V_{k+1} \mid (\widehat X_k, \widehat Y_k) \big]                              \\
			+                                                    & \E \big[ V_{k+1} \mid (\widehat X_k, \widehat Y_k) \big] - \E \big[ \widehat V_{k+1} \mid (\widehat X_k, \widehat Y_k) \big].
		\end{aligned}
	\end{equation}
	And again, each term can be upper-bounded.
	\begin{itemize}[wide=0pt]
		\item The first can be upper-bounded using what we did above on the value function $v_k$ and Hölder's inequality with $p,q \geq 1$ such that $\frac{1}{p} + \frac{1}{q} = 1$
		      \begin{equation}
			      \begin{aligned}
				      \big\Vert \E & \big[ V_{k+1} \mid (X_k, W_k^f, Y_k, W_k^d) \big] - \E \big[ V_{k+1} \mid (X_k,Y_k) \big] \big\Vert_{_2}^2                                                                                                                                               \\
				                   & \leq \big\Vert V_{k+1} - \E \big[ V_{k+1} \mid (X_k,Y_k) \big] \big\Vert_{_2}^2                                                                                                                                                                          \\
				                   & \leq \big\Vert v_{k+1} (X_{k+1}, W_{k+1}^f, Y_{k+1}, W_{k+1}^d)                                                                                                                                                                                          \\
				                   & \qquad\qquad\qquad\qquad\qquad\quad- v_{k+1} \big( X_{k+1}, \E \big[ W_{k+1}^f \mid (X_k,Y_k) \big], Y_{k+1}, \E \big[ W_{k+1}^d \mid (X_k,Y_k) \big] \big) \big\Vert_{_2}^2                                                                             \\
				                   & \leq \Big\Vert \Big( B_{k+1} \big\vert W_{k+1}^f - \E \big[ W_{k+1}^f \mid (X_k,Y_k) \big] \big\vert + D_{k+1} \big\vert W_{k+1}^d - \E \big[ W_{k+1}^d \mid (X_k,Y_k) \big] \big\vert \Big) \widetilde \kappa_{k+1} \Big\Vert_{_2}^2                    \\
				                   & \leq \big\Vert \widetilde \kappa_{k+1} \big\Vert_{_{2q}}^2 \Big( B_{k+1}^2 \big\Vert W_{k+1}^f - \E [ W_{k+1}^f \mid (X_k,Y_k) ] \big\Vert_{_{2p}}^2 + D_{k+1}^2 \big\Vert W_{k+1}^d - \E \big[ W_{k+1}^d \mid (X_k,Y_k) \big] \big\Vert_{_{2p}}^2 \Big)
			      \end{aligned}
		      \end{equation}
		      with coefficients $b_{k+1}$, $B_{k+1}$ and $D_{k+1}$ defined in \eqref{3F:usefullCoeffs1_prdc} and \eqref{3F:usefullCoeffs2_prdc} and
		      \begin{equation}
			      \widetilde \kappa_{k+1} = \e^{ \vert Y_{k+1} \vert + b_{k+1} \vert W_{k+1}^d \vert \vee \vert \E [ W_{k+1}^d \mid (X_k,Y_k) ] \vert}.
		      \end{equation}

		\item For the second, we define
		      \begin{equation}
			      \widetilde v_k (X_k,Y_k) = \E \big[ v_{k+1} (X_{k+1}, W_{k+1}^f, Y_{k+1}, W_{k+1}^d) \mid (X_k,Y_k) \big].
		      \end{equation}
		      Indeed, $\E \big[ V_{k+1} \mid (X_k,Y_k) \big]$ is only a function of $X_k$ and $Y_k$, as shown below
		      \begin{equation}
			      \begin{aligned}
				      \E \big[ V_{k+1} \mid (X_k,Y_k) \big]
				       & = \E \big[ v_{k+1} (X_{k+1}, W_{k+1}^f, Y_{k+1}, W_{k+1}^d) \mid (X_k,Y_k) \big]                                              \\
				       & = \E \Big[ \E \big[ v_{k+1} (X_{k+1}, W_{k+1}^f, Y_{k+1}, W_{k+1}^d) \mid (X_k, W_k^f, Y_k, W_k^d) \big] \mid (X_k,Y_k) \Big] \\
				       & = \E \big[ Pv_{k+1} (X_k, W_k^f, Y_k, W_k^d) \mid (X_k,Y_k) \big].
			      \end{aligned}
		      \end{equation}
		      Moreover, we can rewrite $W_k^f = \lambda_k X_k \overset{\indep}{+} \xi_k$ and $W_k^d = \widetilde \lambda_k Y_k \overset{\indep}{+} \chi_k$ where
		      \begin{equation*}
			      \lambda_k = \frac{\Cov (X_k, W_k^f)}{\V (X_k)}, \qquad \widetilde \lambda_k = \frac{\Cov (Y_k, W_k^d)}{\V (Y_k)}
		      \end{equation*}
		      and $\xi_k \sim \N (0, \sigma_{\xi_k}^2)$ and $\chi_k \sim \N (0, \sigma_{\chi_k}^2)$ with $\sigma_{\xi_k}^2 = \V (W_k^f - \lambda_k X_k)$ and $\sigma_{\chi_k}^2 = \V (W_k^d - \widetilde \lambda_k Y_k)$, then
		      \begin{equation}
			      \begin{aligned}
				       & \E \big[ Pv_{k+1} (X_k, W_k^f, Y_k, W_k^d) \mid (X_k,Y_k) = (x,y) \big]                                                                                       \\
				       & \qquad\qquad\qquad\qquad\qquad\qquad = \E \big[ Pv_{k+1}(x, \lambda_k x + \xi_k, y, \widetilde \lambda_k y + \chi_k) \big] \big\vert_{_{ (x,y) = (X_k, Y_k)}}
			      \end{aligned}
		      \end{equation}
		      yielding
		      \begin{equation}
			      \widetilde v_k (x,y) = \E \big[ Pv_{k+1}(x, \lambda_k x + \xi_k, y, \widetilde \lambda_k y + \chi_k) \big].
		      \end{equation}
		      Now, using Lemma \ref{3F:KindOfLip_Markov} on $\widetilde v_k$, we have
		      \begin{equation}
			      \begin{aligned}
				       & \big\vert \widetilde v_k (x,y) - \widetilde v_k (x', y') \big\vert                                                                                                                                                                                                                                                                                   \\
				       & = \Big\vert \E \big[ Pv_{k+1}(x, \lambda_k x + \xi_k, y, \widetilde \lambda_k y + \chi_k) - Pv_{k+1}(x', \lambda_k x' + \xi_k, y', \widetilde \lambda_k y' + \chi_k) \big] \Big\vert                                                                                                                                                                 \\
				       & \leq \E \bigg[ \Big\vert \big( ( \widetilde A_k + \widetilde B_k \vert \lambda_k \vert ) \vert x - x' \vert + ( 1 + \widetilde C_k \vert \widetilde \lambda_k \vert ) \vert y - y' \vert \big) \e^{ ( 1 + \widetilde b_k \vert \widetilde \lambda_k \vert ) \vert y \vert \vee \vert y' \vert + \widetilde b_k \vert \chi_k \vert } \Big\vert \bigg] \\
				       & \leq \Big( \widebar A_k \vert x - x' \vert + \widebar C_k \vert y - y' \vert \Big) \e^{ \widebar b_k \vert y \vert \vee \vert y' \vert}                                                                                                                                                                                                              \\
			      \end{aligned}
		      \end{equation}
		      where
		      \begin{equation}
			      \widebar A_k = ( \widetilde A_k + \widetilde B_k \vert \lambda_k \vert ) \E \big[ \e^{ \widetilde b_k \vert \chi_k \vert} \big], \qquad \widebar C_k = 1 + \widetilde C_k \vert \widetilde \lambda_k \vert,
		      \end{equation}
		      \begin{equation}
			      \widebar b_k = 1 + \widetilde b_k \vert \widetilde \lambda_k \vert
		      \end{equation}
		      with $\widetilde A_k$, $\widetilde B_k$, $\widetilde C_k$ and $\widetilde b_k$ defined in \eqref{3F:coeffsLemme_1_prdcpayoff} and \eqref{3F:coeffsLemme_2_prdcpayoff}. Hence, using Hölder's inequality with $p,q \geq 1$ such that $\frac{1}{p} + \frac{1}{q} = 1$
		      \begin{equation}
			      \begin{aligned}
				      \big\Vert \E \big[ V_{k+1} & \mid (X_k,Y_k) \big] - \E \big[ V_{k+1} \mid (\widehat X_k, \widehat Y_k) \big] \big\Vert_{_2}^2                                                                                                                                                 \\
				                                 & = \big\Vert \widetilde v_k (X_k, Y_k) - \widetilde v_k (\widehat X_k, \widehat Y_k) \big\Vert_{_2}^2                                                                                                                                             \\
				                                 & \leq \Big\Vert \Big( \widebar A_k \big\vert X_k - \widehat X_k \big\vert + \widebar C_k \big\vert Y_k - \widehat Y_k \big\vert \Big) \e^{ \widebar b_k \vert Y_k \vert \vee \vert \widehat Y_k \vert} \Big\Vert_{_{2}}^2                         \\
				                                 & \leq \big\Vert \e^{\widebar b_k \vert Y_k \vert \vee \vert \widehat Y_k \vert} \big\Vert_{_{2q}}^2 \Big(\widebar A_k^2 \big\Vert X_k - \widehat X_k \big\Vert_{_{2p}}^2 + \widebar C_k^2 \big\Vert Y_k - \widehat Y_k \big\Vert_{_{2p}}^2 \Big).
			      \end{aligned}
		      \end{equation}

		\item The last one is useful for the induction, indeed
		      \begin{equation}
			      \begin{aligned}
				      \big\Vert \E \big[ V_{k+1} \mid (\widehat X_k, \widehat Y_k) \big] - \E \big[ \widehat V_{k+1} \mid (\widehat X_k, \widehat Y_k) \big] \big\Vert_{_2}^2
				       & \leq \big\Vert V_{k+1} - \widehat V_{k+1} \big\Vert_{_2}^2.
			      \end{aligned}
		      \end{equation}
	\end{itemize}

	Finally, using the $L^r$-$L^s$ mismatch theorem on the quadratic optimal quantizers $\widehat X_k$ and $\widehat Y_k$, if $1<p<3/2$, then
	\begin{equation}
		\limsup_{N_k^X} N_k^X \Vert X_k - \widehat X_k \Vert_{_{2p}} < + \infty \textrm{ ~~~~and~~~~ } \limsup_{N_k^Y} N_k^Y \Vert Y_k - \widehat Y_k \Vert_{_{2p}} < + \infty
	\end{equation}
	and
	\begin{equation}
		\begin{aligned}
			\big\Vert V_k - & \widehat V_k \big\Vert_{_2}^2                                                                                                                                                                                                                                                                                                                                        \\
			                & \leq \big\Vert X_k - \widehat X_k \big\Vert_{_{2p}}^2 \Big( [\widebar \psi_k]_{_{Lip}}^2 \big\Vert \e^{ \vert Y_k \vert \vee \vert \widehat Y_k \vert } \big\Vert_{_{2q}}^2 + \widebar A_k^2 \big\Vert \e^{\widebar b_k \vert Y_k \vert \vee \vert \widehat Y_k \vert} \big\Vert_{_{2q}}^2 \Big)                                                                     \\
			                & \qquad + \big\Vert Y_k - \widehat Y_k \big\Vert_{_{2p}}^2 \Big( \big( \varphi_d(t_k) \Vert \psi_{t_k} \Vert_{_{\infty}} + [\widebar \psi_k]_{_{Lip}} \big)^2 \big\Vert \e^{ \vert Y_k \vert \vee \vert \widehat Y_k \vert } \big\Vert_{_{2q}}^2 + \widebar C_k^2 \big\Vert \e^{\widebar b_k \vert Y_k \vert \vee \vert \widehat Y_k \vert} \big\Vert_{_{2q}}^2 \Big) \\
			                & \qquad \qquad + B_{k+1}^2 \big\Vert \widetilde \kappa_{k+1} \big\Vert_{_{2q}}^2 \big\Vert W_{k+1}^f - \E [ W_{k+1}^f \mid (X_k,Y_k) ] \big\Vert_{_{2p}}^2                                                                                                                                                                                                            \\
			                & \qquad \qquad \qquad + D_{k+1}^2 \big\Vert \widetilde \kappa_{k+1} \big\Vert_{_{2q}}^2 \big\Vert W_{k+1}^d - \E \big[ W_{k+1}^d \mid (X_k,Y_k) \big] \big\Vert_{_{2p}}^2 + \big\Vert V_{k+1} - \widehat V_{k+1} \big\Vert_{_2}^2                                                                                                                                     \\
			                & \leq \sum_{l=k}^{n-1} C_{W_{l+1}^f} \big\Vert W_{l+1}^f - \E [ W_{l+1}^f \mid (X_l, Y_l) ] \big\Vert_{_{2p}}^2 + C_{W_{l+1}^d} \big\Vert W_{l+1}^d - \E [ W_{l+1}^d \mid (X_l, Y_l) ] \big\Vert_{_{2p}}^2                                                                                                                                                            \\
			                & \qquad \qquad \qquad \qquad \qquad \qquad + C_{X_{l}} \big\Vert X_{l} - \widehat X_{l} \big\Vert_{_{2p}}^2 + C_{Y_{l}} \big\Vert Y_{l} - \widehat Y_{l} \big\Vert_{_{2p}}^2                                                                                                                                                                                          \\
			                & \xrightarrow{\widebar N \rightarrow +\infty } \sum_{l=k}^{n-1} C_{W_{l+1}^f} \big\Vert W_{l+1}^f - \E [ W_{l+1}^f \mid (X_l, Y_l) ] \big\Vert_{_{2p}}^2 + C_{W_{l+1}^d} \big\Vert W_{l+1}^d - \E [ W_{l+1}^d \mid (X_l, Y_l) ] \big\Vert_{_{2p}}^2.
		\end{aligned}
	\end{equation}
\end{proof}

\paragraph{\textit{Practitioner's corner.}}
Market implied values of $\sigma_f$, $\sigma_d$ and $\sigma_S$ used for the numerical computations are usually of order
\begin{equation}
	\sigma_f \approx 0.005, \qquad \sigma_d \approx 0.005, \qquad \sigma_S \approx 0.5
\end{equation}
and in the most extreme cases, we compute Bermudan options on foreign exchange with maturity 20 years ($T = 20$). Thus, we can estimate the order of the induced systematic error. First, we recall the expression of the related coefficients which depends of
\begin{equation}
	\begin{aligned}
		B_k & = \sigma_f \frac{T}{n} \sum_{l=k+1}^{n} \max_{ l \leq i \leq n} \bigg( [\widebar \psi_{i}]_{_{Lip}} \prod_{j=k+1}^{i} \E [ \kappa_j ] \bigg) ,                                                 \\
		D_k & = \sigma_d \frac{T}{n} \sum_{l=k+1}^{n} \max_{ l \leq i \leq n} \bigg( (\varphi_d(t_i) \Vert \psi_i \Vert_{_{\infty}} + [\widebar \psi_{i}]_{_{Lip}}) \prod_{j=k+1}^{i} \E [ \kappa_j ] \bigg)
	\end{aligned}
\end{equation}
with
\begin{equation}
	\kappa_j = \e^{\vert G^3_j \vert + b_j \vert G^4_j \vert}, \qquad \widetilde \kappa_{l+1} = \e^{ \vert Y_{l+1} \vert + b_{l+1} \vert W_{l+1}^d \vert \vee \vert \E [ W_{l+1}^d \mid (X_l,Y_l) ] \vert}
\end{equation}
and
\begin{equation}
	b_k = \sigma_d T \bigg(1-\frac{k-1}{n}\bigg).
\end{equation}

Now, considering the case where the payoffs are the same at each exercise date, the Lipschitz constants can be upper-bounded by $[\widebar \psi]_{_{Lip}}$:
\begin{equation}
	[\widebar \psi_k]_{_{Lip}} = [\psi_{t_k}]_{_{Lip}} S_{0} \varphi_f(t_k) \e^{- \sigma^2_S t_{k} / 2 } \Vert \psi_{t_k}' \Vert_{_{\infty}} \e^c \leq S_{0} [\psi_{t_k}]_{_{Lip}} \Vert \psi_{t_k}' \Vert_{_{\infty}} \e^c =: [\widebar \psi]_{_{Lip}}
\end{equation}
and let $\kappa$ defined by
\begin{equation}
	\begin{aligned}
		\kappa = \max_{k} \E [ \kappa_k ]
		 & = \E \big[ \e^{\vert G^3_0 \vert + b_0 \vert G^4_0 \vert} \big] \leq \frac{1}{2} \E \big[ \e^{2\vert G^3_0 \vert} + \e^{2 b_0 \vert G^4_0 \vert } \big] \\
	\end{aligned}
\end{equation}
moreover, if $Z \sim \N(0,\sigma^2)$ then $\E \big[ \e^{\lambda \vert Z \vert} \big] = \e^{\lambda^2 \sigma^2 /2}$, thence we can upper-bound $\kappa$
\begin{equation}
	\begin{aligned}
		\kappa & \leq \frac{1}{2} \E \big[ \e^{ \sigma^2_3 / 2 } + \e^{ b_0^2 / 2 } \big] = \frac{1}{2} \E \big[ \e^{ \sigma_d^2 / 96 } + \e^{ \sigma_d^2 T^2 / 2 } \big] \approx 1.
	\end{aligned}
\end{equation}
$\kappa$ being bounded, we notice that the main constants $B_k^2$ and $D_k^2$ in the remaining error are of order $\sigma_d^2$ or $\sigma_f^2$, indeed
\begin{equation}
	\begin{aligned}
		B_k & \leq \sigma_f \frac{T}{n} [\widebar \psi]_{_{Lip}} (n-k)\kappa^{n-k} \approx \sigma_f \frac{T}{n} [\widebar \psi]_{_{Lip}} (n-k),                                                                                                      \\
		D_k & \leq \sigma_d \frac{T}{n} \big(\max_{l} \varphi_d(t_l) \Vert \psi \Vert_{_{\infty}} + [\widebar \psi]_{_{Lip}} \big) (n-k) \kappa^{n-k} \approx \sigma_d \frac{T}{n} ( \Vert \psi \Vert_{_{\infty}} + [\widebar \psi]_{_{Lip}}) (n-k).
	\end{aligned}
\end{equation}
Furthermore
\begin{equation}
	\begin{aligned}
		\E \big[ \widetilde \kappa_{k+1}^{2q} \big]
		 & = \E \Big[ \e^{ 2q \vert Y_{k+1} \vert + 2q b_{k+1} \vert W_{k+1}^d \vert \vee \vert \E [ W_{k+1}^d \mid (X_k,Y_k) ] \vert} \Big]                                                             \\
		 & \leq \frac{1}{2} \bigg( \E \Big[ \e^{ 4q \vert Y_{k+1} \vert } \Big] + \E \Big[ \e^{ 4q b_{k+1} \vert W_{k+1}^d \vert \vee \vert \E [ W_{k+1}^d \mid (X_k,Y_k) ] \vert} \Big] \bigg)          \\
		 & \leq \frac{1}{2} \bigg( \E \Big[ \e^{ 4q \vert Y_{k+1} \vert } \Big] + \E \Big[ \e^{ 4q \sigma_d (T-t_k) \vert W_{k+1}^d \vert \vee \vert \E [ W_{k+1}^d \mid (X_k,Y_k) ] \vert} \Big] \bigg) \\
		 & \leq \frac{1}{2} \bigg( \e^{ 8 q^2 \sigma_d^2 T^{3}/3 } + 2 \e^{ 8 q^2 \sigma_d^2 (T-t_k)^2 t_{k+1} } \bigg)                                                                                  \\
	\end{aligned}
\end{equation}
and from elementary inequality $(a+b)^{1/q} \leq a^{1/q} + b^{1/q}, \, a,b \geq 0, \, q \geq 1$
\begin{equation}
	\begin{aligned}
		\big\Vert \widetilde \kappa_{k+1} \big\Vert_{_{2q}}^2 = \E \big[ \widetilde \kappa_{k+1}^{2q} \big]^{\frac{1}{q}}
		 & \leq \bigg( \frac{1}{2} \e^{ 8 q^2 \sigma_d^2 T^{3}/3 } + \e^{ 8 q^2 \sigma_d^2 (T-t_k)^2 t_{k+1} } \bigg)^{\frac{1}{q}}                             \\
		 & \leq \bigg( \frac{1}{2} \e^{ 8 q^2 \sigma_d^2 T^{3}/3 } \bigg)^{\frac{1}{q}} + \bigg( \e^{ 8 q^2 \sigma_d^2 (T-t_k)^2 t_{k+1} } \bigg)^{\frac{1}{q}} \\
		 & \leq \frac{1}{2^{1/q}} \e^{ 8 q \sigma_d^2 T^{3}/3 } + \e^{ 8 q \sigma_d^2 (T-t_k)^2 t_{k+1} }.
	\end{aligned}
\end{equation}
The two terms on the right-hand side of the inequality do not explode. Indeed, the function $g:t \mapsto (T-t)^2 t$, defined for $t\in [0,T]$ with $T = 20$, attains its maximum on $t=20/3$ and $g(20/3) \approx 1185$, hence for the considered values
\begin{equation}
	\forall k = 1, \dots, n, \qquad \big\Vert \widetilde \kappa_{k+1} \big\Vert_{_{2q}}^2 \leq C_{\widetilde \kappa} \approx 6.
\end{equation}
Finally, rewriting the obtained systematic error induced by the approximation with this new informations in \eqref{3F:methodo_error_approx} we have
\begin{equation}
	\begin{aligned}
		\big\Vert V_k - \widehat V_k \big\Vert_{_2}^2
		\xrightarrow{N \rightarrow +\infty }
		     & \sum_{l=k}^{n-1} B_{l+1}^2 \big\Vert \widetilde \kappa_{l+1} \big\Vert_{_{2q}}^2 \big\Vert W_{l+1}^d - \E [ W_{l+1}^d \mid (X_l, Y_l) ] \big\Vert_{_{2p}}^2                                                                                                                       \\
		     & \qquad \qquad \qquad + D_{l+1}^2 \big\Vert \widetilde \kappa_{l+1} \big\Vert_{_{2q}}^2 \big\Vert W_{l+1}^f - \E [ W_{l+1}^f \mid (X_l, Y_l) ] \big\Vert_{_{2p}}^2                                                                                                                 \\
		\leq & \sigma_f^2 \Big( \frac{T}{n} \Big)^2 [\widebar \psi]_{_{Lip}}^2 \sum_{l=k}^{n-1} (n-l)^2 \kappa^{2(n-l)} C_{\widetilde \kappa} \big\Vert W_{l+1}^d - \E [ W_{l+1}^d \mid (X_l, Y_l) ] \big\Vert_{_{2p}}^2                                                                         \\
		     & + \sigma_d^2 \Big( \frac{T}{n} \Big)^2 \big(\max_{l} \varphi_d(t_l) \Vert \psi \Vert_{_{\infty}} + [\widebar \psi]_{_{Lip}} \big)^2                                                                                                                                               \\
		     & \qquad \qquad \qquad \times \sum_{l=k}^{n-1} (n-l)^2 \kappa^{2(n-l)} C_{\widetilde \kappa} \big\Vert W_{l+1}^f - \E [ W_{l+1}^f \mid (X_l, Y_l) ] \big\Vert_{_{2p}}^2                                                                                                             \\
		\leq & 2 \sigma_f^2 \Big( \frac{T}{n} \Big)^2 [\widebar \psi]_{_{Lip}}^2 \sum_{l=k}^{n-1} (n-l)^2 \kappa^{2(n-l)} C_{\widetilde \kappa} \big\Vert W_{l+1}^d \big\Vert_{_{2p}}^2                                                                                                          \\
		     & + 2 \sigma_d^2 \Big( \frac{T}{n} \Big)^2 \big(\max_{l} \varphi_d(t_l) \Vert \psi \Vert_{_{\infty}} + [\widebar \psi]_{_{Lip}} \big)^2 \sum_{l=k}^{n-1} (n-l)^2 \kappa^{2(n-l)} C_{\widetilde \kappa}  \big\Vert W_{l+1}^f \big\Vert_{_{2p}}^2                                     \\
		\leq & \Big( \sigma_f^2 [\widebar \psi]_{_{Lip}}^2 + \sigma_d^2 \big(\max_{l} \varphi_d(t_l) \Vert \psi \Vert_{_{\infty}} + [\widebar \psi]_{_{Lip}} \big)^2 \Big) 4 \frac{C_{\widetilde \kappa}}{\pi^{1/3}} \Big( \frac{T}{n} \Big)^2 \sum_{l=k}^{n-1} t_{l+1} (n-l)^2 \kappa^{2(n-l)}.
	\end{aligned}
\end{equation}
Hence, the systematic error is upper-bounded by the squared volatilities $\sigma_d^2$ and $\sigma_f^2$. These parameters being of order $5 \times 10^{-3}$ at most, the systematic error is negligible as long as these volatilities stay reasonably small.

\begin{remark}
	As in the Markov case, we can extend this result to the case where the payoffs $(\psi_k)_k$ are Lipschitz continuous, however the residual error can not be as easily estimated and controlled.
\end{remark}

\section{Numerical experiments} \label{3F:section:numerics}

In this section, we illustrate the theoretical results found in Section \ref{3F:section:bermudan_evaluation_quantization} regarding the pricing of Bermudan options in the 3-factor model described in Section \ref{3F:section:diffusion_models}. First, we detail both algorithms and how to compute the quantities that appear in them (conditional expectation, conditional probabilities, ...). Then, we test our two numerical solutions for the pricing of European options, whose price is known in closed form. European options are Bermudan options with only one date of exercise, hence when using the non-Markovian approximate we do not introduce the systematic error shown in Theorem \ref{3F:thm:non_markov} but pricing these kind of options is a good benchmark in order to test our methodologies. Finally, we evaluate Bermudan options and compare our two solutions, the Markovian and the non-Markovian approximation.

We have to keep in mind that the computation time is crucial because these pricers are only a small block in the complex computation of xVA's. Indeed, they will be called hundreds of thousands of time each time these risks measures are needed.

All the numerical tests have been carried out in C++ on a laptop with a 2,4 GHz 8-Core Intel Core i9 CPU. The computations of the transition probabilities and the computations of the conditional expectations are parallelized on the CPU.

\begin{remark}
	The computation times given below measure the time needed for loading the pre-computed optimal grids from files, rescaling the optimal quantizers in order to get the right variance, computing the conditional probabilities (the part that demands the most in term of computing power) and finally computing the expectations for the pricing.
	One has to keep in mind that the complexity is linear in function of $n$, the number of exercise dates. Indeed, if we double the number of exercise dates, we double the number of conditional probability matrices and expectations to compute.
\end{remark}

\paragraph{Characterisation of the Quantization Tree.}

In what follows, we describe the choice of parameters we made when building the quantization tree: the time discretisation and the size of each grid at each time.

\begin{itemize}
	\item The time discretisation is an easy choice because it is decided by the characteristics of the financial product. Indeed, we take only one date (and today's date) in the tree if we want to evaluate European options and if we want to evaluate Bermudan options we take as many discretisation dates (plus today's date) in the tree as there are exercise dates in the description of the product.
	\item Then, we have to decide the size of each grid at each date in the tree. In our case, we consider grids of same size at each date hence $N_k = N, \, k = 1 \dots, n$ and then we take $N^X = 10 N^Y$ for both trees. This choice seems to be reasonable because the risk factor $X_k$ is prominent, due to the value of $\sigma_S$ compare to $\sigma_d$. Now, in the Markovian case, we take $N^X = 4 N^{W_f}$ and $N^Y = 4 N^{W_d}$, indeed the two Brownian Motions are important only when we compute the conditional expectation but not when we want to evaluate the payoffs, hence we want to give as much as possible of the budget $N$ to $N^X$ and $N^Y$.
\end{itemize}

\paragraph{The algorithm: Markovian Case.}

Let $(x_{i_1}^k)_{i_1=1:N^X}$, $(u_{i_2}^k)_{i_2=1:N^{W^f}}$, $(y_{i_3}^k)_{i_3=1:N^Y}$ and $(v_{i_4}^k)_{i_4=1:N^{W^d}}$ be the associated centroids of $\widehat X_k$, $\widehat W_k^f$, $\widehat Y_k$ and $\widehat W_k^d$ respectively, at a given time $t_k$ with $0 \leq k \leq n$. Using the discrete property of the optimal quantizers, the conditional expectation appearing in \eqref{3F:BDPP_quantif_markovian} can be rewritten as
\begin{equation}
	\begin{aligned}
		 & \E \big[ \widehat V_{k+1} \mid (\widehat X_k, \widehat W_k^f, \widehat Y_k, \widehat W_k^d) = \big( x_{i_1}^k, u_{i_2}^k, y_{i_3}^k, v_{i_4}^k \big) \big]                                                                                               \\
		 & \qquad \qquad = \E \big[ \widehat{v}_{k+1} (\widehat X_{k+1}, \widehat W_{k+1}^f, \widehat Y_{k+1}, \widehat W_{k+1}^d) \mid (\widehat X_k, \widehat W_k^f, \widehat Y_k, \widehat W_k^d) = \big( x_{i_1}^k, u_{i_2}^k, y_{i_3}^k, v_{i_4}^k \big) \big] \\
		 & \qquad \qquad = \sum_{j_1,j_2,j_3,j_4} \pi_{i,j}^{\textsc{(m)},k} \, \widehat{v}_{k+1} \big( x_{j_1}^{k+1}, u_{j_2}^{k+1}, y_{j_3}^{k+1}, v_{j_4}^{k+1} \big)
	\end{aligned}
\end{equation}
where $\pi_{i,j}^{\textsc{(m)},k}$, with $i= (i_1, i_2, i_3, i_4)$ and $j=(j_1, j_2, j_3, j_4)$, is the conditional probability defined by
\begin{equation*}
	\begin{aligned}
		\pi_{i,j}^{\textsc{(m)},k}
		 & = \Prob \Big( \big(\widehat X_{k+1}, \widehat W^f_{k+1}, \widehat Y_{k+1}, \widehat W^d_{k+1}\big) = \big( x_{j_1}^{k+1}, u_{j_2}^{k+1}, y_{j_3}^{k+1}, v_{j_4}^{k+1} \big)                \\
		 & \qquad \qquad \qquad \qquad \qquad \qquad \qquad \mid \big(\widehat X_k, \widehat W^f_k, \widehat Y_k, \widehat W^d_k\big) = \big( x_{i_1}^k, u_{i_2}^k, y_{i_3}^k, v_{i_4}^k \big) \Big).
	\end{aligned}
\end{equation*}
Due to the dimension of the problem (4 in this case), we cannot compute these probabilities using deterministic methods, hence one has to simulate trajectories of the processes in order to evaluate them. We refer the reader to \cite{printems2005quantization,bally2003quantization,pages2004optimal} for details on the methodology.

A way to reduce the complexity of the problem is to approximate these probabilities by $\widetilde \pi_{i,j}^{\textsc{(m)},k}$, where the conditional part $\big\{ \big(\widehat X_k, \widehat W^f_k, \widehat Y_k, \widehat W^d_k\big) = \big(x_{i_1}^k, u_{i_2}^k, y_{i_3}^k, v_{i_4}^k \big) \big\}$ is replaced by $\big\{ (X_k, W^f_k, Y_k, W^d_k) = \big(x_{i_1}^k, u_{i_2}^k, y_{i_3}^k, v_{i_4}^k \big) \big\}$, yielding
\begin{equation} \label{3F:proba_markovian}
	\begin{aligned}
		\widetilde \pi_{i,j}^{\textsc{(m)},k}
		 & = \Prob \Big( \big(\widehat X_{k+1}, \widehat W^f_{k+1}, \widehat Y_{k+1}, \widehat W^d_{k+1}\big) = \big( x_{j_1}^{k+1}, u_{j_2}^{k+1}, y_{j_3}^{k+1}, v_{j_4}^{k+1} \big) \\
		 & \qquad \qquad \qquad \qquad \qquad \qquad \qquad \mid (X_k, W^f_k, Y_k, W^d_k) = \big( x_{i_1}^k, y_{i_2}^k, u_{i_3}^k, v_{i_4}^k \big) \Big).
	\end{aligned}
\end{equation}
The reason for replacing $\big\{ \big(\widehat X_k, \widehat W^f_k, \widehat Y_k, \widehat W^d_k\big) = \big(x_{i_1}^k, u_{i_2}^k, y_{i_3}^k, v_{i_4}^k \big) \big\}$ by $\big\{ (X_k, W^f_k, Y_k, W^d_k) = \big(x_{i_1}^k, u_{i_2}^k, y_{i_3}^k, v_{i_4}^k \big) \big\}$ is explained in the next paragraph dealing with the Non-Markovian case with lighter notations (see Equation \eqref{3F:proba_2D_not_approximated} and \eqref{3F:proba_2D_approximated}).
Although, these probabilities are easier to calculate, one still has to devise a Monte Carlo simulation in order to evaluate them. This simplification will be useful later in the uncorrelated case.

These remarks allow us to rewritte the QBDPP in the Markovian case \eqref{3F:BDPP_quantif_markovian} as
\begin{equation}
	\left\{
	\begin{aligned}
		 & \widehat{v}_n \big( x_{i_1}^n, u_{i_2}^n, y_{i_3}^n, v_{i_4}^n \big) = h_n \big( x_{i_1}^n, y_{i_3}^n \big),                                                                                                                                                                              \\
		 & \widehat{v}_k \big( x_{i_1}^k, u_{i_2}^k, y_{i_3}^k, v_{i_4}^k \big) = \max \bigg( h_k \big( x_{i_1}^k, y_{i_3}^k \big), \sum_{j_1,j_2,j_3,j_4} \widetilde \pi_{i,j}^{\textsc{(m)},k} \, \widehat{v}_{k+1} \big( x_{j_1}^{k+1}, u_{j_2}^{k+1}, y_{j_3}^{k+1}, v_{j_4}^{k+1} \big) \bigg).
	\end{aligned}
	\right.
\end{equation}

\paragraph{The algorithm: Non-Markovian case.}

Let $(x_{i_1}^k)_{i_1=1:N^X}$ and $(y_{i_3}^k)_{i_3=1:N^Y}$ be the associated centroids of $\widehat X_k$ and $\widehat Y_k$ respectively, at a given time $t_k$ with $0 \leq k \leq n$. Again, as in the Markovian case, using the discrete property of the optimal quantizers, the conditional expectation appearing in \eqref{3F:BDPP_quantif_nonmarkovian} can be rewritten as
\begin{equation}
	\begin{aligned}
		\E \big[ \widehat V_{k+1} \mid (\widehat X_k, \widehat Y_k) = \big(x_{i_1}^k, y_{i_2}^k \big) \big]
		 & = \E \big[ \widehat{v}_{k+1} (\widehat X_{k+1}, \widehat Y_{k+1}) \mid (\widehat X_k, \widehat Y_k) = \big(x_{i_1}^k, y_{i_2}^k \big) \big] \\
		 & = \sum_{j_1,j_2} \pi_{i,j}^{\textsc{(nm)},k} \, \widehat{v}_{k+1} \big( x_{j_1}^{k+1}, y_{j_2}^{k+1} \big)
	\end{aligned}
\end{equation}
where $\pi_{i,j}^{\textsc{(nm)},k}$, with $i= (i_1, i_2)$ and $j=(j_1, j_2)$, is the conditional probability defined by
\begin{equation*}
	\pi_{i,j}^{\textsc{(nm)},k} = \Prob \Big( \big(\widehat X_{k+1}, \widehat Y_{k+1}\big) = \big( x_{j_1}^{k+1}, y_{j_2}^{k+1} \big) \mid \big(\widehat X_k, \widehat Y_k\big) = \big(x_{i_1}^k, y_{i_2}^k \big) \Big).
\end{equation*}
This probability can be computed by numerical integration, ie
\begin{equation} \label{3F:proba_2D_not_approximated}
	\begin{aligned}
		\pi_{i,j}^{\textsc{(nm)},k}
		 & = \Prob \Big( \big(\widehat X_{k+1}, \widehat Y_{k+1}\big) = \big(x_{j_1}^{k+1}, y_{j_2}^{k+1} \big) \mid \big(\widehat X_k, \widehat Y_k\big) = \big(x_{i_1}^k, y_{i_2}^k \big) \Big)                                                  \\
		 & = \Prob \Big( \big(\widehat X_{k+1}, \widehat Y_{k+1}\big) = \big( x_{j_1}^{k+1}, y_{j_2}^{k+1} \big) \mid X_k \in \big( x_{i_1-1/2}^k, x_{i_1+1/2}^k \big), Y_k \in \big( y_{i_2-1/2}^k, y_{i_2+1/2}^k \big) \Big)                     \\
		 & = \int_{x_{i_1-1/2}^k}^{x_{i_1+1/2}^k} \int_{y_{i_2-1/2}^k}^{y_{i_2+1/2}^k} \Prob \Big( \big(\widehat X_{k+1}, \widehat Y_{k+1}\big) = \big( x_{j_1}^{k+1}, y_{j_2}^{k+1} \big) \mid (X_k, Y_k) = (x, y) \Big) f_{\Sigma}(x,y) dx \, dy \\
	\end{aligned}
\end{equation}
where $f_{\Sigma}(x,y)$ is the joint density of a centered bivariate Gaussian vector with covariance matrix $\Sigma$ given by
\begin{equation}
	\Sigma =
	\begin{pmatrix}
		\V (X_k)        & \Cov (X_k, Y_k) \\
		\Cov (X_k, Y_k) & \V (Y_k)        \\
	\end{pmatrix}.
\end{equation}
However, computing the probability in Equation \eqref{3F:proba_2D_not_approximated} can be too time consuming, hence once again, we approximate this probability by $\widetilde \pi_{i,j}^{\textsc{(nm)},k}$, where the conditional part $\big\{ \big(\widehat X_k, \widehat Y_k\big) = \big( x_{i_1}^k, y_{i_2}^k \big) \big\}$ is replaced by $\big\{ (X_k, Y_k) = \big( x_{i_1}^k, y_{i_2}^k \big) \big\}$, yielding
\begin{equation}\label{3F:proba_2D_approximated}
	\widetilde \pi_{i,j}^{\textsc{(nm)},k}
	= \Prob \Big( \big(\widehat X_{k+1}, \widehat Y_{k+1}\big) = \big( x_{j_1}^{k+1}, y_{j_2}^{k+1} \big) \mid (X_k, Y_k) = \big(x_{i_1}^k, y_{i_2}^k \big) \Big).
\end{equation}

From the definition of an optimal quantizer and Equation \eqref{3F:eq:markov_property_tuple4}, this probability can be rewritten as the probability that a correlated bivariate normal distribution lies in a rectangular domain
\begin{equation}\label{3F:proba_nonmar}
	\begin{aligned}
		\widetilde \pi_{i,j}^{\textsc{(nm)},k}
		 & = \Prob \Big( \widehat X_{k+1} = x_{j_1}^{k+1}, \widehat Y_{k+1} = y_{j_2}^{k+1} \mid X_k = x_{i_1}^k, Y_k = y_{i_2}^k \Big)                                                                                               \\
		 & = \Prob \Big( X_{k+1} \in \big( x_{j_1-1/2}^{k+1}, x_{j_1+1/2}^{k+1} \big), Y_{k+1} \in \big( y_{j_2-1/2}^{k+1}, y_{j_2+1/2}^{k+1} \big) \mid X_k = x_{i_1}^k, Y_k = y_{i_2}^k \Big)                                       \\
		 & = \Prob \Big( x_{i_1}^k + \sigma_f \delta W_k^f + G_{k+1}^1 \in \big( x_{j_1-1/2}^{k+1}, x_{j_1+1/2}^{k+1} \big), y_{i_2}^k - \sigma_d \delta W_k^d + G_{k+1}^3 \in \big( y_{j_2-1/2}^{k+1}, y_{j_2+1/2}^{k+1} \big) \Big) \\
		 & = \Prob \Big( Z^1 \in \big( x_{j_1-1/2}^{k+1} - x_{i_1}^k, x_{j_1+1/2}^{k+1}- x_{i_1}^k \big), Z^2 \in \big( y_{j_2-1/2}^{k+1} - y_{i_2}^k, y_{j_2+1/2}^{k+1} - y_{i_2}^k \big) \Big)                                      \\
	\end{aligned}
\end{equation}
where
\begin{equation}
	\begin{pmatrix}
		Z^1 \\ Z^2
	\end{pmatrix}
	\sim
	\N
	\left(
	\begin{pmatrix}
		0 \\ 0
	\end{pmatrix}
	,
	\begin{pmatrix}
		\sigma_{_{Z^1}}^2                                  & \rho_{_{Z^1, Z^2}} \sigma_{_{Z^1}} \sigma_{_{Z^2}} \\
		\rho_{_{Z^1, Z^2}} \sigma_{_{Z^1}} \sigma_{_{Z^2}} & \sigma_{_{Z^2}}^2                                  \\
	\end{pmatrix}
	\right)
\end{equation}
with $\sigma_{_{Z^1}}^2 = \V (\sigma_f \delta W_k^f + G_{k+1}^1)$, $\sigma_{_{Z^2}}^2 = \V (- \sigma_d \delta W_k^d + G_{k+1}^3)$ and $\rho_{_{Z^1, Z^2}} = \Corr (\sigma_f \delta W_k^f + G_{k+1}^1, - \sigma_d \delta W_k^d + G_{k+1}^3)$.

The advantage of expressing \eqref{3F:proba_nonmar} as the probability that a bivariate Gaussian vector lies in a rectangular domain is that it can be rewritten as a linear combination of bivariate cumulative distribution functions.

\begin{wrapfigure}{l}{0.30\textwidth}
	\centering
	\includegraphics[width=0.30\textwidth]{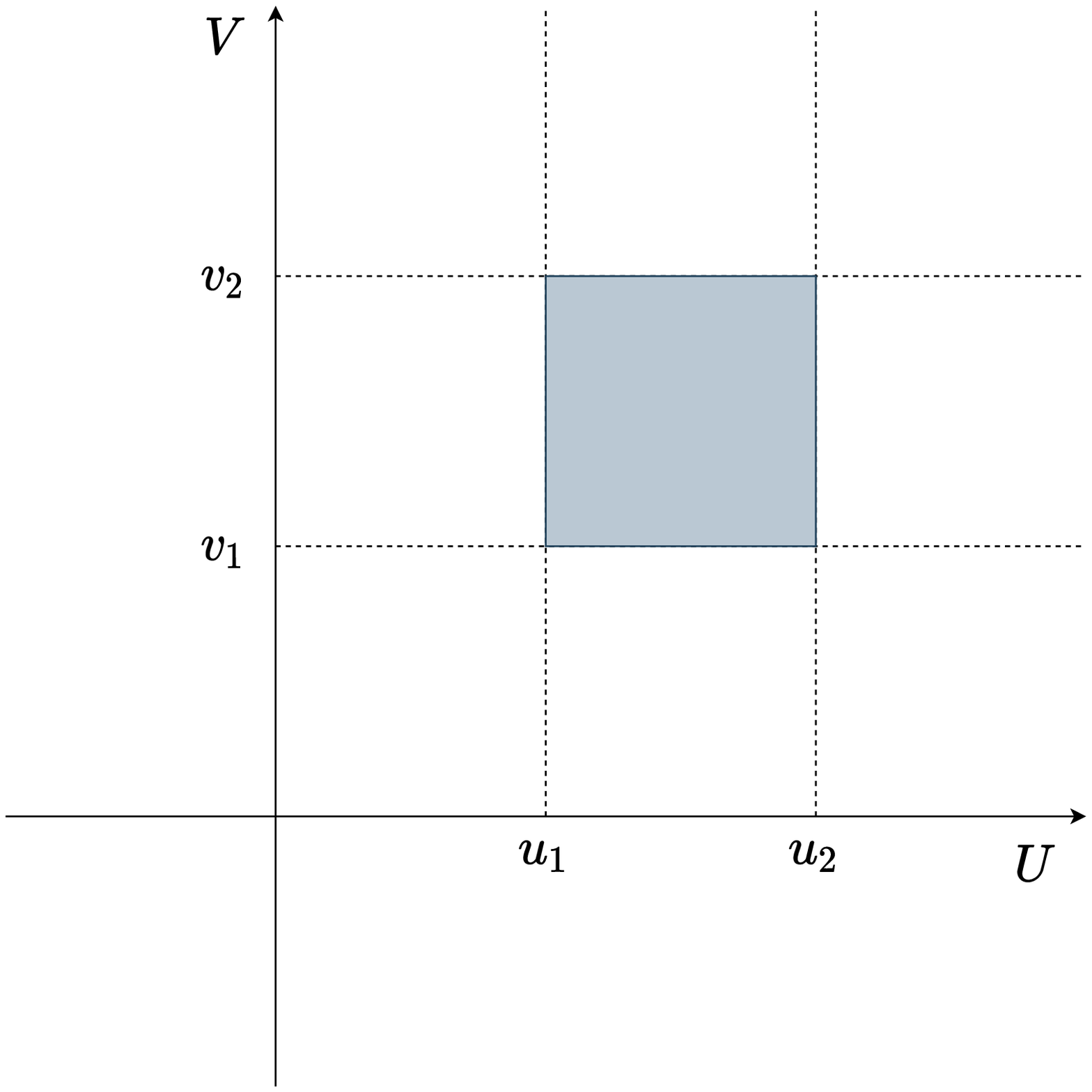}
	\caption[Domain of integration for probabilities of correlated two-dimensional Gaussian random vector.]{}
	\label{3F:fig:domain_integral}
\end{wrapfigure}
Indeed, let $(U,V)$ a two-dimensional correlated and standardized normal distribution with correlation $\rho$ and cumulative distribution function (CDF) given by $F_{U,V}^{\rho} (u,v) = \Prob ( U \leq u, V \leq v)$. Fast and efficient numerical implementation of such function exists (for example, a C++ implementation of the upper right tail of a correlated bivariate normal distribution can be found in John Burkardt's website, see \cite{cdf_bi_variate}, which is based on the work of \cite{donnelly1973algorithm} and \cite{owen1958tables}. In our case, we are interested in the computation of probabilities of the form
\begin{equation} \label{3F:proba_bivariate}
	\Prob \big( U \in ( u_1, u_2 ), V \in ( v_1, v_2 ) \big).
\end{equation}
This probability is represented graphically as the integral of the two-dimensional density over the rectangular domain in grey in Figure \ref{3F:fig:domain_integral}. Now, using $F_{U,V}^{\rho} (u,v)$, the probability \eqref{3F:proba_bivariate} is given by
\begin{equation}\label{3F:decomposition_bivar_proba}
	\begin{aligned}
		\Prob \big( U \in ( u_1, u_2 ), V \in ( v_1, v_2 ) \big)
		 & = F_{U,V}^{\rho} (u_2,v_2) - F_{U,V}^{\rho} (u_1,v_2) - F_{U,V}^{\rho} (u_2,v_1) + F_{U,V}^{\rho} (u_1,v_1).
	\end{aligned}
\end{equation}

This remark will allow us to reduce drastically the computation time induced by the evaluation of the conditional probabilities and so, of the conditional expectations.

Now, going back to our problem, the QBDPP in the non-Markovian case rewrites \eqref{3F:BDPP_quantif_nonmarkovian}
\begin{equation}
	\left\{
	\begin{aligned}
		 & \widehat{v}_n \big( x_{i_1}^n, y_{i_2}^n \big) = h_n \big( x_{i_1}^n, y_{i_2}^n \big), \qquad 1 \leq i_1 \leq N_n^X, \qquad 1 \leq i_2 \leq N_n^Y,                                                                         \\
		 & \widehat{v}_k \big( x_{i_1}^k, y_{i_2}^k \big) = \max \bigg( h_k \big( x_{i_1}^k, y_{i_2}^k \big), \sum_{j_1, j_2} \pi_{i,j}^{\textsc{(nm)},k} \, \widehat{v}_{k+1} \big( x_{j_1}^{k+1}, y_{j_2}^{k+1} \big) \bigg). \quad
	\end{aligned}
	\right.
\end{equation}

In order to test numerically the two methods, we will evaluate PRDC European and Bermudan options with maturities $2Y$, $5Y$  and $10Y$. We describe below the market and products parameters we consider. The volatilities of the domestic and the foreign interest rates are not detailed below because we investigate the behaviour of the methods with respect to $\sigma_d$ and $\sigma_f$.
\begin{table}[H]
	\centering
	\begin{tabular}{||c | c || c | c || c | c ||}
		\hline
		$P_d(0,t)$ & $\exp ( - r_d t )$ & $r_d$      & $0.015$ & $\rho_{Sd}$ & $0$ \\
		$P_f(0,t)$ & $\exp ( - r_f t )$ & $r_f$      & $0.01$  & $\rho_{Sf}$ & $0$ \\
		$S_0$      & $88.17$            & $\sigma_S$ & $0.5$   & $\rho_{df}$ & $0$ \\
		\hline
	\end{tabular}
	\caption[Market values for the three factors model.]{\textit{Market values.}}
	\label{3F:table:market_values}
\end{table}

\begin{table}[H]
	\centering
	\begin{tabular}{||l | c || l | c ||}
		\hline
		$\forall k \in 1, \dots, n, \quad C_d(t_k)$          & $15 \%$   & $\forall k \in 1, \dots, n, \quad C_f(t_k)$            & $18.9 \%$ \\
		$\forall k \in 1, \dots, n, \quad \textrm{Cap}(t_k)$ & $5.55 \%$ & $\forall k \in 1, \dots, n, \quad \textrm{Floor}(t_k)$ & $0 \%$    \\
		Exercise date (EU): $t_n$                            & $T$       & Exercise dates (US): $t_k$                             & $Tk/n$    \\
		\hline
	\end{tabular}
	\caption[PRDC product description.]{\textit{Product description.}}
	\label{3F:table:product_description}
\end{table}

\begin{remark}
	When the correlations $\rho_{df}$ and $\rho_{Sd}$ are equal to zero, the numerical computation of probabilities $\widetilde \pi_{i,j}^{\textsc{(m)},k}$ and $\widetilde \pi_{i,j}^{\textsc{(nm)},k}$ can be accelerated. Indeed, in the Markovian case, \eqref{3F:proba_markovian} can be rewritten as
	\begin{equation}
		\begin{aligned}
			\widetilde \pi_{i,j}^{\textsc{(m)},k}
			 & = \Prob \Big( (\widehat X_{k+1}, \widehat W^f_{k+1}) = \big( x_{j_1}^{k+1}, u_{j_2}^{k+1} \big) \mid (X_k, W^f_k) = \big(x_{i_1}^k, u_{i_2}^k \big) \Big)                      \\
			 & \qquad \qquad \times \Prob \Big( (\widehat Y_{k+1}, \widehat W^d_{k+1}) = \big( y_{j_3}^{k+1}, v_{j_4}^{k+1} \big) \mid (Y_k, W^d_k) = \big( y_{i_3}^k, v_{i_4}^k \big) \Big).
		\end{aligned}
	\end{equation}
	In that case, we can use the CDF of a correlated bivariate normal distribution, as detailed above for the non-Markovian case in \eqref{3F:decomposition_bivar_proba}, for computing these probabilities in a very effective and faster way rather than performing a Monte Carlo simulation.

	In the non-Markovian case, \eqref{3F:proba_nonmar} can be rewritten as
	\begin{equation}
		\begin{aligned}
			\widetilde \pi_{i,j}^{\textsc{(nm)},k}
			 & = \Prob \Big( Z^1 \in \big( x_{j_1-1/2}^{k+1} - x_{i_1}^k, x_{j_1+1/2}^{k+1}- x_{i_1}^k \big) \Big) \Prob \Big( Z^2 \in \big( y_{j_2-1/2}^{k+1} - y_{i_2}^k, y_{j_2+1/2}^{k+1} - y_{i_2}^k \big) \Big)                                                                                                                              \\
			 & = \bigg( F_{_Z} \Big( \frac{x_{j_1+1/2}^{k+1}- x_{i_1}^k}{\sigma_{_{Z^1}}} \Big) - F_{_Z} \Big( \frac{x_{j_1-1/2}^{k+1}- x_{i_1}^k}{\sigma_{_{Z^1}}} \Big) \bigg) \bigg( F_{_Z} \Big( \frac{y_{j_2+1/2}^{k+1} - y_{i_2}^k}{\sigma_{_{Z^2}}} \Big) - F_{_Z} \Big( \frac{y_{j_2-1/2}^{k+1} - y_{i_2}^k}{\sigma_{_{Z^2}}} \Big) \bigg)
		\end{aligned}
	\end{equation}
	where $F_{_Z} (\cdot)$ is the CDF of a one-dimensional normal distribution, $\sigma_{_{Z^1}}$ is the standard deviation of $Z^1$ and $\sigma_{_{Z^2}}$ is the standard deviation of $Z^2$. This remark allows us to drastically reduce the computation time of the conditional probabilities in the case of zero correlations.
\end{remark}

\subsection{European Option}

First of all, we compare the asymptotic behaviour of the Markovian and the non-Markovian approaches when pricing European PRDC Options with different volatilities and maturities. In this case, we consider only two dates in the tree: $t_0=0$ an $t_n=T$, the algorithm is a regular cubature formula and no systematic error is induced by the non-markovianity of the couple $(X_k, Y_k)$. These numerical tests confirm that both approaches give the same value, however the non-Markovian approach converges much faster due to the dimension of the product quantization, $2$ for the first one and $4$ for the last one. Indeed, the complexity of the 2 dimensional pricer is of order of $N= N^X \times N^Y$ while the complexity of the 4 dimensional pricer is of order $N=N^X \times N^Y \times N^{W^d} \times N^{W^f}$. $N$ being the size of the product quantizer at each date (in two dimensions: $N = N^X\times N^Y$ and in four dimensions $N = N^X \times N^{W^f} \times N^Y \times N^{W^d}$).

\medskip

In the case of the European options, we have a closed-form formula for the price of \eqref{3F:payoffPRDC}. The benchmark price is computed using the rewriting of \eqref{3F:payoffPRDC} as a sum of Calls: at a time $t_k$, the payoff can be expressed as
\begin{equation*}
	\begin{aligned}
		\psi_{t_k} ( S_{t_k} ) & = \min \bigg( \max \bigg( \dfrac{C_f(t_k)}{S_0} S_{t_k} - C_d(t_k), \textrm{Floor}(t_k) \bigg),\textrm{Cap}(t_k) \bigg) \\
		                       & = \textrm{Floor}(t_k) - a_k ( S_{t_k} - K_k^1 )_+ + a_k ( S_{t_k} - K_k^2 )_+
	\end{aligned}
\end{equation*}
with $a_k = \dfrac{C_f(t_k)}{S_0}$, $K_k^1 = \dfrac{\textrm{Cap}(t_k) + C_d(t_k)}{C_f(t_k)} \times S_0$ and $K_k^2 = \dfrac{\textrm{Floor}(t_k) + C_d(t_k)}{C_f(t_k)} \times S_0$ and the closed-form formula for the price of a Call is detailed in Appendix \ref{3F:benchmark_european_call}. The prices given by the closed-form formula of the European options we consider (different values of volatilities and different maturities) are given in Table \ref{3F:tab:prices_EU}.

\begin{table}[H]
	\centering
	\begin{tabular}{c||cc}
		\toprule
		                             & \multicolumn{2}{c}{ Exact price }                \\ \midrule
		\backslashbox{$T$}{$\sigma$} & $50bp$                            & $500bp$      \\ \midrule \midrule
		2Y                           & 2.171945242                       & 2.159404007  \\ \midrule
		5Y                           & 1.630435483                       & 1.539295559  \\ \midrule
		10Y                          & 1.127330259                       & 0.8013151892 \\ \bottomrule
	\end{tabular}
	\caption[Prices given by closed-form formula of European options with zero correlations.]{\textit{Prices given by closed-form formula of European options with zero correlations. ($\sigma_d = \sigma_f = \sigma$) }}
	\label{3F:tab:prices_EU}
\end{table}

The difference of speed of convergence between the two methods is illustrated in Figures \ref{3F:fig:EU_50bp_rel_error} and \ref{3F:fig:EU_500bp_rel_error} for the relative errors for both methods compared to the benchmark. $N$ in the label of each graphic represents the size of the product quantizer ($N^X \times N^{W^f} \times N^Y \times N^{W^d}$ in the Markovian case and $N^X \times N^Y$ in the other case), hence the complexity of both trees are the same.

\begin{figure}[H]
	\centering
	\begin{subfigure}[b]{0.48\textwidth}
		\centering
		\includegraphics[width=\textwidth]{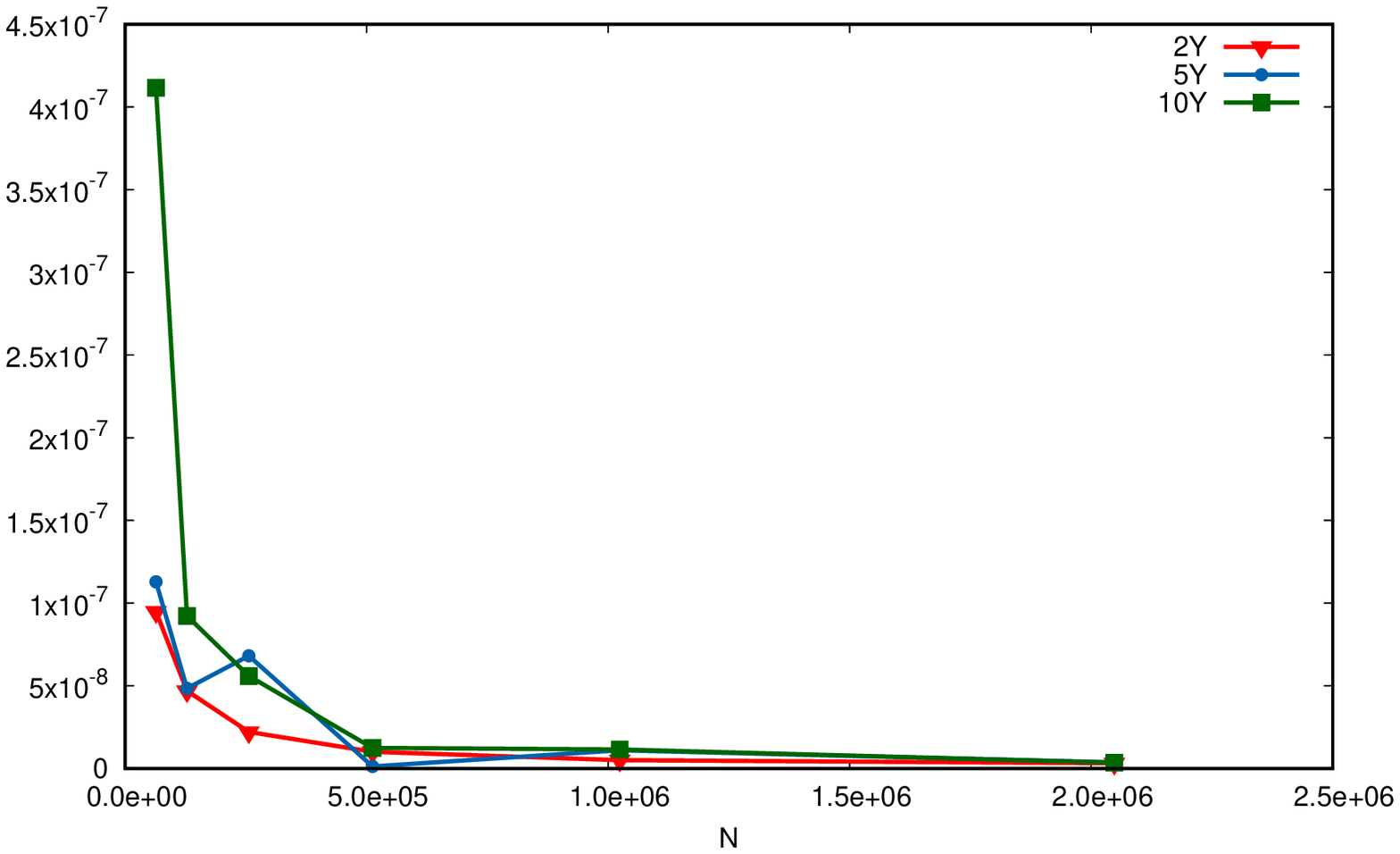}
		\caption{\textit{Non-Markovian -- 2d}}
		\label{3F:fig:EU_50bp_2d_rel_error}
	\end{subfigure}
	~
	\begin{subfigure}[b]{0.48\textwidth}
		\centering
		\includegraphics[width=\textwidth]{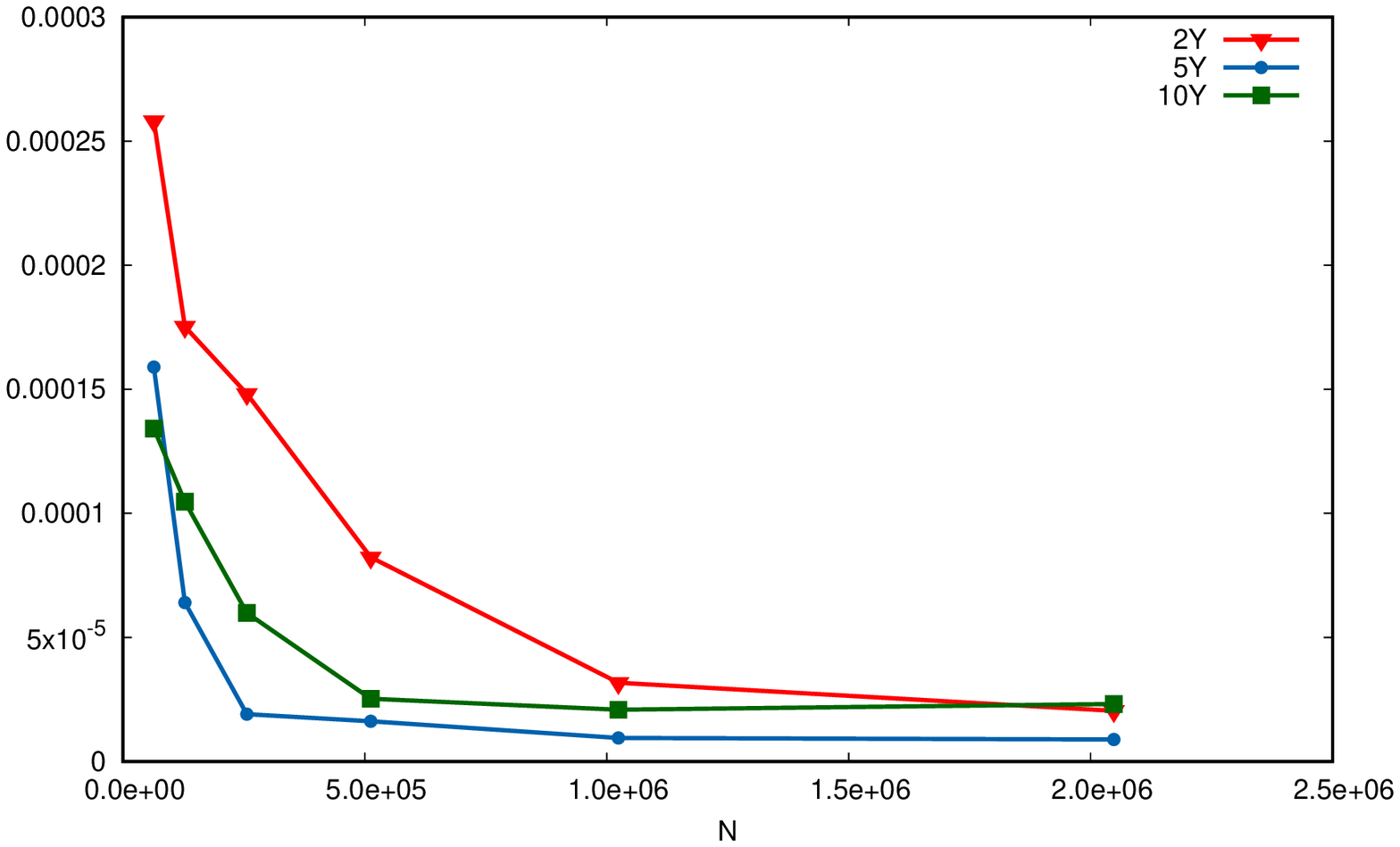}
		\caption{\textit{Markovian -- 4d}}
		\label{3F:fig:EU_50bp_4d_rel_error}
	\end{subfigure}
	\caption[Relative errors for both methods based on product quantization for 2Y, 5Y and 10Y European options pricing (with zero correlations and $\sigma_d = \sigma_f = 50bp$).]{\textit{$\sigma_d = \sigma_f = 50bp$ -- Relative errors for both methods for 2Y, 5Y and 10Y European options pricing (with zero correlations).}}
	\label{3F:fig:EU_50bp_rel_error}
\end{figure}

\begin{figure}[H]
	\centering
	\begin{subfigure}[b]{0.48\textwidth}
		\centering
		\includegraphics[width=\textwidth]{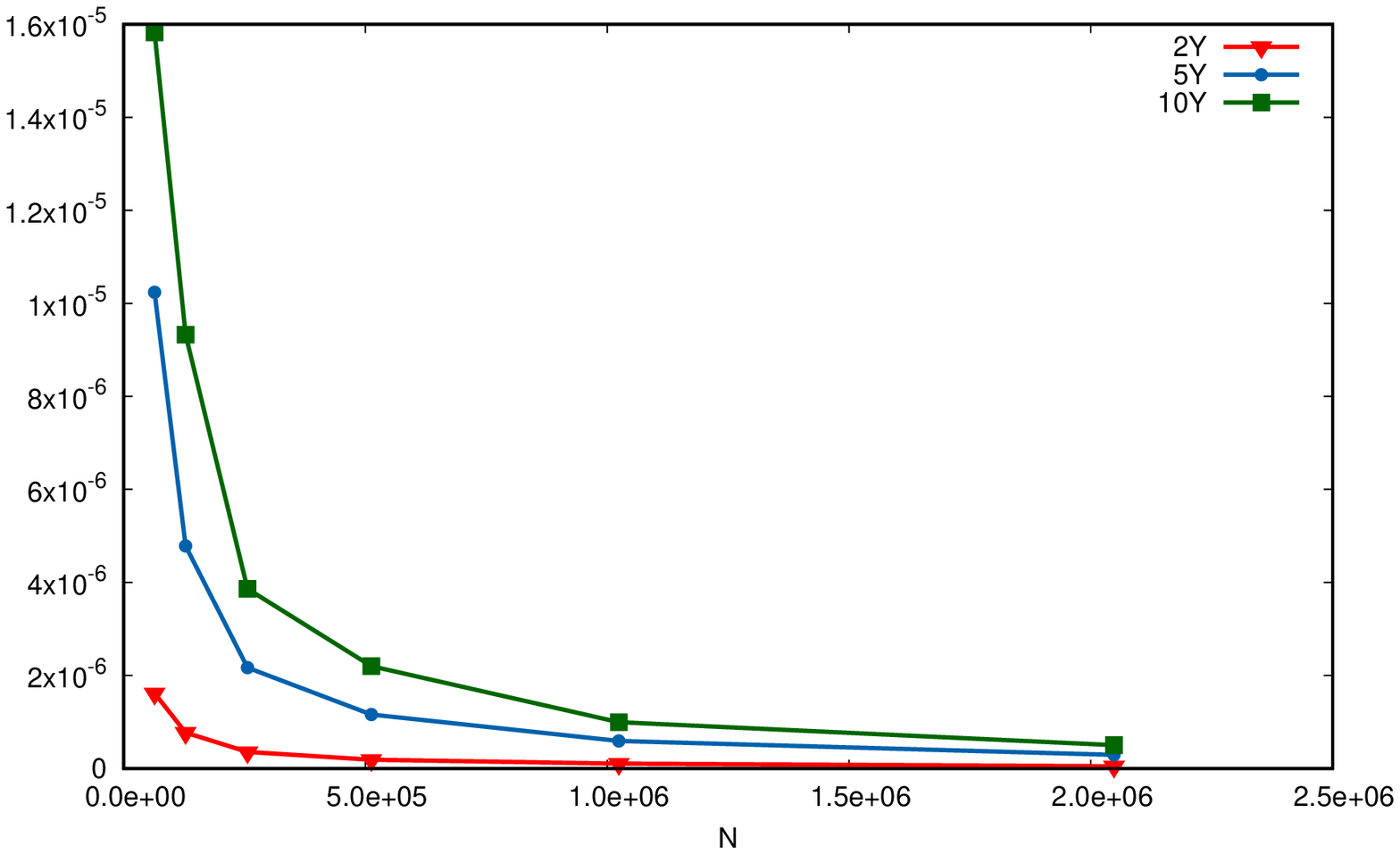}
		\caption{\textit{Non-Markovian -- 2d}}
		\label{3F:fig:EU_500bp_2d_rel_error}
	\end{subfigure}
	~
	\begin{subfigure}[b]{0.48\textwidth}
		\centering
		\includegraphics[width=\textwidth]{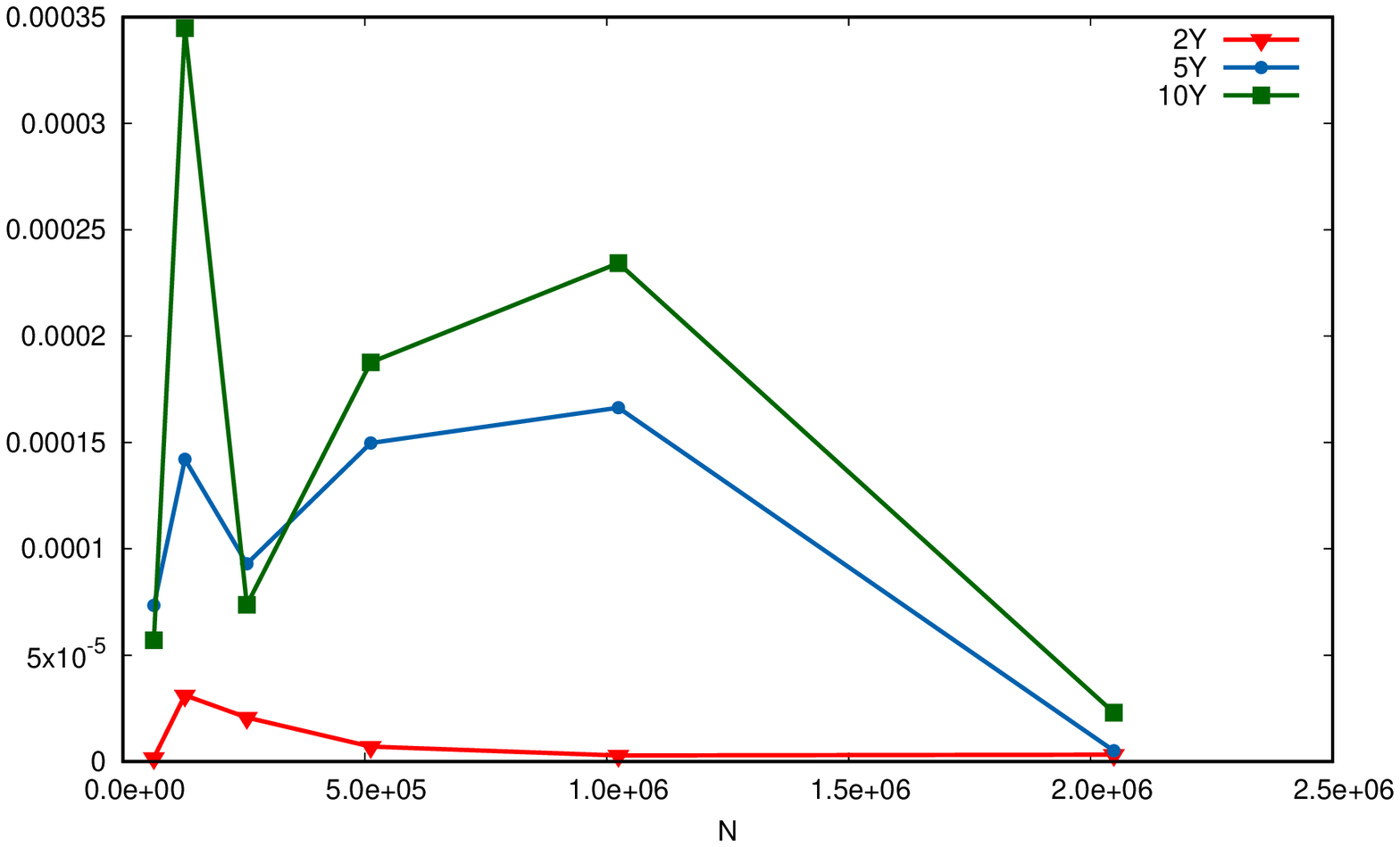}
		\caption{\textit{Markovian -- 4d}}
		\label{3F:fig:EU_500bp_4d_rel_error}
	\end{subfigure}
	\caption[Relative errors for both methods based on product quantization for 2Y, 5Y and 10Y European options pricing (with zero correlations and $\sigma_d = \sigma_f = 500bp$).]{\textit{$\sigma_d = \sigma_f = 500bp$ -- Relative errors for both methods for 2Y, 5Y and 10Y European options pricing (with zero correlations).}}
	\label{3F:fig:EU_500bp_rel_error}
\end{figure}

For both methods, a relative error of $1bp$ is quickly reached, even for high values of $\sigma_d$ and $\sigma_f$. Indeed, the time needed in order to achieve a $1bp$ precision for building a quantization tree with $2$ dates, computing the probabilities and pricing a European option is at most 6 ms for the non-Markovian method and at most 85ms for the Markovian one when the correlations are equal to zero. The computation times needed for a $1bp$ relative error are summarised in Table \ref{3F:tab:time_eu_nocorrel}.

\begin{table}[H]
	\centering
	\begin{tabular}{c||cc|cc}
		\toprule
		                             & \multicolumn{2}{c|}{ Non-Markovian -- 2d } & \multicolumn{2}{c}{ Markovian -- 4d }                                    \\ \midrule
		\backslashbox{$T$}{$\sigma$} & $50bp$                                     & $500bp$                               & $50bp$         & $500bp$         \\ \midrule \midrule
		2Y                           & 1 ms (32000)                               & 4 ms (32000)                          & 24 ms (512000) & 4 ms (64000)    \\ \midrule
		5Y                           & 4 ms (32000)                               & 6 ms (32000)                          & 4 ms (64000)   & 85 ms (2048000) \\ \midrule
		10Y                          & 4 ms (32000)                               & 3 ms (32000)                          & 14 ms (256000) & 83 ms (2048000) \\ \bottomrule
	\end{tabular}
	\caption[Computation times for European options pricing with zero correlations using both methods based on product quantization.]{\textit{Times in milliseconds needed for reaching a $1bp$ precision for European options pricing with zero correlations using both methods with, in parenthesis, the size $N$ of the grid at each time step. ($\sigma_d = \sigma_f = \sigma$) }}
	\label{3F:tab:time_eu_nocorrel}
\end{table}

\begin{remark}
	Of course, the pricers can be used even when we consider non-zero correlations. We choose to show only the asymptotic behaviour of the non-Markovian method because it converges much faster and the computations of the probabilities can be made deterministically using the CDF of a correlated bivariate normal distribution. However, if we want to use the Markovian approach, we need to compute the transition probabilities using a Monte Carlo simulation, which is a drawback for the method as it increases its computation time. We consider the following correlations
	\begin{equation*}
		\rho_{Sf} = -0.0272, \qquad \rho_{Sd} = 0.1574, \qquad \rho_{df} =0.6558.
	\end{equation*}
	Table \ref{3F:tab:prices_EU_withcorrel} summarises the prices given by the closed-form formula.
	\begin{table}[H]
		\centering
		\begin{tabular}{c||cc}
			\toprule
			                             & \multicolumn{2}{c}{ Exact price }               \\ \midrule
			\backslashbox{$T$}{$\sigma$} & $50bp$                            & $500bp$     \\ \midrule \midrule
			2Y                           & 2.173803852                       & 2.185536786 \\ \midrule
			5Y                           & 1.636518082                       & 1.652226813 \\ \midrule
			10Y                          & 1.141944391                       & 1.103531914 \\ \bottomrule
		\end{tabular}
		\caption[Prices given by closed-form formula of European options with correlations.]{\textit{Prices given by closed-form formula of European options with correlations. ($\sigma_d = \sigma_f = \sigma$) }}
		\label{3F:tab:prices_EU_withcorrel}
	\end{table}

	Figures \ref{3F:fig:EU_withcorrel_50bp_rel_error} and \ref{3F:fig:EU_withcorrel_500bp_rel_error} display the relative error induced by the numerical method as a function of $N$. And in Table \ref{3F:tab:time_eu_withcorrel}, we summarise the computation needed in order to reach a $1bp$ relative error.

	\begin{figure}[H]
		\centering
		\begin{subfigure}[b]{0.48\textwidth}
			\centering
			\includegraphics[width=\textwidth]{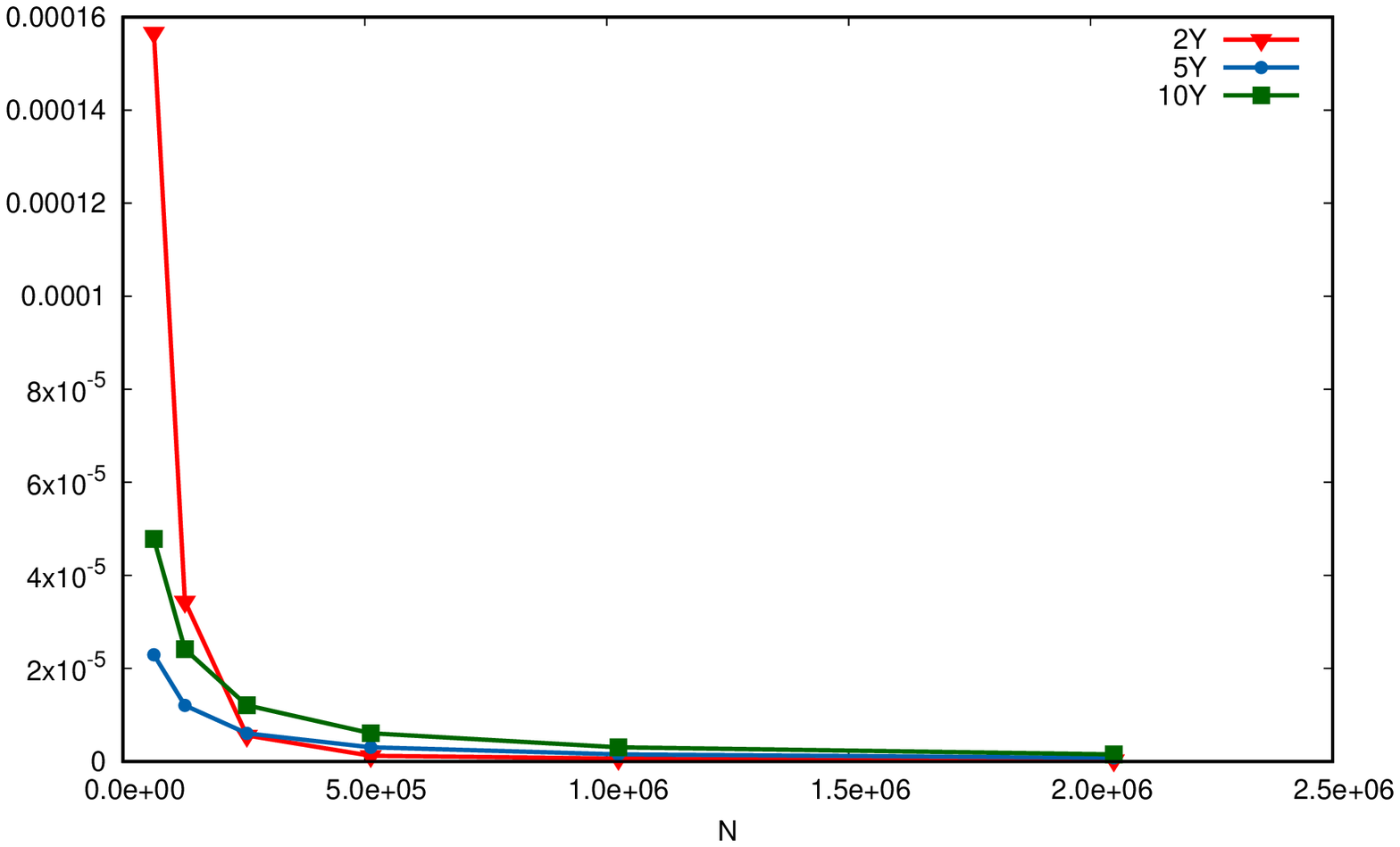}
			\caption{\textit{$\sigma_d = \sigma_f = 50bp$}}
			\label{3F:fig:EU_withcorrel_50bp_rel_error}
		\end{subfigure}
		~
		\begin{subfigure}[b]{0.48\textwidth}
			\centering
			\includegraphics[width=\textwidth]{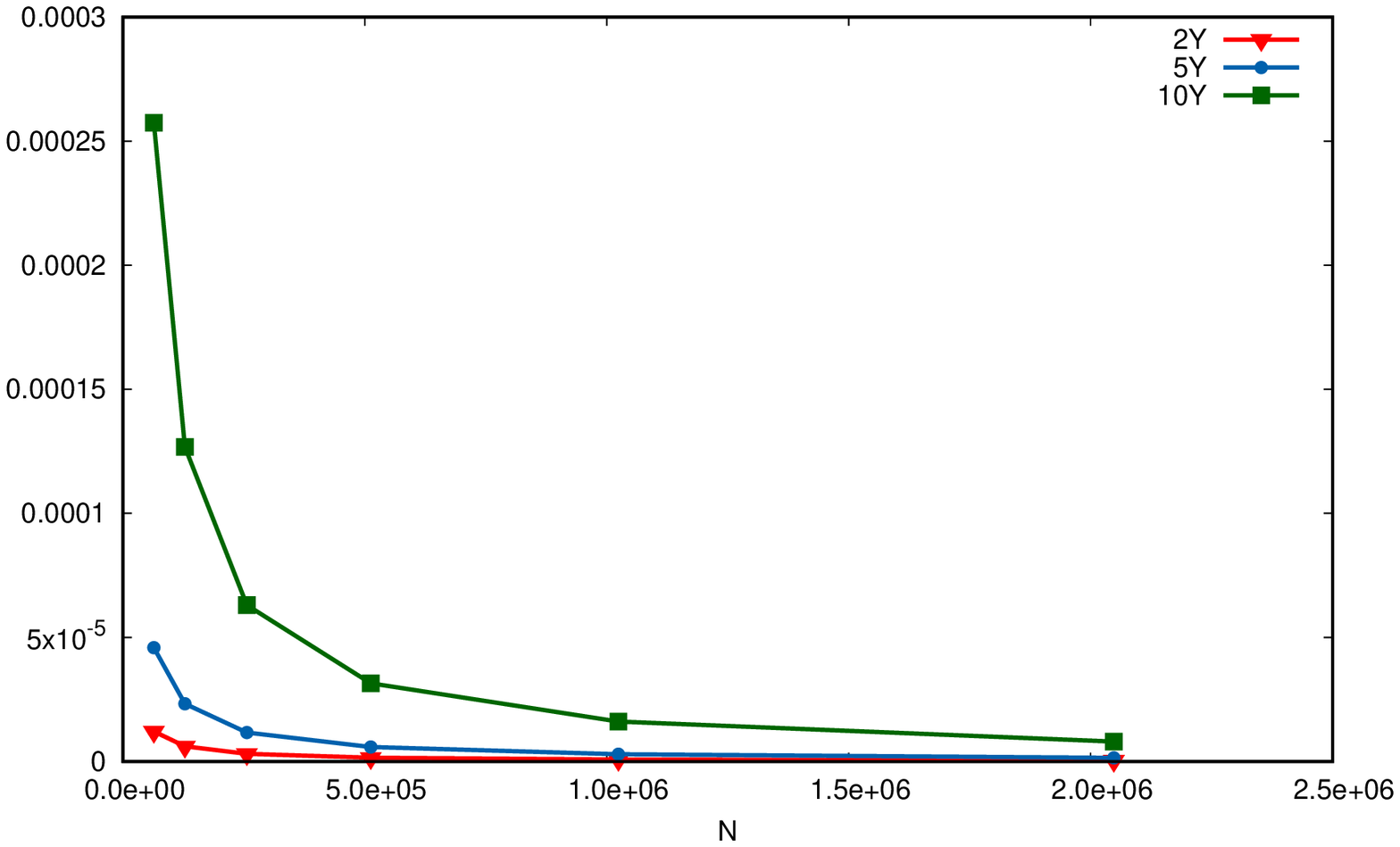}
			\caption{\textit{$\sigma_d = \sigma_f = 500bp$}}
			\label{3F:fig:EU_withcorrel_500bp_rel_error}
		\end{subfigure}
		\caption[Relative errors for the non-Markovian method for 2Y, 5Y and 10Y European options pricing (with correlations).]{\textit{Relative errors for the non-Markovian method for 2Y, 5Y and 10Y European options pricing (with correlations).}}
		\label{3F:fig:EU_withcorrel_rel_error}
	\end{figure}

	\begin{table}[H]
		\centering
		\begin{tabular}{c||cc}
			\toprule
			                             & \multicolumn{2}{c}{ Non-Markovian -- 2d }                   \\ \midrule
			\backslashbox{$T$}{$\sigma$} & $50bp$                                    & $500bp$         \\ \midrule \midrule
			2Y                           & 71 ms (64000)                             & 34 ms (32000)   \\ \midrule
			5Y                           & 31 ms (32000)                             & 31 ms (32000)   \\ \midrule
			10Y                          & 32 ms (32000)                             & 139 ms (128000) \\ \bottomrule
		\end{tabular}
		\caption[Computation times for European options pricing with correlations using the non-Markovian method.]{\textit{Times in milliseconds needed for reaching a $1bp$ relative error of European options pricing with correlations using the non-Markovian method with, in parenthesis, the size $N$ of the grid at each time step. ($\sigma_d = \sigma_f = \sigma$) }}
		\label{3F:tab:time_eu_withcorrel}
	\end{table}

\end{remark}

It is clear that one should prefer the non-Markovian methodology to the Markovian one for the evaluation of European options as it is a fast and accurate method for producing prices in the 3-factor model.

\subsection{Bermudan option}

Now, we compare the asymptotic behaviour of both approaches when pricing true Bermudan PRDC options. The following figures represent the price and the rescaled difference of the prices given by the two approaches as a function of $N$, which is the size of the product quantizer at each date (in two dimensions: $N = N^X\times N^Y$ and in four dimensions $N = N^X \times N^{W^f} \times N^Y \times N^{W^d}$). The financial products we consider are yearly exercisable Bermudan options with different values for the maturity date ($2$ years, $5$ years and $10$ years) and the domestic/foreign volatilities ($50bp$ and $500bp$).

When using domestic and foreign volatilities close to market values, we observe numerically that the non-Markovian method converges a lot faster than the Markovian one for a given complexity. However both methods do not converge to the same value (see Figures \ref{3F:fig:US_50bp_tree2_2Y}, \ref{3F:fig:US_50bp_tree5_5Y}, \ref{3F:fig:US_50bp_tree10_10Y}), which is consistent with the results we found in Theorems \ref{3F:thm:markov} and \ref{3F:thm:non_markov}. As in the European case, $N$ in the label of each graph represents the size of the product quantizer ($N^X \times N^{W^f} \times N^Y \times N^{W^d}$ in the Markovian case and $N^X \times N^Y$ in the other case), hence the complexity of both trees are the same.
\begin{figure}[H]
	\centering
	\begin{subfigure}[b]{0.48\textwidth}
		\centering
		\includegraphics[width=\textwidth]{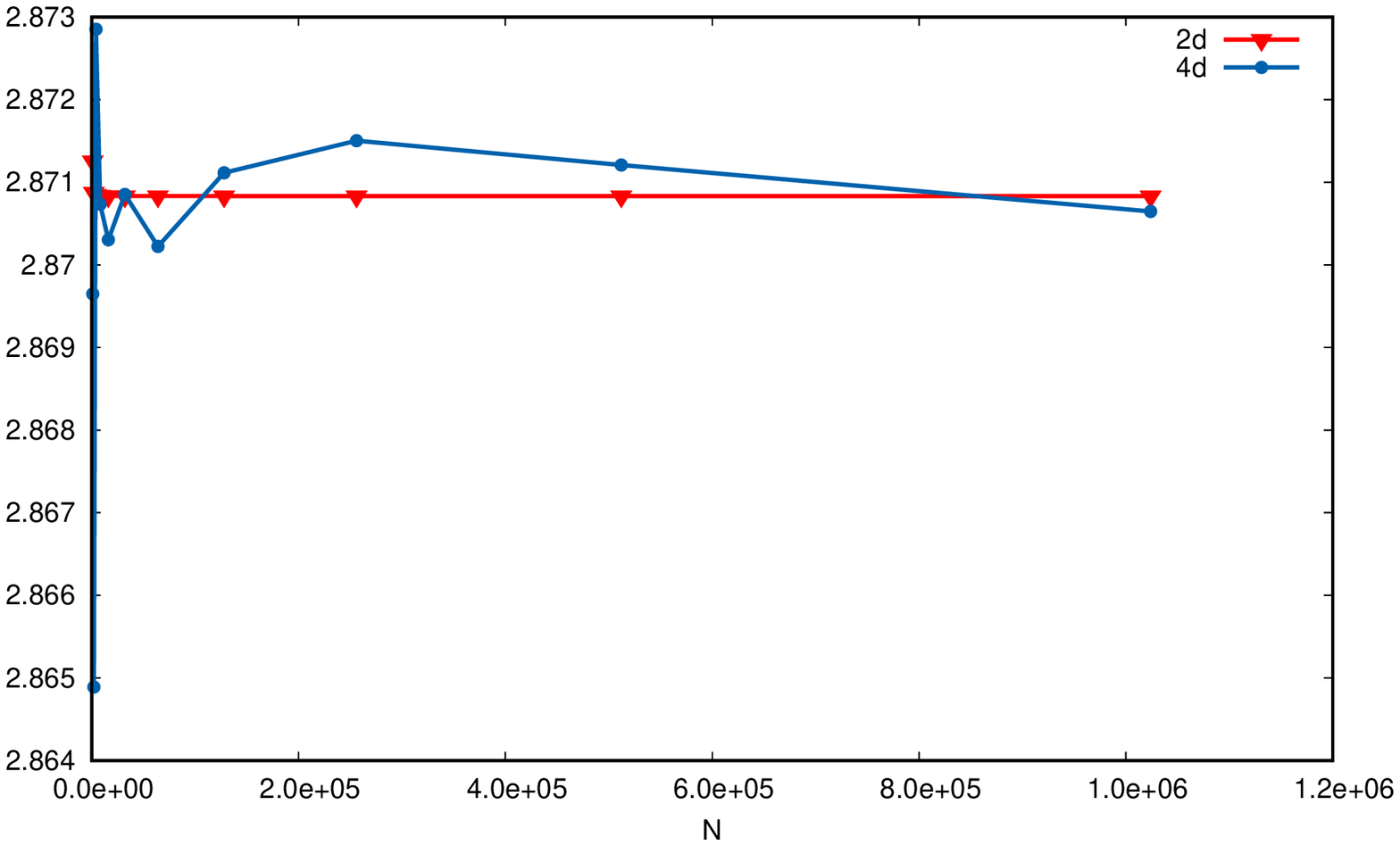}
		\caption{\textit{2Y}}
		\label{3F:fig:US_50bp_tree2_2Y}
	\end{subfigure}
	~
	\begin{subfigure}[b]{0.48\textwidth}
		\centering
		\includegraphics[width=\textwidth]{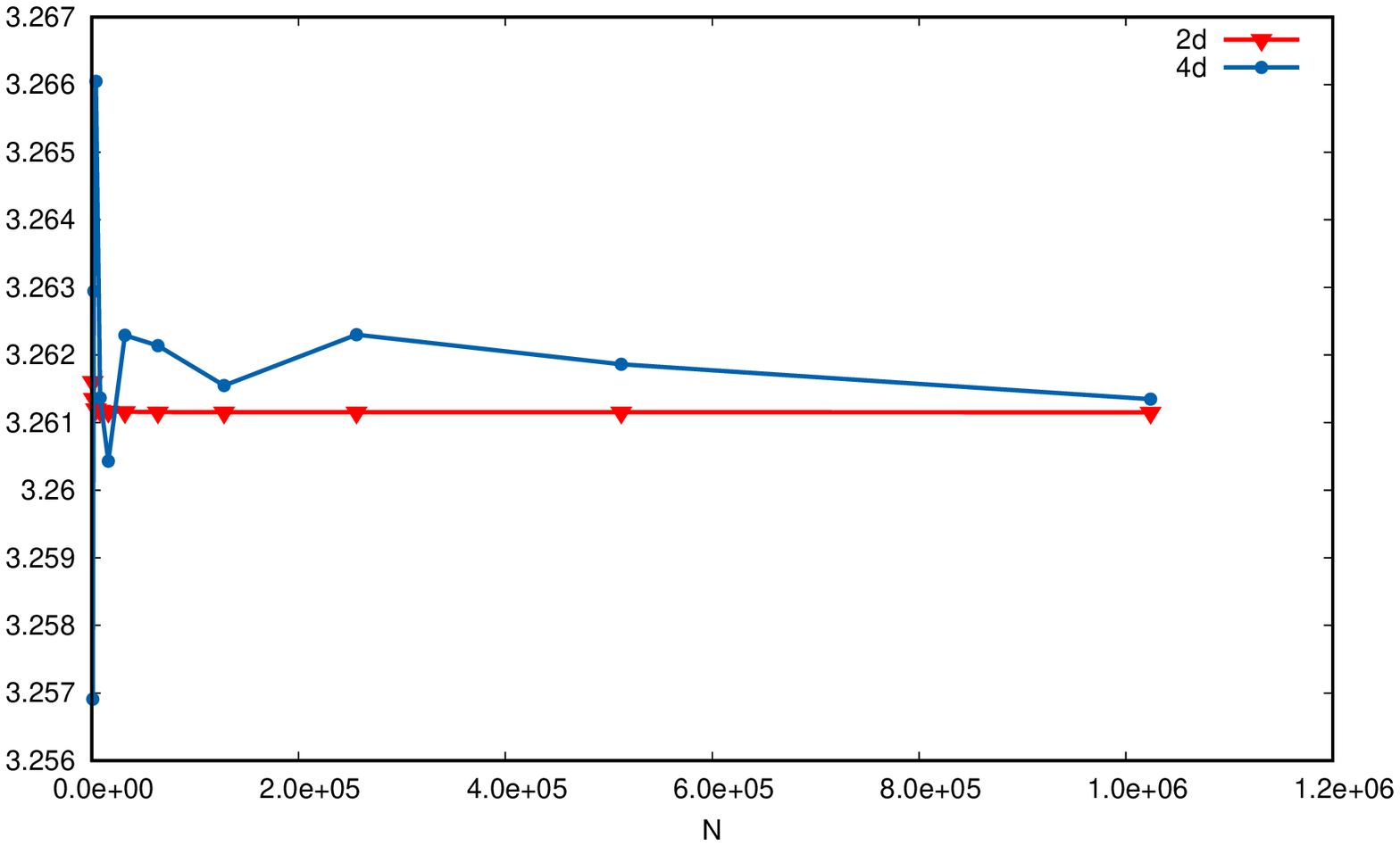}
		\caption{\textit{5Y}}
		\label{3F:fig:US_50bp_tree5_5Y}
	\end{subfigure}
	~
	\begin{subfigure}[b]{0.48\textwidth}
		\centering
		\includegraphics[width=\textwidth]{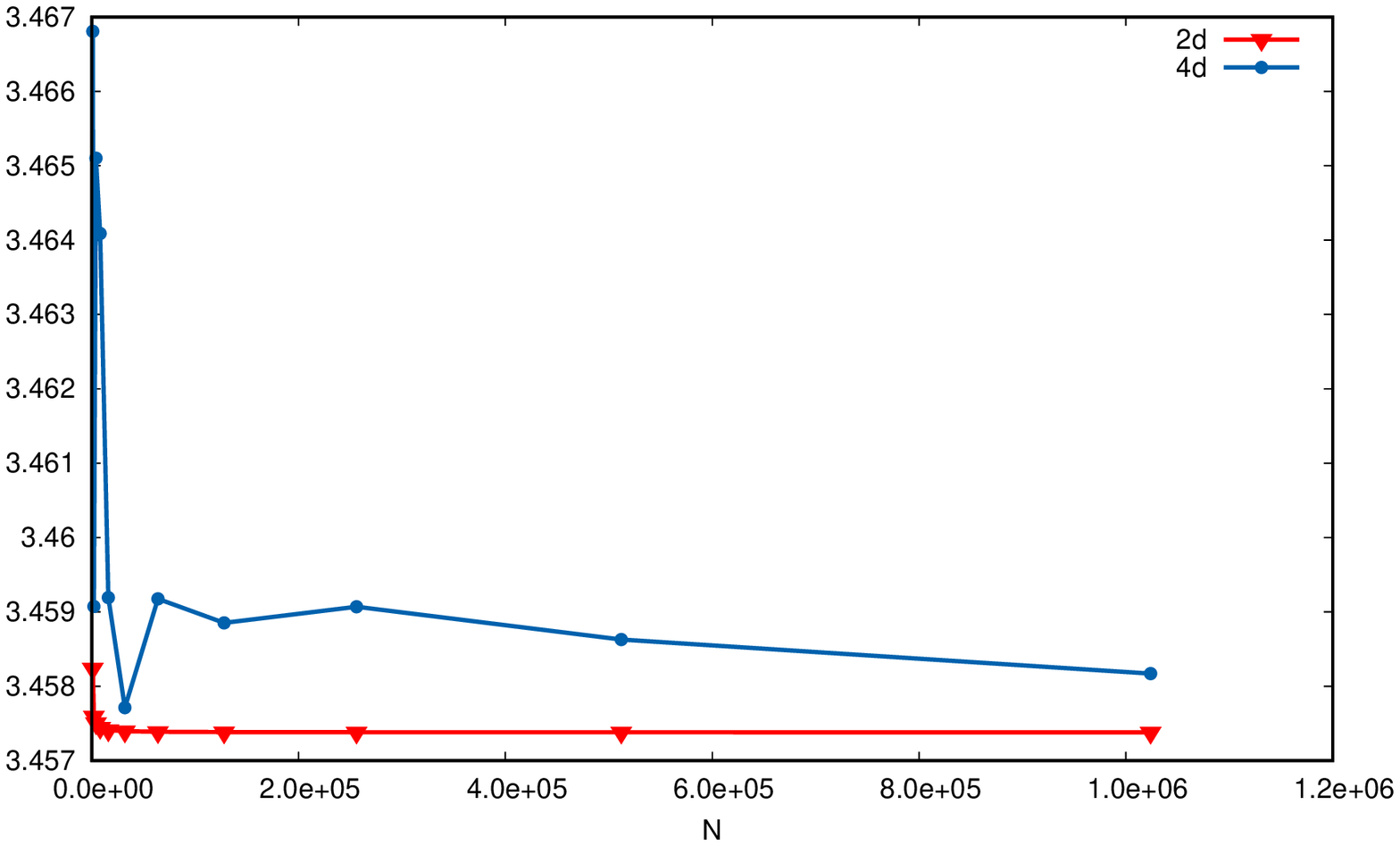}
		\caption{\textit{10Y}}
		\label{3F:fig:US_50bp_tree10_10Y}
	\end{subfigure}
	\caption[Price with the two methods based on product quantization for 2Y, 5Y and 10Y yearly exercisable Bermudan options (with zero correlations and $\sigma_d = \sigma_f = 50bp$).]{\textit{$\sigma_d = \sigma_f = 50bp$ -- Price with the two methods for 2Y, 5Y and 10Y yearly exercisable Bermudan options (with zero correlations).}}
	\label{3F:fig:US_50bp_prices}
\end{figure}

However, the relative systematic error induced by the non-Markovian methodology is negligible as can be seen in Figure \ref{3F:fig:US_50bp_relerrors}, at most $5bp$ for a 10-year annual Bermudan option. Hence, one should prefer, again, the non-Markovian methodology when considering to evaluate Bermudan options.

\begin{figure}[H]
	\centering
	\includegraphics[width=.6\textwidth]{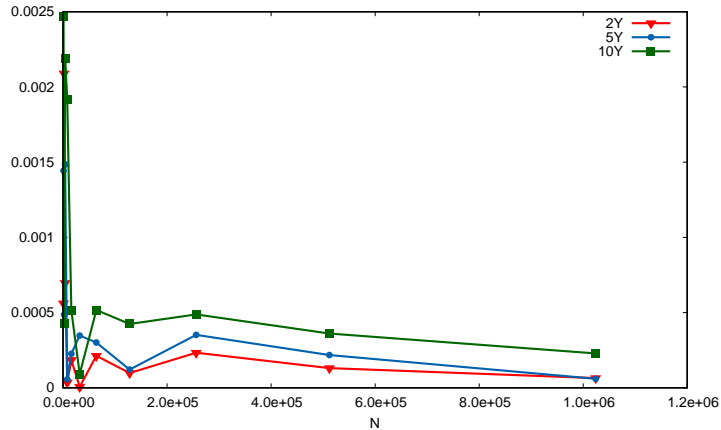}
	\caption[Relative differences between the two methods based on product quantization for 2Y, 5Y and 10Y yearly exercisable Bermudan options (with zero correlations and $\sigma_d = \sigma_f = 50bp$).]{\textit{$\sigma_d = \sigma_f = 50bp$ -- Relative differences between the two methods for 2Y, 5Y and 10Y yearly exercisable Bermudan options (with zero correlations).}}
	\label{3F:fig:US_50bp_relerrors}
\end{figure}

\begin{remark}
	If we consider more exercise dates for the Bermudan option, the systematic errors increase, as shown in Figure \ref{3F:fig:US_50bp_relerrors_biannual} where we considered Bermudan options exercisable every $6$ months and the same parameters as before with zero correlations and $\sigma_d = \sigma_f = 50bp$. However, even-though the error is higher for small $N$, when the non-Markovian pricer has converged, the relative difference between both methods is still acceptable (lower than $5bp$).
	\begin{figure}[H]
		\centering
		\includegraphics[width=.6\textwidth]{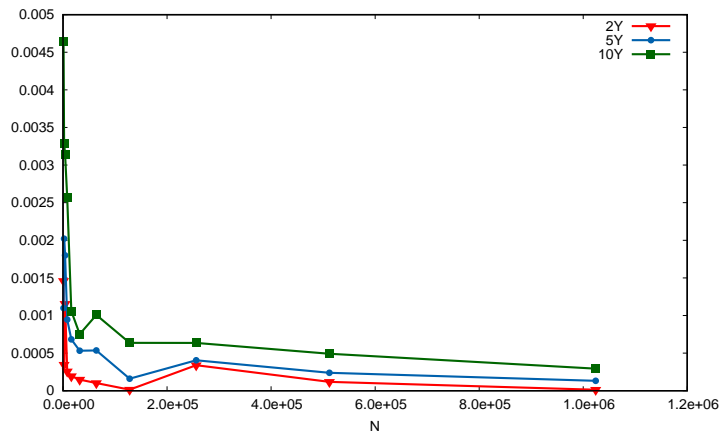}
		\caption[Relative differences between the two methods based on product quantization for 2Y, 5Y and 10Y bi-annual exercisable Bermudan options (with zero correlations and $\sigma_d = \sigma_f = 50bp$).]{\textit{$\sigma_d = \sigma_f = 50bp$ -- Relative differences between the two methods for 2Y, 5Y and 10Y bi-annual exercisable Bermudan options (with zero correlations).}}
		\label{3F:fig:US_50bp_relerrors_biannual}
	\end{figure}
\end{remark}

When we consider higher values the volatilities, $\sigma_d = \sigma_f = 500bp$, as expected the prices produced by the non-Markovian methodology produce a systematic error bigger than the case where $\sigma_d = \sigma_f = 50bp$ (see Figures \ref{3F:fig:US_500bp_tree2_2Y}, \ref{3F:fig:US_500bp_tree5_5Y}, \ref{3F:fig:US_500bp_tree10_10Y} and \ref{3F:fig:US_500bp_relerrors}). However, the relative difference between the two methods after convergence is reasonable: less than $0.1\%$ for expiry 2 years, $0.4\%$ for 5 years and around $1.1\%$ for $10$ years.

\begin{figure}[H]
	\centering
	\begin{subfigure}[b]{0.48\textwidth}
		\centering
		\includegraphics[width=\textwidth]{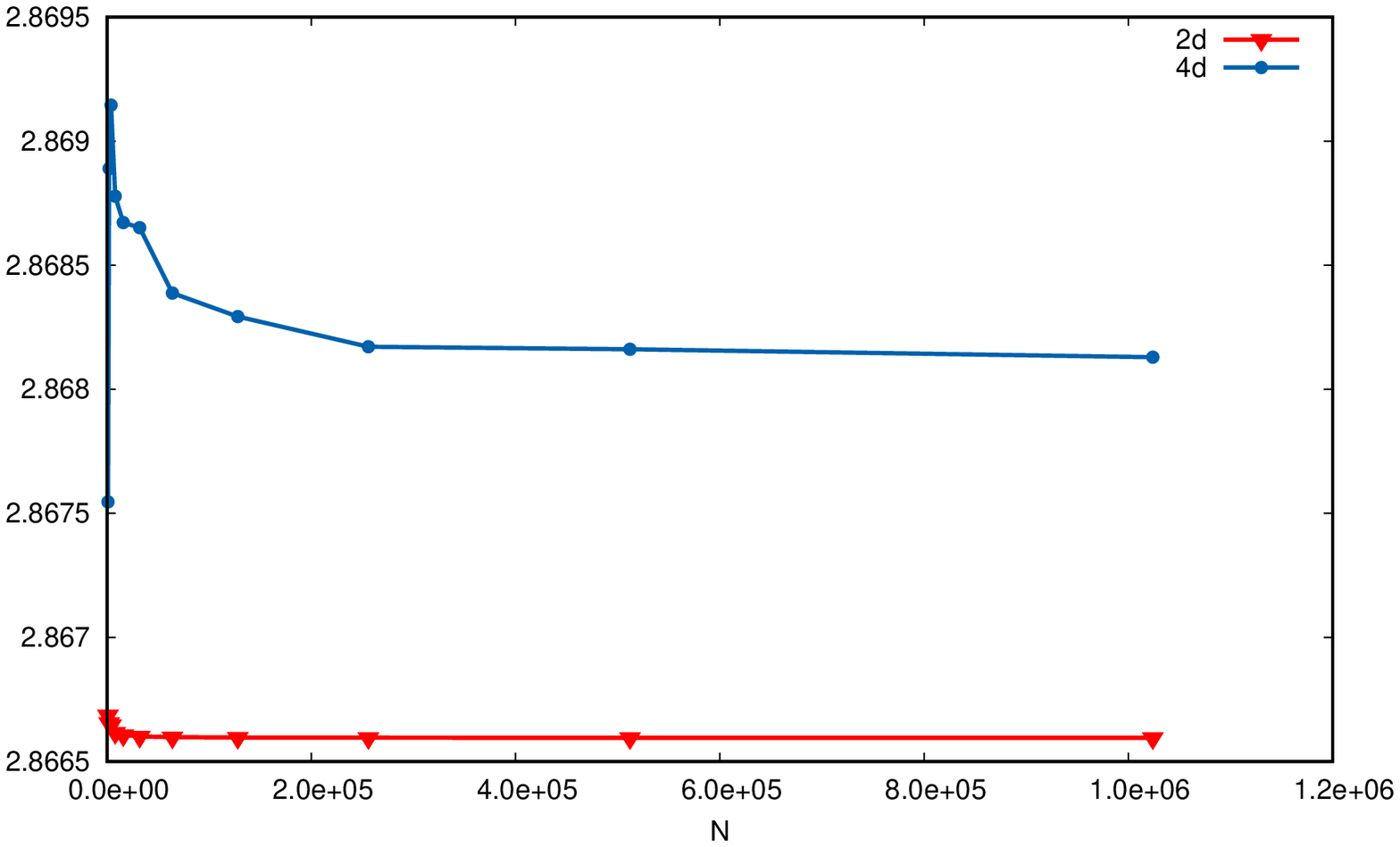}
		\caption{\textit{2Y}}
		\label{3F:fig:US_500bp_tree2_2Y}
	\end{subfigure}
	~
	\begin{subfigure}[b]{0.48\textwidth}
		\centering
		\includegraphics[width=\textwidth]{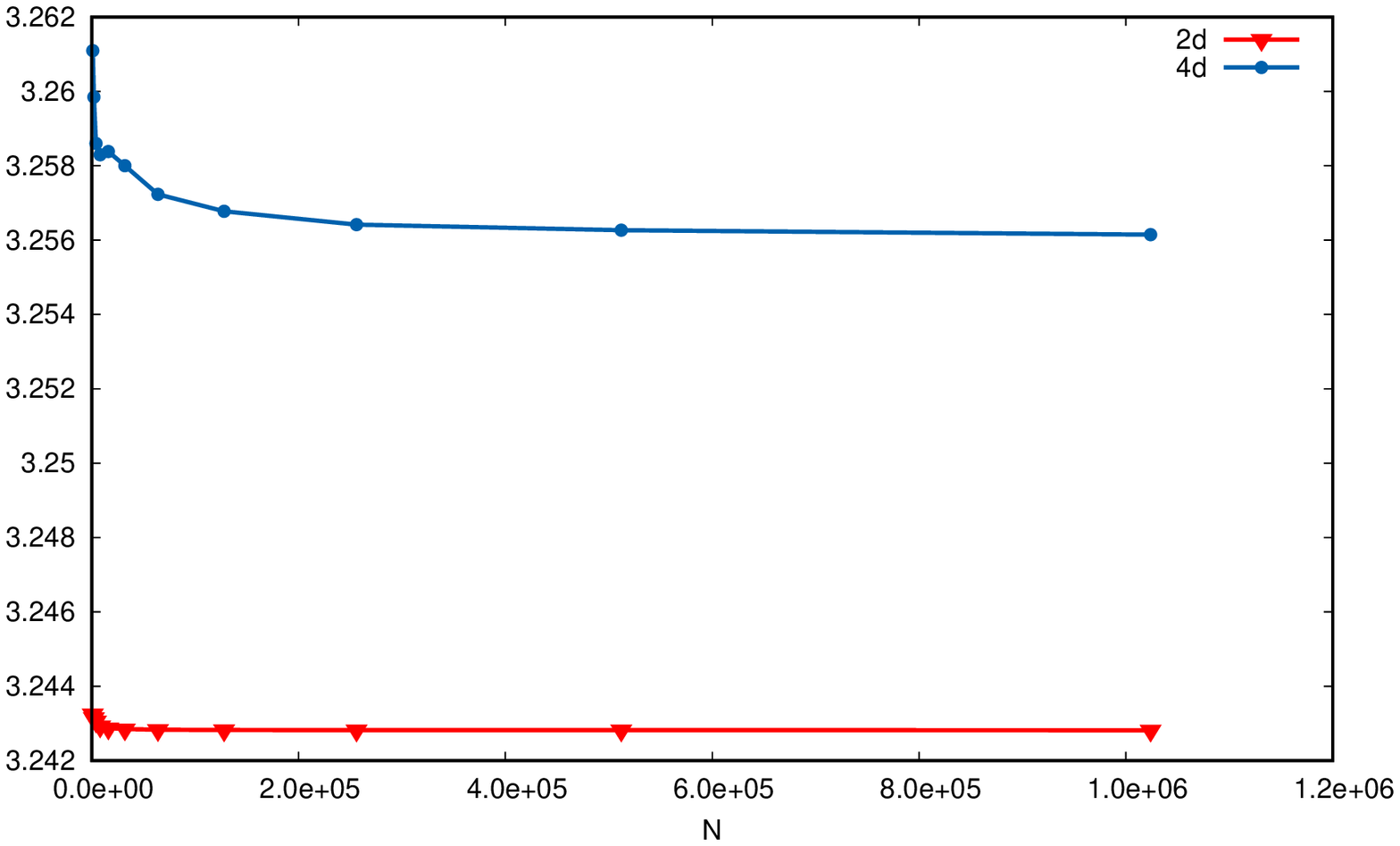}
		\caption{\textit{5Y}}
		\label{3F:fig:US_500bp_tree5_5Y}
	\end{subfigure}
	~
	\begin{subfigure}[b]{0.48\textwidth}
		\centering
		\includegraphics[width=\textwidth]{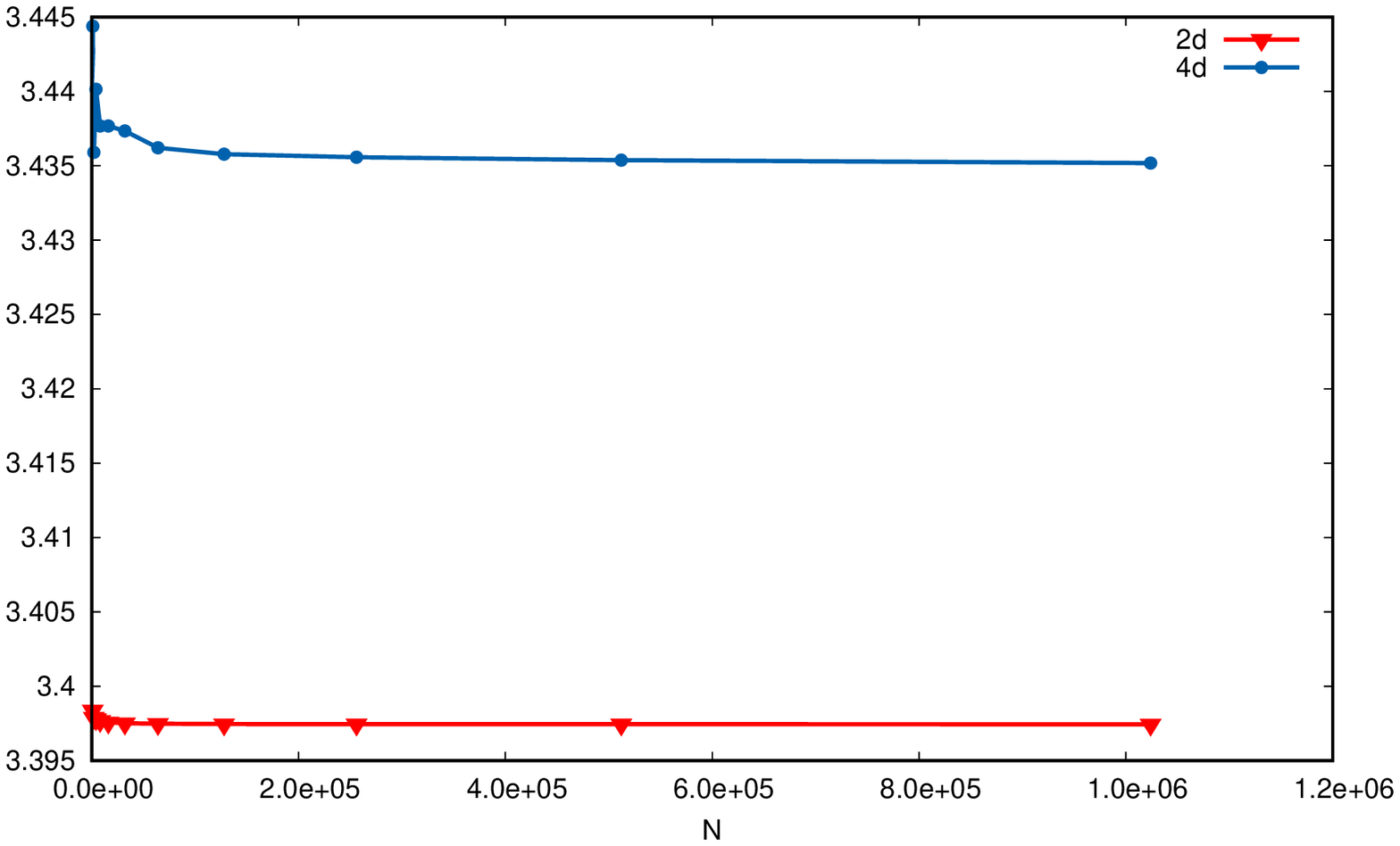}
		\caption{\textit{10Y}}
		\label{3F:fig:US_500bp_tree10_10Y}
	\end{subfigure}
	\caption[Price with the two methods based on product quantization for 2Y, 5Y and 10Y yearly exercisable Bermudan options (with zero correlations and $\sigma_d = \sigma_f = 500bp$).]{\textit{$\sigma_d = \sigma_f = 500bp$ -- Price with the two methods for 2Y, 5Y and 10Y yearly exercisable Bermudan options (with zero correlations).}}
	\label{3F:fig:US_500bp_prices}
\end{figure}

\begin{figure}[H]
	\centering
	\includegraphics[width=.6\textwidth]{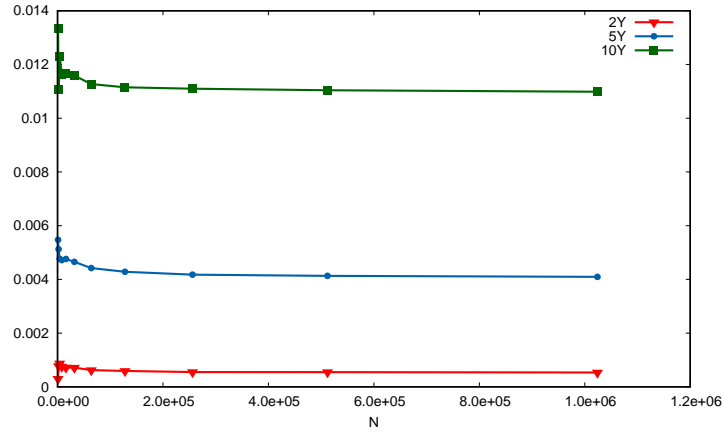}
	\caption[Relative differences between the two methods based on product quantization for 2Y, 5Y and 10Y yearly exercisable Bermudan options (with zero correlations and $\sigma_d = \sigma_f = 500bp$).]{\textit{$\sigma_d = \sigma_f = 500bp$ -- Relative differences between the two methods for 2Y, 5Y and 10Y yearly exercisable Bermudan options (with zero correlations).}}
	\label{3F:fig:US_500bp_relerrors}
\end{figure}

In Figure \ref{3F:tab:time_us_nocorrel}, we reference the time needed for reaching a $5bp$ relative precision (we compare the price given by grids of size $N$ to the "asymptotic", which is the price given by the same method with  a very large $N$) for the pricing of Bermudan options in a scenario of zero correlations. The non-Markovian method attains better  precision than a relative precision of $5bp$ in a few milliseconds, at most 7 ms where the Markovian one can need 4 seconds for reaching that precision. Hence, the 2 dimensional approximation seems again to be the better choice.

\begin{table}[H]
	\centering
	\begin{tabular}{c||cc|cc}
		\toprule
		                             & \multicolumn{2}{c|}{ Non-Markovian -- 2d } & \multicolumn{2}{c}{ Markovian -- 4d }                                    \\ \midrule
		\backslashbox{$T$}{$\sigma$} & $50bp$                                     & $500bp$                               & $50bp$         & $500bp$         \\
		\midrule \midrule
		2Y                           & 1 ms (1000)                                & 1 ms (1000)                           & 25 ms (8000)   & 4 ms (1000)     \\ \midrule
		5Y                           & 3 ms (1000)                                & 4 ms (1000)                           & 98 ms (8000)   & 1903 ms (64000) \\ \midrule
		10Y                          & 7 ms (1000)                                & 7 ms (1000)                           & 468 ms (16000) & 3850 ms (64000) \\ \bottomrule
	\end{tabular}
	\caption[Computation times for Bermudan yearly exercisable options pricing with zero correlations using both methods based on product quantization.]{\textit{Times in milliseconds needed for reaching a $5bp$ relative precision for Bermudan options pricing using both methods with zero correlation and, in parenthesis, the size $N$ of the grid at each time step. ($\sigma_d = \sigma_f = \sigma$) }}
	\label{3F:tab:time_us_nocorrel}
\end{table}

\begin{remark}
	Again, the pricers can even be used when we consider non-zero correlations and we choose to show only the asymptotic behaviour of the non-Markovian method, for the same reasons as the European case. We consider the same correlations as in the European case
	\begin{equation*}
		\rho_{Sf} = -0.0272, \qquad \rho_{Sd} = 0.1574, \qquad \rho_{df} =0.6558.
	\end{equation*}
	Figures \ref{3F:fig:US_withcorrel_50bp_tree2_2Y}, \ref{3F:fig:US_withcorrel_50bp_tree5_5Y} and \ref{3F:fig:US_withcorrel_50bp_tree10_10Y} display the price given by the numerical method as a function of $N$ and Table \ref{3F:tab:time_us_withcorrel} summarises the computation time needed in order to do better than a $3bp$ precision.

	\begin{figure}[H]
		\centering
		\begin{subfigure}[b]{0.48\textwidth}
			\centering
			\includegraphics[width=\textwidth]{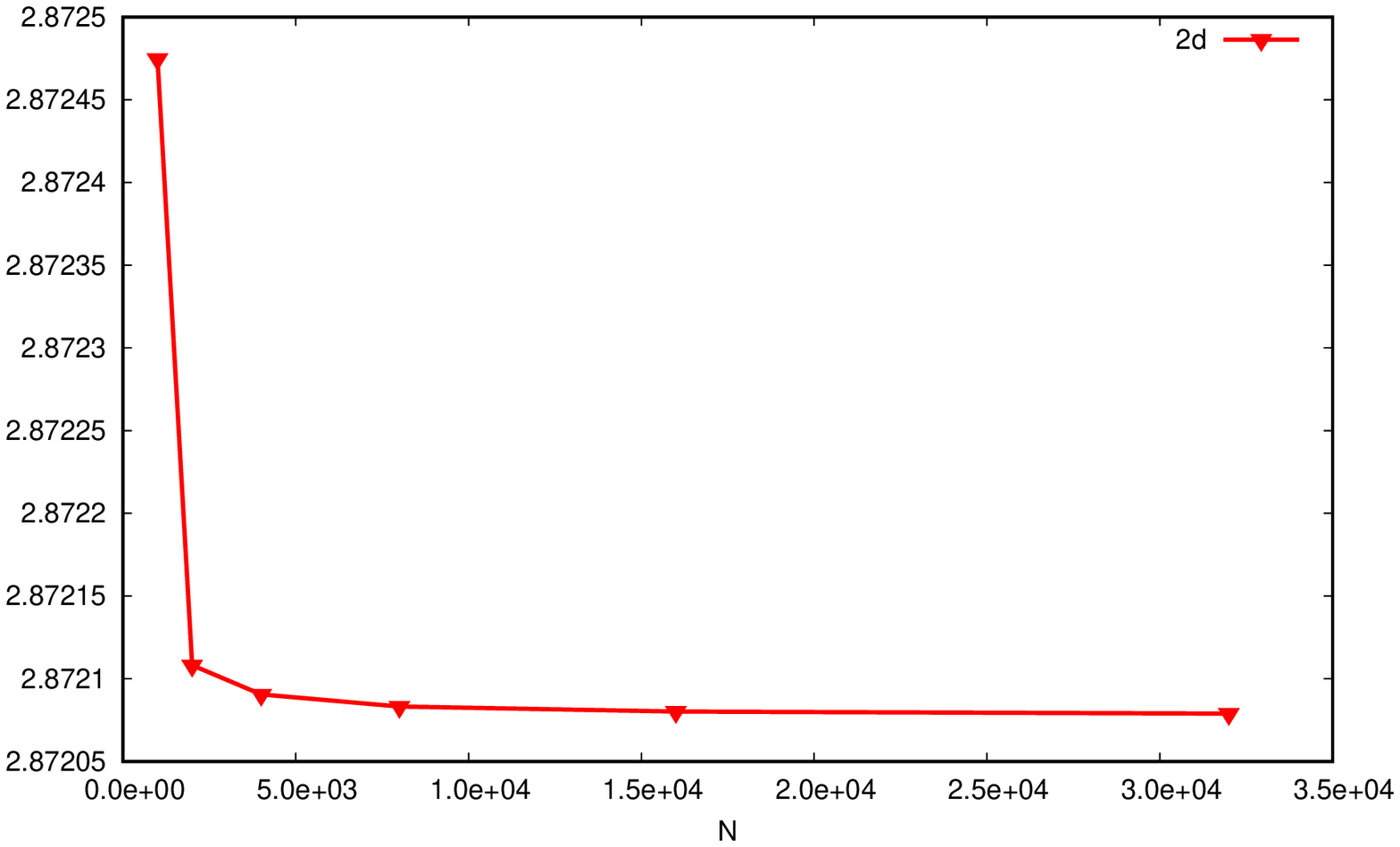}
			\caption{\textit{2Y}}
			\label{3F:fig:US_withcorrel_50bp_tree2_2Y}
		\end{subfigure}
		~
		\begin{subfigure}[b]{0.48\textwidth}
			\centering
			\includegraphics[width=\textwidth]{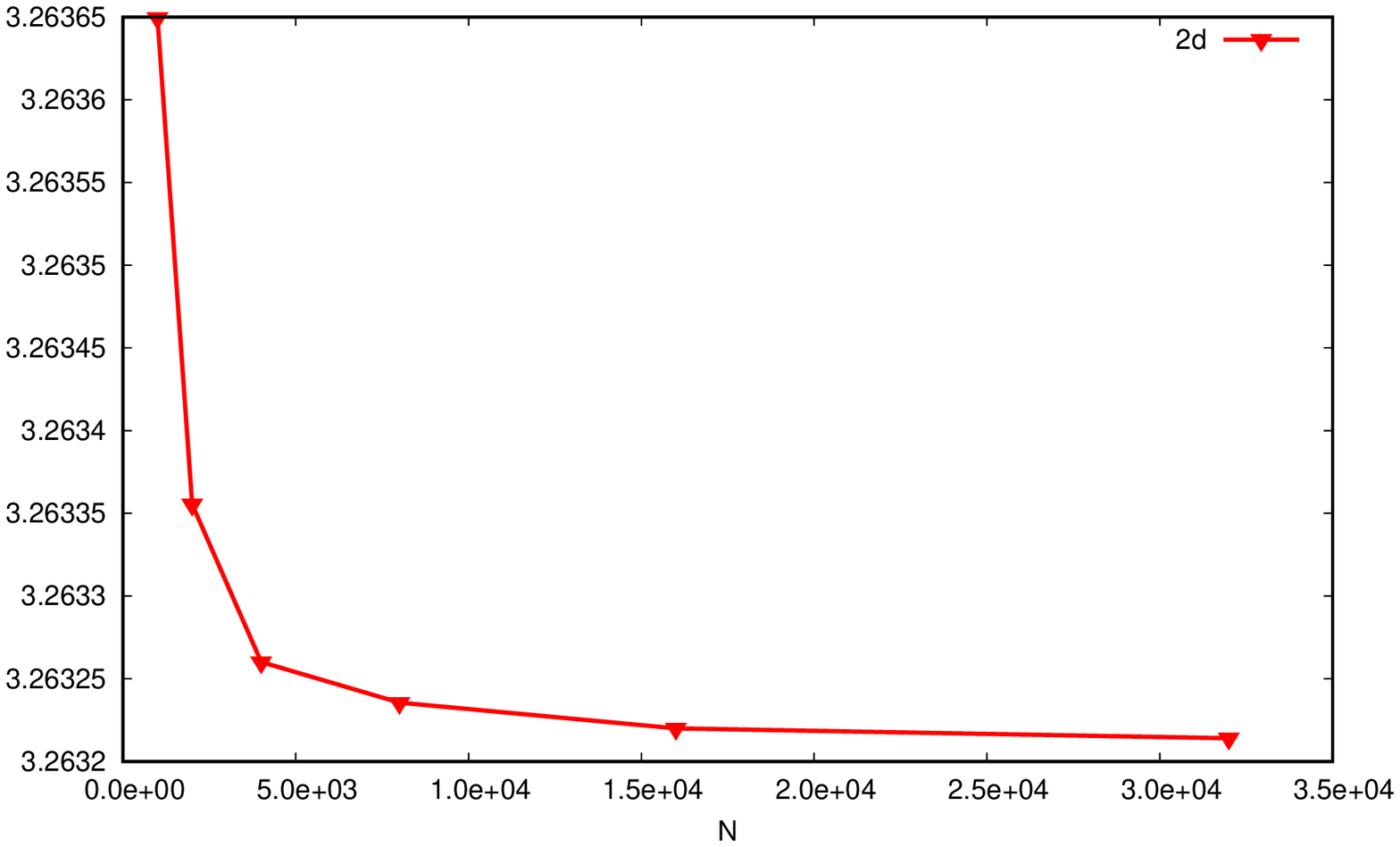}
			\caption{\textit{5Y}}
			\label{3F:fig:US_withcorrel_50bp_tree5_5Y}
		\end{subfigure}
		~
		\begin{subfigure}[b]{0.48\textwidth}
			\centering
			\includegraphics[width=\textwidth]{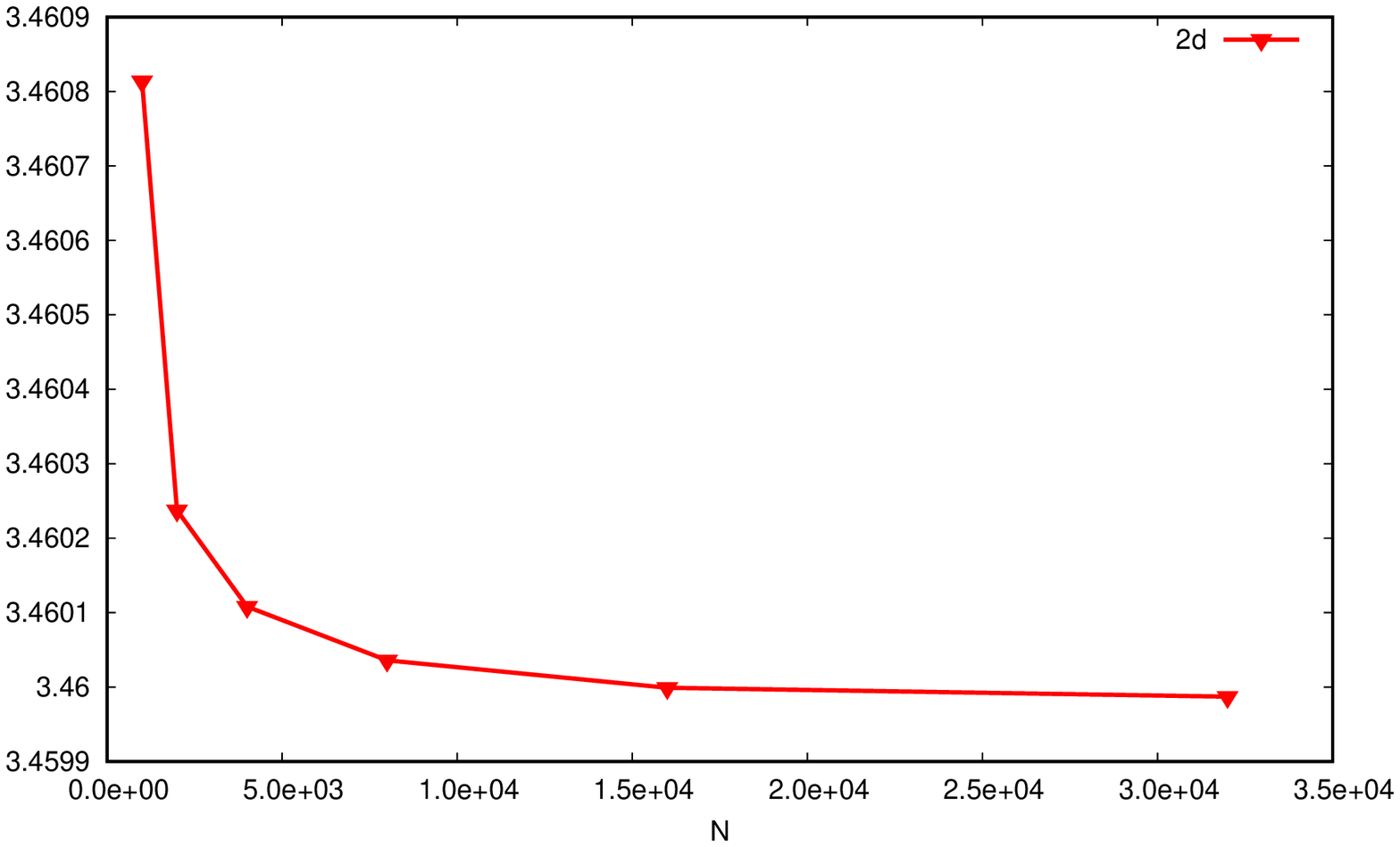}
			\caption{\textit{10Y}}
			\label{3F:fig:US_withcorrel_50bp_tree10_10Y}
		\end{subfigure}
		\caption[Price of 2Y, 5Y and 10Y yearly exercisable Bermudan options using the non-Markovian method (with zero correlations and $\sigma_d = \sigma_f = 50bp$).]{\textit{$\sigma_d = \sigma_f = 50bp$ -- Price of 2Y, 5Y and 10Y yearly exercisable Bermudan options using the non-Markovian method (with correlations).}}
		\label{3F:fig:US_withcorrel_50bp}
	\end{figure}

	\begin{table}[H]
		\centering
		\begin{tabular}{c||c}
			\toprule
			                             & Non-Markovian -- 2d \\ \midrule
			\backslashbox{$T$}{$\sigma$} & $50bp$              \\ \midrule \midrule
			2Y                           & 122 ms (1000)       \\ \midrule
			5Y                           & 553 ms (1000)       \\ \midrule
			10Y                          & 1283 ms (1000)      \\ \bottomrule
		\end{tabular}
		\caption[Computation times for Bermudan yearly exercisable options pricing with correlations using the non-Markovian method.]{\textit{Times in milliseconds needed for reaching a $3bp$ relative precision for Bermudan yearly exercisable options pricing with correlations using the non-Markovian method with, in parenthesis, the size $N$ of the grid at each time step. ($\sigma_d = \sigma_f = \sigma$) }}
		\label{3F:tab:time_us_withcorrel}
	\end{table}

\end{remark}

\section*{Conclusion}

We were looking for a numerical method able to produce accurate prices of Bermudan PRDC options with a 3-factor model in a very short time because the pricing of such products arises in a more complex framework: the computation of counterparty risk measures, also called xVA's.
We proposed two numerical methods based on product optimal quantization with a preference for the non-Markovian one. Indeed, even if we introduce a systematic error with our approximation, the error is controlled, as long as the volatilities of the domestic and foreign interest rates stay reasonable. Moreover, the numerical tests we conducted confirmed that idea: we are able to produce prices of Bermudan options in the 3-factor model in a fast and accurate way.

\section*{Declaration of Interest}
The author reports no conflicts of interest. The author alone is responsible for the content and writing of the paper.

\section*{Acknowledgments}
The PhD thesis of Thibaut Montes is funded by a CIFRE grand from The Independent Calculation Agent (ICA) and French ANRT.

\nocite{*}
\bibliography{bibli}
\bibliographystyle{alpha}

\newpage

\begin{appendices}

	\section{$W^f$ is a Brownian motion under the domestic risk-neutral measure} \label{3F:section:proof_foreignBM}

	Let $(\widetilde W^f)$ a $\widetilde \Prob$-Brownian motion. In this section, we show that the process $W^f$ defined by
	\begin{equation}
		d W^f_s = d \widetilde W^f_s + \rho_{Sf} \sigma_S ds
	\end{equation}
	is a $\Prob$-Brownian motion.

	First, we define the following change of numéraire, where $\widetilde \Prob$ is the foreign risk-neutral probability and $\Prob$ is the domestic risk-neutral probability,r
	\begin{equation*}
		\begin{aligned}
			d \widetilde \Prob
			 & = \frac{S_T}{ S_0 } \exp \bigg(- \int_0^T r_s^d ds \bigg) \exp \bigg( \int_0^T r_s^f ds \bigg) d\Prob \\
			 & = \exp \bigg( \sigma_S W_T^S - \frac{\sigma_S^2}{2} T \bigg) d\Prob                                   \\
		\end{aligned}
	\end{equation*}
	or equivalently
	\begin{equation}
		\begin{aligned}
			d \Prob
			 & = \exp \bigg( - \sigma_S W_T^S + \frac{\sigma_S^2}{2} T \bigg) d \widetilde \Prob                  \\
			 & = \exp \bigg( - \sigma_S ( W_T^S - \sigma_S T ) - \frac{\sigma_S^2}{2} T \bigg) d \widetilde \Prob \\
			 & = \exp \bigg( - \sigma_S \widetilde W_T^S - \frac{\sigma_S^2}{2} T \bigg) d\widetilde \Prob
		\end{aligned}
	\end{equation}
	where $\widetilde W^S$ is a $\widetilde \Prob$-Brownian motion defined by $d \widetilde W_t^S = d W_t^S - \sigma_S dt$. More details concerning the definition of the foreign risk-neutral probability can be found in the Chapter 9 of \cite{shreve2004stochastic}.

	Now, we are looking for $q \in \R$ such that $d W^f_s = d \widetilde W^f_s - q dt$ is a $\Prob$-Brownian motion. Let $\lambda \in \R$ and $\forall t>s$
	\begin{equation}
		\begin{aligned}
			\E \Big[ \e^{\lambda \big( ( \widetilde W^f_t - q t) - ( \widetilde W^f_s - q s) \big) } \mid \F_s \Big]
			 & = \widetilde \E \Big[ \e^{\lambda \big( ( \widetilde W^f_t - q t) - ( \widetilde W^f_s - q s) \big) - \sigma_S ( \widetilde W_T^S - \widetilde W_s^S) - \frac{\sigma_S^2}{2} (T-s) } \mid \F_s \Big]  \\
			 & = \widetilde \E \Big[ \e^{\lambda \big( ( \widetilde W^f_t - q t) - ( \widetilde W^f_s - q s) \big) - \sigma_S ( \widetilde W_t^S - \widetilde W_s^S) - \frac{\sigma_S^2}{2} (t-s) } \mid \F_s \Big]  \\
			 & = \e^{ - \lambda q (t-s) - \frac{\sigma_S^2}{2} (t-s) } \widetilde \E \Big[ \e^{ \lambda ( \widetilde W^f_t - \widetilde W^f_s ) - \sigma_S ( \widetilde W_t^S - \widetilde W_s^S ) } \mid \F_s \Big] \\
			 & = \e^{ - \lambda q (t-s) - \frac{\sigma_S^2}{2} (t-s) } \e^{ \frac{\lambda^2 }{2}(t-s) - \lambda \sigma_S \rho_{Sf} (t-s) + \frac{\sigma_S^2}{2}(t-s) }                                               \\
			 & = \e^{ \frac{\lambda^2 }{2}(t-s) } \e^{ - \lambda q (t-s) - \lambda \sigma_S \rho_{Sf} (t-s)}                                                                                                         \\
			 & = \e^{ \frac{\lambda^2 }{2}(t-s) }                                                                                                                                                                    \\
		\end{aligned}
	\end{equation}
	the last equality is ensured if and only if
	\begin{equation}
		\begin{aligned}
			0 = - \lambda q (t-s) - \lambda \sigma_S \rho_{Sf} (t-s) \quad \iff \quad q = - \sigma_S \rho_{Sf}.
		\end{aligned}
	\end{equation}
	Hence, $W^f$ defined by
	\begin{equation*}
		d W^f_s = d \widetilde W^f_s + \rho_{Sf} \sigma_S ds
	\end{equation*}
	is a $\Prob$-Brownian motion.

	\section{FX Derivatives - European Call} \label{3F:benchmark_european_call} \label{3F:section:closed_form_EU_price}
	The payoff at maturity $t$ of a European Call on $FX$ rate is given by
	\begin{equation*}
		( S_t - K )_+
	\end{equation*}
	with $K$ the strike and $S_t$ the $FX$ rate at time $t$.

	Our aim will be to evaluate $V_{0}$
	\begin{equation*}
		V_{0} = \E \Big[ \e^{ - \int_0^t r_s^d ds } ( S_t - K )_+ \Big].
	\end{equation*}

	\begin{proposition}
		If we consider a 3-factor model on Foreign Exchange and Zero-coupon as defined in \eqref{3F:3FModelFX}, $V_{0}$ is given by\footnote{We ignore the settlements details in the present paper in order to alleviate the notations but the formula can easily be extended to take them into account.}
		\begin{equation*}
			V_{0}
			= S_{0} P^f(0,t) \N \Bigg( \frac{\log \Big( \frac{S_{0} P^f(0,t)}{K P^d(0,t)} \Big) + \mu(0,t) }{ \sigma(0,t) } \Bigg) - K P^d(0,t) \N \Bigg( \frac{\log \Big( \frac{S_{0} P^f(0,t)}{K P^d(0,t)} \Big) - \mu(0,t) }{\sigma(0,t) } \Bigg)
		\end{equation*}
		with
		\begin{equation*}
			\begin{aligned}
				\mu(0,t)
				 & = \int_0^t \frac{1}{2} \big(\sigma^2_S(s) + \sigma_f^2 (s,t) + \sigma_d^2 (s,t)\big) ds                                                               \\
				 & \qquad + \int_0^t \big(\rho_{Sf} \sigma_S(s) \sigma_f (s,t) - \rho_{Sd} \sigma_S(s) \sigma_d(s,t) ds - \rho_{fd} \sigma_f (s,t) \sigma_d(s,t) \big)ds \\
			\end{aligned}
		\end{equation*}
		and
		\begin{equation*}
			\sigma^2(0,t) = 2 \mu(0,t).
		\end{equation*}
	\end{proposition}

	\begin{proof}

		In this part, we want to evaluate
		\begin{equation*}
			V_{0} = \E \Big[ \e^{ - \int_0^t r_s^d ds } ( S_t - K )_+ \Big].
		\end{equation*}
		If we consider a 3-factor model on Foreign Exchange and Zero-coupon as defined in \eqref{3F:3FModelFX}, we have
		\begin{equation*}
			\begin{aligned}
				V_0
				 & = \E \Big[ \e^{ - \int_0^t r_s^d ds } ( S_t - K )_+ \Big]                                                                                     \\
				 & = \E \Big[ \big( \e^{ - \int_0^t r_s^d ds } S_t - \e^{ - \int_0^t r_s^d ds } K \big)_+ \Big]                                                  \\
				 & = \E \Big[ \big( \e^{ - \int_0^t r_s^d ds } S_t - \e^{ - \int_0^t r_s^d ds } K \big) \1_{ \{ S_t \geq K \} } \Big]                            \\
				 & = \E \Big[ \e^{ - \int_0^t r_s^d ds } S_t \1_{ \{ S_t \geq K \} } \Big] - K \E \Big[ \e^{ - \int_0^t r_s^d ds }\1_{ \{ S_t \geq K \} } \Big]. \\
			\end{aligned}
		\end{equation*}
		We focus on the first term
		\begin{equation}\label{3F:firstterm}
			K \E \Big[ \e^{ - \int_0^t r_s^d ds }\1_{ \{ S_t \geq K \} } \Big].
		\end{equation}
		We do the following change of numéraire:
		\begin{equation*}
			\frac{d \widetilde \Q }{ d \Prob} = \frac{\widetilde Z_{t}}{\widetilde Z_{0}}
		\end{equation*}
		with
		\begin{equation*} \left\{
			\begin{aligned}
				\widetilde Z_{t} & = \exp \Big( \widetilde Y_{t} - \frac{1}{2} < \widetilde Y, \widetilde Y>_{t} \Big), \\
				\widetilde Z_{0} & = 1                                                                                  \\
			\end{aligned}\right.
		\end{equation*}
		where $\widetilde Y_{t} = \int_0^t \sigma_d (s,t) dW_s^d$ and $< \widetilde Y, \widetilde Y>_{t} = \int_0^t \sigma_d^2 (s,t)ds$.

		Hence, we can define the following Brownian Motions $\widetilde{W}^d$, $\widetilde{W}^f$, $\widetilde{W}^S$ under $\widetilde \Q$:
		\begin{equation*}
			\begin{aligned}
				d \widetilde{W}^d_s & = d W^d_s - d< Y,W^d >_s & = & ~ d W^d_s - \sigma_d(s,t) ds,           \\
				d \widetilde{W}^f_s & = d W^f_s - d< Y,W^f >_s & = & ~ d W^f_s - \rho_{fd} \sigma_d(s,t) ds, \\
				d \widetilde{W}^S_s & = d W^S_s - d< Y,W^S >_s & = & ~ d W^S_s - \rho_{Sd} \sigma_d(s,t) ds
			\end{aligned}
		\end{equation*}
		and $S_t$ becomes
		\begin{equation*}
			\begin{aligned}
				S_t
				 & = S_{0} \exp \bigg( \int_0^t \bigg( r^d_s - r^f_s - \frac{\sigma^2_S(s)}{2} \bigg) ds + \int_0^t \sigma_S(s) dW^S_s \bigg)                                                                        \\
				 & = \frac{S_{0} P^f(0,t)}{P^d(0,t)} \exp \bigg( \int_0^t - \frac{1}{2} \big(\sigma^2_S(s) + \sigma_f^2 (s,t) - \sigma_d^2 (s,t)\big) - \rho_{Sf} \sigma_S(s) \sigma_f (s,t) ~ ds \bigg)             \\
				 & \qquad \times \exp \bigg( \int_0^t \sigma_S(s) dW^S_s + \int_0^t \sigma_f (s,t) dW_s^f - \int_0^t \sigma_d (s,t) dW_s^d \bigg)                                                                    \\
				 & = \frac{S_{0} P^f(0,t)}{P^d(0,t)} \exp \bigg( - \int_0^t \frac{1}{2} \big(\sigma^2_S(s) + \sigma_f^2 (s,t) + \sigma_d^2 (s,t)\big) ~ ds \bigg)                                                    \\
				 & \qquad \times \exp \bigg( - \int_0^t \big(\rho_{Sf} \sigma_S(s) \sigma_f (s,t) - \rho_{Sd} \sigma_S(s) \sigma_d(s,t) - \rho_{fd} \sigma_f (s,t) \sigma_d(s,t) \big) ds \bigg)                     \\
				 & \qquad \qquad \times \exp \bigg( \int_0^t \sigma_S(s) d\widetilde W^S_s + \int_0^t \sigma_f (s,t) d\widetilde W_s^f - \int_0^t \sigma_d (s,t) d\widetilde W_s^d \bigg)                            \\
				 & = \frac{S_{0} P^f(0,t)}{P^d(0,t)} \exp \bigg( - \mu(0,t) + \int_0^t \sigma_S(s) d\widetilde W^S_s + \int_0^t \sigma_f (s,t) d\widetilde W_s^f - \int_0^t \sigma_d (s,t) d\widetilde W_s^d \bigg).
			\end{aligned}
		\end{equation*}

		Hence, as $\exp \Big( - \int_0^t r_s^d ds \Big) = P^d(0,t) \times \widetilde Z_{t}$, \eqref{3F:firstterm} becomes
		\begin{equation*}
			\begin{aligned}
				K \E \Big[ \e^{ - \int_0^t r_s^d ds }\1_{ \{ S_t \geq K \} } \Big]
				 & = K P^d(0,t) \E^{\widetilde \Q} \Big[ \1_{ \{ S_t \geq K \} } \Big]                                                                 \\
				 & = K P^d(0,t) \widetilde \Q ( S_t \geq K )                                                                                           \\
				 & = K P^d(0,t) \widetilde \Q \Bigg( Z \geq \frac{\log \Big( \frac{K P^d(0,t)}{S_{0} P^f(0,t)} \Big) + \mu(0,t) }{\sigma(0,t) } \Bigg) \\
				 & = K P^d(0,t) \widetilde \Q \Bigg( Z \leq \frac{\log \Big( \frac{S_{0} P^f(0,t)}{K P^d(0,t)} \Big) - \mu(0,t) }{\sigma(0,t) } \Bigg) \\
				 & = K P^d(0,t) \N \Bigg( \frac{\log \Big( \frac{S_{0} P^f(0,t)}{K P^d(0,t)} \Big) - \mu(0,t) }{\sigma(0,t) } \Bigg)                   \\
			\end{aligned}
		\end{equation*}
		where $Z \sim \N (0,1)$ with
		\begin{equation*}
			\begin{aligned}
				\mu(0,t)
				 & = \int_0^t \frac{1}{2} \big( \sigma^2_S(s) + \sigma_f^2 (s,t) + \sigma_d^2 (s,t) \big) ds                                                                                                                               \\
				 & \qquad + \int_0^t \big( \rho_{Sf} \sigma_S(s) \sigma_f (s,t) - \rho_{Sd} \sigma_S(s) \sigma_d(s,t) ds - \rho_{fd} \sigma_f (s,t) \sigma_d(s,t) \big) ds,                                                                \\
				\sigma^2(0,t)
				 & = \V \bigg( \int_0^t \sigma_S(s) d\widetilde W^S_s + \int_0^t \sigma_f (s,t) d\widetilde W_s^f - \int_0^t \sigma_d (s,t) d\widetilde W_s^d \bigg)                                                                       \\
				 & = \V \bigg( \int_0^t \sigma_S(s) d\widetilde W^S_s \bigg) + \V \bigg( \int_0^t \sigma_f (s,t) d\widetilde W_s^f \bigg) + \V \bigg( \int_0^t \sigma_d (s,t) d\widetilde W_s^d \bigg)                                     \\
				 & \qquad + 2 \Cov \bigg( \int_0^t \sigma_S(s) d\widetilde W^S_s, \int_0^t \sigma_f (s,t) d\widetilde W_s^f\bigg) - 2 \Cov \bigg( \int_0^t \sigma_S(s) d\widetilde W^S_s, \int_0^t \sigma_d (s,t) d\widetilde W_s^d \bigg) \\
				 & \qquad \qquad - 2 \Cov \bigg(\int_0^t \sigma_f (s,t) d\widetilde W_s^f, \int_0^t \sigma_d (s,t) d\widetilde W_s^d \bigg)                                                                                                \\
				 & = \int_0^t \big( \sigma_S^2(s) + \sigma_f^2 (s,t) + \sigma_d^2(s,t) \big) ds                                                                                                                                            \\
				 & \qquad + 2 \int_0^t \big(\rho_{Sf} \sigma_S(s) \sigma_f (s,t) - \rho_{Sd} \sigma_S(s) \sigma_d (s,t) - \rho_{fd} \sigma_f (s,t) \sigma_d (s,t) \big) ds.
			\end{aligned}
		\end{equation*}

		Now, we deal with the term
		\begin{equation}\label{3F:secondterm}
			\E \Big[ \e^{ - \int_0^t r_s^d ds } S_t \1_{ \{ S_t \geq K \} } \Big] = P^d(0,t) \E^{\widetilde \Q} \big[ S_t \1_{ \{ S_t \geq K \} } \big]
		\end{equation}
		using directly the formula of the first partial moment of a log-normal random variable. Let $X \sim \textrm{Log-}\N(\mu, \sigma^2)$, then
		\begin{equation*}
			\begin{aligned}
				\E \big[ X \1_{ \{ X \geq x \} } \big] = \e^{\mu + \frac{\sigma^2}{2} } \N \bigg( \frac{ \mu + \sigma^2 - \log(x)}{ \sigma } \bigg).
			\end{aligned}
		\end{equation*}
		Finally, as $S_t = \frac{S_{0} P^f(0,t)}{P^d(0,t)} X$ with $X \overset{\widetilde \Q}{\sim} \textrm{Log-}\N(-\mu(0,t), \sigma^2(0,t))$, we get

		\begin{equation*}
			\begin{aligned}
				\eqref{3F:secondterm}
				 & = S_{0} P^f(0,t) \E^{\widetilde \Q} \Bigg[ X \1_{ \Big\{ X \geq \frac{K P^d(0,t)}{S_{0} P^f(0,t)} \Big\} } \Bigg]                                                                   \\
				 & = S_{0} P^f(0,t) \e^{- \mu(0,t) + \frac{\sigma^2(0,t)}{2} } \N \Bigg( \frac{ - \mu(0,t) + \sigma^2(0,t) - \log\Big( \frac{K P^d(0,t)}{S_{0} P^f(0,t)} \Big) }{ \sigma(0,t) } \Bigg) \\
				 & = S_{0} P^f(0,t) \N \Bigg( \frac{ \log\Big(\frac{S_{0} P^f(0,t) }{ K P^d(0,t) }\Big) + \mu(0,t) }{ \sigma(0,t) } \Bigg)                                                             \\
			\end{aligned}
		\end{equation*}
		noticing that $\mu(0,t) = \frac{\sigma^2(0,t)}{2}$.

		Finally, we get
		\begin{equation*}
			\begin{aligned}
				V_{0}
				 & = \E \Big[ \e^{ - \int_0^t r_s^d ds } ( S_t - K )_+ \Big]                                                                                                                                                                                 \\
				 & = \E \Big[ \e^{ - \int_0^t r_s^d ds }S_t \1_{ \{ S_t \geq K \} } \Big] - K \E \Big[ \e^{ - \int_0^t r_s^d ds }\1_{ \{ S_t \geq K \} } \Big]                                                                                               \\
				 & = S_{0} P^f(0,t) \N \Bigg( \frac{\log \Big( \frac{S_{0} P^f(0,t)}{K P^d(0,t)} \Big) + \mu(0,t) }{ \sigma(0,t) } \Bigg) - K P^d(0,t) \N \Bigg( \frac{\log \Big( \frac{S_{0} P^f(0,t)}{K P^d(0,t)} \Big) - \mu(0,t) }{\sigma(0,t) } \Bigg). \\
			\end{aligned}
		\end{equation*}

		Special case of constant volatility: $\sigma_S(s) = \sigma_S$, $ \sigma_d (s,t) = \sigma_d \times (t-s) $ $ \sigma_f (s,t) = \sigma_f \times (t-s)$

		\begin{equation*}
			\begin{aligned}
				\mu(0,t)
				 & = \int_0^t \frac{1}{2} \big( \sigma^2_S(s) + \sigma_f^2 (s,t) + \sigma_d^2 (s,t) \big) ds                                                                                                                                           \\
				 & \qquad + \int_0^t \big( \rho_{Sf} \sigma_S(s) \sigma_f (s,t) - \rho_{Sd} \sigma_S(s) \sigma_d(s,t) - \rho_{fd} \sigma_f (s,t) \sigma_d(s,t) \big) ds                                                                                \\
				 & = \int_0^t \frac{1}{2} \big( \sigma_S^2 + \sigma_f^2 (t-s)^2 + \sigma_d^2 (t-s)^2 \big) ds                                                                                                                                          \\
				 & \qquad + \int_0^t \rho_{Sf} \sigma_S \sigma_f (t-s) - \rho_{Sd} \sigma_S \sigma_d (t-s) - \rho_{fd} \sigma_f \sigma_d (t-s)^2 ds                                                                                                    \\
				 & = \frac{1}{2} \bigg( \sigma_S^2 t + \sigma_f^2 \frac{t^3}{3} + \sigma_d^2 \frac{t^3}{3} \bigg) + \rho_{Sf} \sigma_S \sigma_f \frac{t^2}{2} - \rho_{Sd} \sigma_S \sigma_d \frac{t^2}{2} - \rho_{fd} \sigma_f \sigma_d \frac{t^3}{3}, \\
				\sigma^2(0,t)
				 & = \int_0^t \big( \sigma_S^2(s) + \sigma_f^2 (s,t) + \sigma_d^2 (s,t) \big) ds                                                                                                                                                       \\
				 & \qquad + 2 \int_0^t \big(\rho_{Sf} \sigma_S(s) \sigma_f (s,t) - \rho_{Sd} \sigma_S(s) \sigma_d (s,t) - \rho_{fd} \sigma_f (s,t) \sigma_d (s,t) \big) ds                                                                             \\
				 & = 2 \mu(0, t).                                                                                                                                                                                                                      \\
			\end{aligned}
		\end{equation*}
	\end{proof}

\end{appendices}

\end{document}